\documentclass[12pt]{article}

 \usepackage{graphicx} \usepackage{float} \usepackage{verbatim} \usepackage{url}  \usepackage{psfrag} \usepackage{multicol} \usepackage{multirow} \usepackage{bigstrut} \usepackage{enumerate} \usepackage{booktabs}  \usepackage{amsfonts} \usepackage{rotating} \usepackage{ragged2e} \usepackage{xcolor}

\usepackage{subcaption}
\usepackage{setspace}
\usepackage{fullpage}
\usepackage[utf8]{inputenc}
\usepackage[USenglish]{babel}
\usepackage{amssymb}
\usepackage{amsthm}
\usepackage{amsmath}
\usepackage{ulem}
\usepackage{hyperref}
\usepackage{xcolor}
\usepackage{booktabs}
\usepackage{enumitem}
\hypersetup{
   colorlinks=true,linkcolor=blue,citecolor=blue, filecolor=blue, citebordercolor=blue,
   linkbordercolor = blue
}
\usepackage{titlesec}
\titleformat*{\subsection}{\bfseries}

\usepackage{enumitem}
\usepackage{caption}
\usepackage{subcaption}
\usepackage{mwe}

\newtheorem{proposition}{Proposition}

\usepackage[authoryear]{natbib}
\renewcommand{\cite}{\citet}
\newtheorem{corollary}{Corollary}[section]

\newtheorem{definition}{Definition}[section]
\makeindex

\begin{document} 
\onehalfspacing
\renewcommand{\tablename}{{{Table}}}

\title{\vspace{-5ex}\LARGE{{Reviving Micro Real Rigidities: The Importance of Demand Shocks\thanks{The views expressed here are those of the authors and not necessarily those of the Federal Reserve Bank of Atlanta or the Federal Reserve System. The authors are grateful for the invaluable feedback they received on an earlier draft from Jón Steinsson. The authors would like to thank Nir Jaimovich (editor) and four anonymous referees, as well as Isaac Baley, Mark Bils, Laura Castillo-Mart\'{i}nez (discussant), Basile Grassi, Susanto Basu, John Haltiwanger, Pete Klenow, Francesco Lippi, Virgiliu Midrigan, Tobias Renkin, Luminita Stevens, Javier Tur\'{e}n, Rosen Valchev, and seminar participants at the NBER 2024 SI (ME), UC Berkeley, Boston College, the University of Maryland, the Federal Reserve Board of Governors, the Federal Reserve Banks of Atlanta, Boston, Chicago, and Richmond, ITAM, IMF, Banque de France, BIS, Bocconi, EIEF, CREI, the University of Iowa, Georgetown University, PUC-Chile, Michigan State University, the European Central Bank, the Deutsche Bundesbank, the European University Institute, North Carolina State University, City University of Hong Kong, and Shanghai University of Finance and Economics for useful comments and suggestions. Any remaining errors are the authors’ responsibility. }}}}
\author{S. Bora\u{g}an Aruoba\thanks{University of Maryland and NBER} \and Eugene Oue\thanks{Shandong University} \and Felipe Saffie\thanks{University of Virginia, Darden School of Business and NBER} \and Jonathan L. Willis\thanks{Federal Reserve Bank of Atlanta} }
\date{\today}

\maketitle

\vspace{-8ex}

\begin{center}\end{center}
\vspace{-4ex}
\begin{abstract}

We revisit micro real rigidities as a source of monetary non-neutrality in a menu-cost model with variable markups, using firm-level evidence to pin down key primitives. We embed a non-CES demand system in a quantitative monetary model and use firm-dynamics evidence to identify demand curvature and firm-level productivity and demand processes. The calibrated model matches untargeted micro pricing moments, the markup distribution, and cost pass-through, while generating comparable non-neutrality. The key innovation is an empirically supported placement of idiosyncratic demand shocks that shifts residual demand, thereby moving desired markups and prices under non-CES demand.  The broader implication is that this calibrated model provides a portable framework linking monetary economics with trade and IO evidence.


\end{abstract}
\begin{center}\end{center}
\vspace{-4ex}

\vspace{1ex} 

\vspace{1ex}

\section{Introduction}

Modeling the response of output and prices to monetary policy shocks has long been central to macroeconomic research. Empirical studies consistently show that such shocks have persistent and substantial effects on real output \citep{christiano1999monetary, ramey2016macroeconomic}. These real effects arise because prices do not fully adjust to nominal disturbances. How much they adjust depends on two distinct forces: nominal frictions that prevent firms from changing prices, and structural features of demand and costs that limit how much firms want to adjust prices even when they can. \cite{ball1990real} show that nominal frictions alone are insufficient, and that a second force, which they termed real rigidities, is necessary to generate quantitatively significant monetary non-neutrality.\footnote{The term real rigidities is standard in the macroeconomic literature following \cite{ball1990real}. Closely related mechanisms, however, are widely studied in the industrial organization and international trade literatures under different terminologies, including incomplete cost pass-through, variable markups, demand curvature, and strategic complementarities in pricing. Thus, what macroeconomists refer to as real rigidities largely corresponds to features of demand systems, competitive structure, and markup adjustment that are central objects of study in these other fields.}

A key determinant of how much firms want to adjust prices is the shape of the demand they face. Under constant elasticity of substitution (CES) demand, the standard assumption in quantitative pricing models, desired markups are constant and firms pass cost changes to prices one for one. Deviations from CES, where demand becomes more or less price sensitive depending on a firm's market share, generate variable markups and incomplete pass-through of cost shocks, features that are quantitatively important across many fields of economics. Non-CES demand systems such as \cite{kimball1995quantitative} are widely used to study markup variability \citep{edmond2018costly, arkolakis2019elusive}, exchange rate pass-through \citep{gopinath2010frequency,gopinath2011search, amiti2019international, berger2019shocks}, and inflation dynamics \citep{smets2007shocks, harding2022resolving, harding2023understanding, fujiwara2022competition}. Yet in the quantitative monetary economics literature, their adoption has been limited.

Instead, the quantitative monetary economics literature has primarily relied on macro-level real rigidities to amplify the real effects of nominal shocks, including production networks with sticky intermediate inputs \citep{basu1995intermediate,nakamura2010monetary}, segmented labor markets \citep{woodford2003,gertler2008phillips}, and real wage rigidity \citep{blanchard2007real}. These mechanisms are effective at generating monetary non-neutrality, but they operate at the aggregate level and bypass the firm-level demand structure that governs pricing, markups, and pass-through. As a result, a natural connection between the monetary pricing literature and a large body of work in industrial organization, trade, and firm dynamics, where the shape of demand and its interaction with firm-level shocks are central objects of study, has remained largely detached.

This paper bridges this gap by bringing an empirically disciplined non-CES demand system into a quantitative monetary economics model. We make three contributions. First, we show theoretically how variable markups interact with idiosyncratic demand shocks to shape firms' desired prices, and we identify firm-level moments that discipline the key demand parameters. Second, we build a menu-cost model calibrated to firm-level evidence on productivity and demand processes from \cite{foster2008reallocation}, leaving standard pricing moments untargeted, and show that the model matches them well alongside the markup distribution, cost pass-through rates, and firm growth dynamics. Third, we show that the estimated demand system generates monetary non-neutrality comparable in magnitude to the macro-level mechanisms that the literature has pursued. The key innovation is to discipline and incorporate idiosyncratic demand shocks, which under non-CES demand move desired markups and prices.

The theoretical contribution, developed in Section \ref{sec:simple}, can be summarized as follows. Under CES demand, the elasticity a firm faces is constant, so the markup over marginal cost is constant, and only shocks to marginal cost, such as productivity shocks, pass through to prices. In this case, prices and productivity are perfectly negatively correlated, leading to a full pass-through of cost shocks to prices. However, a large body of empirical work in trade, industrial organization, and macroeconomics documents that the pass-through of cost shocks to prices is well below one, typically in the range of 20 to 50 percent. This is evidence that demand elasticities, or equivalently markups, are not constant, which is precisely what non-CES demand systems deliver. Under the standard parametrization of \cite{kimball1995quantitative}, for instance, a firm with a larger effective market share faces less price-sensitive demand, allowing it to charge a higher markup. 

This variable-elasticity structure has two implications. First, it generates incomplete pass-through of cost shocks, as changes in output move the firm along its demand curve and partially offset the incentive to adjust prices. Second, it opens the door for demand shocks to affect desired prices: a positive demand shock that raises a firm's effective market share shifts the residual demand conditions it faces and raises its desired markup and its desired price. The price-productivity correlation estimated at $-0.54$ by \cite{foster2008reallocation}, far from the $-1$ implied by CES, is direct evidence that demand shocks move desired markups in the data. Moreover, the distance from $-1$ and the estimate of the demand elasticity from the same study provide a strategy to identify the shape of the demand system.

We implement this identification strategy in a menu-cost model with a \cite{kimball1995quantitative} demand system, calibrated to match firm-level moments from \cite{foster2008reallocation}. This study is among the few that separately identify productivity and demand processes at the firm level using data on prices, quantities, and inputs for U.S. manufacturing firms. We target the persistence and cross-sectional dispersion of firm-level productivity and demand, the price-productivity correlation, an estimate of the demand elasticity, and the frequency of price changes. Importantly, we do not target any other pricing moments. Despite this, the calibrated model replicates the average size and direction of price changes, the dispersion of non-zero price adjustments, and a mildly downward-sloping hazard of price adjustment, all consistent with U.S. Consumer Price Index (CPI) microdata. The model also generates a cross-sectional distribution of markups that closely resembles estimates from Compustat data, and a cost pass-through rate of 43 percent, within the range documented in the empirical literature.\footnote{\cite{aruoba2025pricing} underscore the flexibility of the Kimball framework using Chilean microdata, incorporating additional features such as leptokurtic shocks and news shocks. While distinct from our approach, their findings further highlight the value of this framework for explaining pricing behavior.} These results show that a model disciplined by firm dynamics evidence can account for pricing patterns that the literature has traditionally targeted directly.

We assess the robustness of these results along several dimensions. We recalibrate the model using firm-level estimates from Colombian manufacturing data from \cite{eslava2024size}, who employ an alternative methodology that relaxes the orthogonality assumption between productivity and demand shocks. The qualitative patterns are similar. We also show that the results are stable when we allow correlated shocks, target a broader set of moments, or introduce leptokurtic demand shocks to match the kurtosis of price changes.

Turning to monetary non-neutrality, the calibrated model generates cumulative output responses to a nominal expenditure shock approximately 34\% larger than that produced by a comparable CES model. Two forces drive this result: the non-CES demand system, which induces strategic complementarity in pricing, and the inclusion of both productivity and demand shocks, which weakens selection with respect to the aggregate shock by increasing the share of adjusters responding to idiosyncratic conditions rather than the aggregate disturbance. The first force accounts for the larger share of the amplification. Notably, our micro-calibrated demand system more than doubles monetary non-neutrality even in a \cite{calvo1983staggered} model, where selection is absent by construction. The degree of non-neutrality generated by our model is comparable to that reported in frameworks relying on macro-level complementarities, such as the multi-sector production network model of \cite{nakamura2010monetary}.

\paragraph{Related literature.} Our paper contributes to the quantitative menu cost literature initiated by \cite{golosov2007menu}, who show that standard menu-cost models with CES demand generate limited monetary non-neutrality. A key reason non-CES demand was set aside in this literature is the finding by \cite{klenow2016real} that, in productivity-only calibrations with variable markups, matching the distribution of price changes requires very volatile idiosyncratic driving forces, with counterfactual implications for observable firm-level outcomes. We show that this tension is conditional on two features of their analysis: the absence of idiosyncratic demand shocks, and the treatment of demand curvature as given rather than estimated from firm-level data. In Section \ref{sec:ext_valid_robust}, we restate the analysis in terms of observable moments and show that, once curvature is disciplined and demand shocks are included, the tension disappears. 

Our work is complementary to \cite{mongey2021market}, who shows that strategic interactions in a dynamic oligopoly model can also overturn the \cite{klenow2016real} finding. However, here our approach is different: rather than proposing a different market structure, we aim to make the benchmark framework of monopolistic competition with non-CES demand, one commonly used elsewhere in economics, empirically disciplined and quantitatively usable in the monetary economics literature. In fact, \cite{wang2022dynamic} provide a bridge between these approaches by showing that, for standard monetary experiments, an oligopolistic model can be well approximated by a non-strategic benchmark with modified Kimball preferences, with most quantitative effects operating through residual-demand feedback rather than strategic interaction. This supports the use of a disciplined Kimball demand system as a parsimonious and portable reduced form across applications. \cite{beck2020price} provide micro-based estimates of demand elasticities and super-elasticities using European data, finding moderate curvature that on its own generates limited non-neutrality, which underscores the importance of incorporating demand shocks as we do. More broadly, our identification strategy is part of a growing effort to discipline pricing primitives using firm-level microdata rather than aggregate moments, as in \cite{gagliardone2025anatomy}.

Firms may also find relative price adjustments costly if they face increasing marginal costs. When marginal costs rise with output, per-unit profits fall as firms expand, restraining price cuts and generating strategic complementarity through the cost side. Cost-side mechanisms include decreasing returns to scale \citep{burstein2007prices} and segmented input markets \citep{woodford2003, gertler2008phillips}. We explore this alternative in Section \ref{sec:ext_valid_robust} and find that it can match several moments of the data, but requires very low returns to scale and does not replicate the empirical markup distribution as well as the demand-side specification.

The paper is organized as follows. Section \ref{sec:simple} develops a simple theoretical framework that formalizes the role of demand curvature and idiosyncratic demand shocks in price setting, and establishes the identification results summarized above. Section \ref{sec:quantitative} introduces the quantitative menu-cost model. Section \ref{sec:calibration} describes the calibration strategy and reports the fit to targeted and untargeted moments. Section \ref{sec:ext_valid_robust} discusses external validity and presents robustness exercises, including the Klenow--Willis variants and the cost-side alternative. Section \ref{sec:nonneutrality} examines the implications for monetary non-neutrality. Section \ref{sec:conclusion} concludes.

\section{Real Rigidities and Demand Shocks} \label{sec:simple}

This section develops a simple framework to isolate real rigidities -- forces that dampen desired price changes in response to changes in aggregate real demand.  To do so, this framework abstracts from any nominal pricing frictions that may prevent firms from  responding to shocks. The exposition follows \cite{ball1990real} and \cite{nakamura2008five} and adapts their framework to our setting.

We proceed in three steps. First, we formalize macro and micro real rigidities in a frictionless price-setting problem and connect micro rigidities to primitives through demand and cost curvature. Second, we introduce deviations from CES as a source of micro real rigidities. Third, we show how the placement of idiosyncratic demand shocks affects desired markups and the correlation between desired prices and idiosyncratic productivity.

\subsection{Micro and Macro Real Rigidities} \label{sec:simple_micro_macro}

Consider a static frictionless price-setting problem. The firm chooses a desired relative price to maximize profit:
\begin{equation} \label{eq:Pi}
\max_{p_i} \ \Pi \left(\frac{p_i}{P},\frac{S}{P},A_i\right),
\end{equation}
where ${p_i}/{P}$ is the firm's relative price, ${S}/{P} = Y$ is real aggregate demand ($S$ represents nominal aggregate demand), and $A_i$ collects idiosyncratic primitives. Note that the profit function includes only the technological cost of production and abstracts from nominal rigidities. The firm's profit-maximizing relative price ${p_i^*}/{P}$ satisfies the first-order condition
\begin{equation}
\Pi_1\left(\frac{p_i^*}{P},\frac{S}{P},A_i\right) = 0,
\end{equation}
where $\Pi_1$ denotes the partial derivative with respect to the firm's relative price.

Real rigidities refer to mechanisms that attenuate the responsiveness of firms’ desired prices to changes in aggregate shocks. To formalize this, define the responsiveness of the firm’s desired relative price with respect to real aggregate demand as
\begin{equation}\label{eq:RR}
\phi \equiv \frac{\partial \left(\dfrac{p_i^*}{P}\right)}{\partial \left(\dfrac{S}{P} \right)}
= - \frac{\Pi_{12}\left(\dfrac{p_i^*}{P},\dfrac{S}{P},A_i\right)}{\Pi_{11}\left(\dfrac{p_i^*}{P},\dfrac{S}{P},A_i\right)}.
\end{equation}
The second-order condition for profit maximization \(\Pi_{11}<0\), together with the local stability condition \(\Pi_{12}>0\), implies \(\phi>0\).\footnote{\label{fn:stability} To determine the sign of \(\Pi_{12}\), consider an exogenous increase in real aggregate demand \(dY>0\) with \(Y\equiv S/P\) (for example, a nominal spending shock \(dS>0\) at given \(P\)). If \(\Pi_{12}>0\), then \(\partial (p_i^*/P)/\partial Y>0\) at the firm level; in equilibrium the aggregate price level \(P\) rises and partially offsets the increase in \(Y\) (stabilizes \(Y\)). If \(\Pi_{12}=0\), desired relative prices are invariant to \(Y\), so \(P\) is unchanged and the shock passes one-for-one to \(Y\). If \(\Pi_{12}<0\), desired relative prices fall, \(P\) declines, and the disturbance is amplified, resulting in output instability.}
A smaller \(\phi\) corresponds to stronger real rigidity, as desired prices respond less to changes in real aggregate demand.

Real rigidities are closely linked to strategic interactions in pricing. Let 
\[\zeta \equiv \left.\dfrac{\partial \ln p_i^*}{\partial \ln P}\right|_{S,\,A_i}\]
denote the elasticity of a firm’s desired price with respect to the aggregate price level, holding nominal spending \(S\) and idiosyncratic primitives fixed. In a symmetric equilibrium with \(p_i^*=P\) and \(Y\equiv S/P\), the elasticity of the desired relative price with respect to real aggregate demand satisfies
\[
\left.\frac{\partial \ln \left(\dfrac{p_i^*}{P}\right )}{\partial \ln Y}\right|_{S,\,A_i}=1-\zeta.
\]
Thus, stronger real rigidity in desired relative prices corresponds to stronger strategic complementarity (larger \(\zeta\)); Appendix \ref{app:rigidity_sc} provides the derivation.

Following \cite{ball1990real}, it is useful to distinguish macro and micro sources of real rigidity. The specification in \eqref{eq:RR} provides a convenient way of doing so: stronger real rigidities (small $\phi$) can arise from either small $\Pi_{12}$ or large $|\Pi_{11}|$. Macro real rigidities are those that make $\Pi_{12}$ small: as $Y$ changes, the responsiveness of profits to the firm's choice of its own price does not change by much.\footnote{\cite{nakamura2010monetary} label the associated strategic complementarities as \(\Omega\)-type for macro sources and \(\omega\)-type for micro sources. Real wage rigidity and sticky input prices are two popular macro real rigidities. In both cases, when aggregate demand rises, marginal cost movements are muted and common across firms, attenuating the response of desired relative prices to changes in aggregate demand.} 

Micro real rigidities stem from curvature in the profit function as captured by $|\Pi_{11}|$. A high degree of curvature implies that the profit-maximizing relative price responds only modestly to shocks, as is immediate from \eqref{eq:RR}. The profit function can have high curvature, a large $|\Pi_{11}|$, either due to a more concave demand schedule or a more convex cost schedule. To connect $|\Pi_{11}|$ to primitives, consider a standard environment underlying \(\Pi(\cdot)\) in which profits take the form \((p_i-mc_i(y_i))\,y_i\), where \(y_i\) is residual demand and \(mc_i(\cdot)\) is marginal cost. More concave residual demand makes mispricing increasingly costly: moving away from \(p_i^*\) triggers a disproportionately large loss in quantity and revenue, raising \(|\Pi_{11}|\). More convex costs also raise \(|\Pi_{11}|\): price cuts that expand sales push up marginal cost and compress per-unit profits, dampening the gain from adjusting the relative price.\footnote{Cost-side mechanisms generating such rigidity include decreasing returns to scale in production \citep{burstein2007prices} and segmented input markets that yield upward-sloping cost curves at the firm level \citep{woodford2003,gertler2008phillips}.} We focus on demand curvature in the main text and present the parallel cost-side results in Appendix~\ref{app:cost_side_theory}.\footnote{
Differentiating the first-order condition with respect to an idiosyncratic primitive shows that the responsiveness of desired prices to idiosyncratic shocks is decreasing in $|\Pi_{11}|$. In environments with stronger micro real rigidities, achieving a given magnitude of price adjustment therefore requires larger disturbances to idiosyncratic driving forces, which is the basis of the critique in \cite{klenow2016real}.}

\subsection{Curvature of Demand: Deviations from CES} \label{sec:simple_nonCES}

In order to generate stronger micro real rigidities via curvature of demand, we focus on deviations from the standard CES demand structure. We start by defining an operating point and local elasticity.

\begin{definition}
The operating point for variety \(i\) is 
\[\varphi_i \equiv \dfrac{y_i}{Y}.\]
\end{definition}

\begin{definition}
The local own-price (residual) elasticity at \(\varphi_i\) is
\[
\sigma(\varphi_i) \equiv -\,\left.\frac{\partial \ln\!\left(\dfrac{y_i}{Y}\right)}{\partial \ln\!\left(\dfrac{p_i}{P}\right)}\right|_{P,\,Y}.
\]
\end{definition} 

That is, \(\sigma(\varphi_i)\) is the residual own-price elasticity holding the aggregate environment \((P,Y)\) fixed.

For concreteness, we organize the discussion around a homothetic demand class that nests the \cite{kimball1995quantitative} formulation, in which local demand curvature depends on the operating point \(\varphi_i\).\footnote{We work with the HDIA (homothetic direct implicit additivity) class, which nests Kimball. Appendix \ref{app:hiia} reports parallel derivations for the homothetic indirect implicit additivity (HIIA) class, see \citet{matsuyama2023non}.} Because \(y_i\) depends on prices and shocks, \(\varphi_i = y_i/Y\) varies with both idiosyncratic and aggregate conditions. Thus any change in \(p_i/P\) or in aggregate conditions \((P \text{ or } Y)\) that alters \(y_i/Y\) shifts \(\varphi_i\) and, in turn, \(\sigma(\varphi_i)\). For brevity we write \(\sigma(\varphi_i)\), suppressing other arguments.

\begin{proposition}[Lerner Rule]\label{prop:markup}
For a price-setting firm facing downward-sloping residual demand with local elasticity \(\sigma(\varphi_i)>1\) (taking \(P\) and \(Y\) as given), the optimal price $p^*_i$ is a markup over marginal cost $mc_i$
\[
p_i^* = \frac{\sigma(\varphi_i)}{\sigma(\varphi_i)-1}\, mc_i.
\]
\end{proposition}
\begin{proof}
See Appendix \ref{app:lerner_proof}.
\end{proof}

The curvature of demand is summarized directly by how the local elasticity varies with the operating point, via \(d\ln\sigma(\varphi_i)/d\ln\varphi_i\). Using the Lerner rule in Proposition \ref{prop:markup}, and allowing \(\varphi_i\) to move with the firm’s choices, the cost pass-through in logs is
\begin{equation} \label{eq:section2_cost_pt}
    \frac{\partial \ln p_i^*}{\partial \ln mc_i}
=\frac{1}{\,1+\dfrac{d\ln\sigma(\varphi_i)}{d\ln\varphi_i}\cdot
\dfrac{1}{\sigma(\varphi_i)-1}\cdot
\dfrac{\partial \ln \varphi_i}{\partial \ln p_i^*}\,}.\footnote{See the full derivations in Appendix~\ref{app:generic_pt}. If technology exhibits constant returns so that \(\ln mc_i=-\ln z_i\), the productivity pass-through to desired prices is \(-\,\partial \ln p_i^*/\partial \ln mc_i\). }
\end{equation}
This shows that the cost pass-through to desired price depends on the interaction of demand curvature as summarized by \(d\ln\sigma(\varphi_i)/d\ln\varphi_i\), a markup wedge \(1/(\sigma(\varphi_i)-1)\), and the slope of residual demand \(\partial\ln \varphi_i/\partial\ln p_i^*\).

Under CES, \(d\ln\sigma(\varphi_i)/d\ln\varphi_i=0\), so cost pass-through to desired price is complete (\(\partial \ln p_i^*/\partial \ln mc_i=1\)). With variable elasticity and downward-sloping demand \linebreak (\(\partial \ln \varphi_i/\partial \ln p_i^*<0\)), cost pass-through becomes incomplete when \(d\ln\sigma(\varphi_i)/d\ln\varphi_i<0\). Furthermore, the more negative this term is -- that is, as consumers grow increasingly unwilling to absorb additional output along the demand schedule -- demand becomes more concave and cost pass-through falls further.

\subsection{Idiosyncratic Demand Shocks} \label{sec:simple_demand}

The previous subsection showed that when demand curvature rises and deviates further from CES, the pass-through of idiosyncratic productivity shocks to prices declines, which dampens price responsiveness. We now introduce an idiosyncratic demand shock \(\nu_i\) and study how it affects desired prices when demand is non-CES. Allowing the operating point \(\varphi_i\) to depend on \(\nu_i\) and on the firm’s own price, the demand pass-through in logs is

\begin{equation} \label{eq:section2_demand_pt}
    \frac{\partial \ln p_i^{*}}{\partial \ln \nu_i}
=
\frac{-\,\dfrac{d\ln\sigma(\varphi_i)}{d\ln\varphi_i}\cdot \dfrac{\partial \ln \varphi_i}{\partial \ln \nu_i}}{\;\sigma(\varphi_i)-1\;}
\cdot
\left[
1+\frac{d\ln\sigma(\varphi_i)}{d\ln\varphi_i}\cdot \frac{1}{\sigma(\varphi_i)-1}\cdot \frac{\partial \ln \varphi_i}{\partial \ln p_i^{*}}
\right]^{-1}.\footnote{See the full derivations in Appendix~\ref{app:generic_pt}.}
\end{equation}

A first implication is immediate: if \(\nu_i\) does not move the operating point \(\big(\partial \ln \varphi_i/\partial \ln \nu_i=0\big)\), then \(\partial \ln p_i^{*}/\partial \ln \nu_i=0\) regardless of demand curvature. We refer to this as a level-only placement of idiosyncratic demand shocks, akin to taste shifters placed outside the aggregator as in \cite{arkolakis2017variable} and \cite{wang2022dynamic}. In this case local elasticity and the desired price are locally invariant to the idiosyncratic demand shock.

Alternatively, if the demand shifter enters the operating point directly \(\big(\partial \ln \varphi_i/\partial \ln \nu_i>0\big)\), given downward-sloping non-CES demand, then demand pass-through to desired price is positive \(\partial \ln p_i^{*}/\partial \ln \nu_i>0\) and is larger when demand is more concave (\(|d\ln\sigma(\varphi_i)/d\ln\varphi_i|\) is larger). We refer to this as an effective-share placement. This placement is closely related to variable-markup frameworks in which desired markups depend on residual-demand conditions, often summarized by market share or by a firm-specific shifter to residual demand or desired markups, so that shocks that move these conditions translate into markup and price movements (e.g., \citet{amiti2019international}; \citet{gagliardone2025anatomy}). Under CES, where \(d\ln\sigma(\varphi_i)/d\ln\varphi_i=0\), demand pass-through is zero regardless of how \(\nu_i\) enters.

Beyond its role for demand pass-through, the degree of variable elasticity in demand and the placement of idiosyncratic shocks both affect the correlation between desired prices and firm productivity, which we characterize in Corollary~\ref{cor:corr_ps}. Although both placements of idiosyncratic demand shocks have been used in the literature, Corollary~\ref{cor:corr_ps} yields an empirically testable criterion to discipline the modeling of idiosyncratic demand shocks.

\begin{corollary}[Price--Productivity Correlation] \label{cor:corr_ps}
Consider a firm that faces orthogonal TFP and demand shocks \(z_i\) and \(\nu_i\). Log-linearized around a symmetric steady state, the firm's optimal pricing rule can be written as
\[
\widehat{p}_i^{\,*}=\alpha\,\widehat{z}_i+\beta\,\widehat{\nu}_i+\text{constant},
\]
with 
\[\alpha\equiv\dfrac{\partial \ln p_i^{*}}{\partial \ln z_i} \text{ and  } \beta\equiv\dfrac{\partial \ln p_i^{*}}{\partial \ln \nu_i}\].
\begin{enumerate}
    \item If the demand system is CES, then \(\beta=0\) and \(\mathrm{Corr}(\widehat{p}_i^{\,*},\widehat{\mathrm{TFPQ}}_i)=-1\) exactly.
    \item If the demand system is non-CES and \(\nu_i\) does not move the operating point, then \(\beta=0\) and \(\mathrm{Corr}(\widehat{p}_i^{\,*},\widehat{\mathrm{TFPQ}}_i)=-1\) exactly.
    \item If the demand system is non-CES and \(\nu_i\) moves the operating point so that \(\beta>0\), then \(\mathrm{Corr}(\widehat{p}_i^{\,*},\widehat{\mathrm{TFPQ}}_i)>-1\).
\end{enumerate}
\end{corollary}
\begin{proof}
    See Appendix~\ref{app:corr_ps_appendix}.
\end{proof}

\subsection{Kimball Demand System and a Parametric Example} \label{sec:simple_kimball}

As a parametric illustration consistent with the preceding analysis, we adopt a \cite{kimball1995quantitative} aggregator, which makes the local elasticity depend on a firm's effective share. Recall the operating point \(\varphi_i \equiv y_i/Y\) and define the effective share \(\tilde{\varphi}_i \equiv \nu_i\,\varphi_i\),  where \(\nu_i\) is an idiosyncratic demand primitive and \(Y \equiv S/P\). The final-good producer aggregates intermediates as
\[
1=\int_0^1 G(x_i)\,di,\qquad x_i \equiv \tilde{\varphi}_i,
\]
with \(G'(x)>0\) and \(G''(x)<0\). Under this specification the local elasticity \(\sigma(.)\) is evaluated at \(\tilde{\varphi}_i\).

Following \cite{dotsey2005implications}, we use the parametric form
\begin{equation}\label{eq:kimball_spec}
G(x)=\frac{\omega}{1+\omega\psi}\Big[(1+\psi)x-\psi\Big]^{\frac{1+\omega\psi}{\omega(1+\psi)}}+1-\frac{\omega}{1+\omega\psi},
\qquad \omega>1,\ \psi\in\mathbb{R}.
\end{equation}
Here \(\omega\) governs the elasticity level in a symmetric equilibrium and \(\psi\) governs the global curvature as summarized by the super-elasticity -- the elasticity of $\sigma$ with respect to $\tilde{\varphi}_i$. When \(\psi=0\), this specification collapses to the standard CES aggregator. We derive \(\sigma(\tilde{\varphi})\) and the super-elasticity in Appendix \ref{app:kimball_elast}. In this specification, \(\nu_i\) shifts the effective share \(\tilde{\varphi}_i\), so it changes the local elasticity and therefore moves desired markups and desired prices. By contrast, under a level-only placement in which the elasticity is evaluated at \(\varphi_i\) (rather than \(\tilde{\varphi}_i\)), \(\nu_i\) scales quantities without affecting the local elasticity and leaves desired prices unchanged. We will show in Section~\ref{sec:calibration} that the data favors the effective-share placement.

Under this parametric form, we obtain closed-form cost and demand pass-through expressions analogous to \ref{eq:section2_cost_pt} and \ref{eq:section2_demand_pt}.
\begin{corollary}\label{prop:DK_pt}
With the Kimball aggregator \eqref{eq:kimball_spec}, no pricing frictions, and constant returns to scale, cost and demand pass-through to the desired price satisfy
\[
\frac{\partial \widehat{p}_i^{\,*}}{\partial \widehat{mc}_i}=\frac{-1}{\omega\psi-1},
\qquad
\frac{\partial \widehat{p}_i^{\,*}}{\partial \widehat{\nu}_i}=\frac{\omega\psi}{\omega\psi-1},
\]
with \(\widehat{mc}_i\equiv-\widehat{z}_i\).
\end{corollary}
\begin{proof}
See Appendix \ref{app:kimball_pt}.
\end{proof}

\begin{figure}[t!]
\caption{Pass-through of Idiosyncratic Demand and Productivity to the Desired Price}
\begin{center}
\includegraphics[scale=0.42]{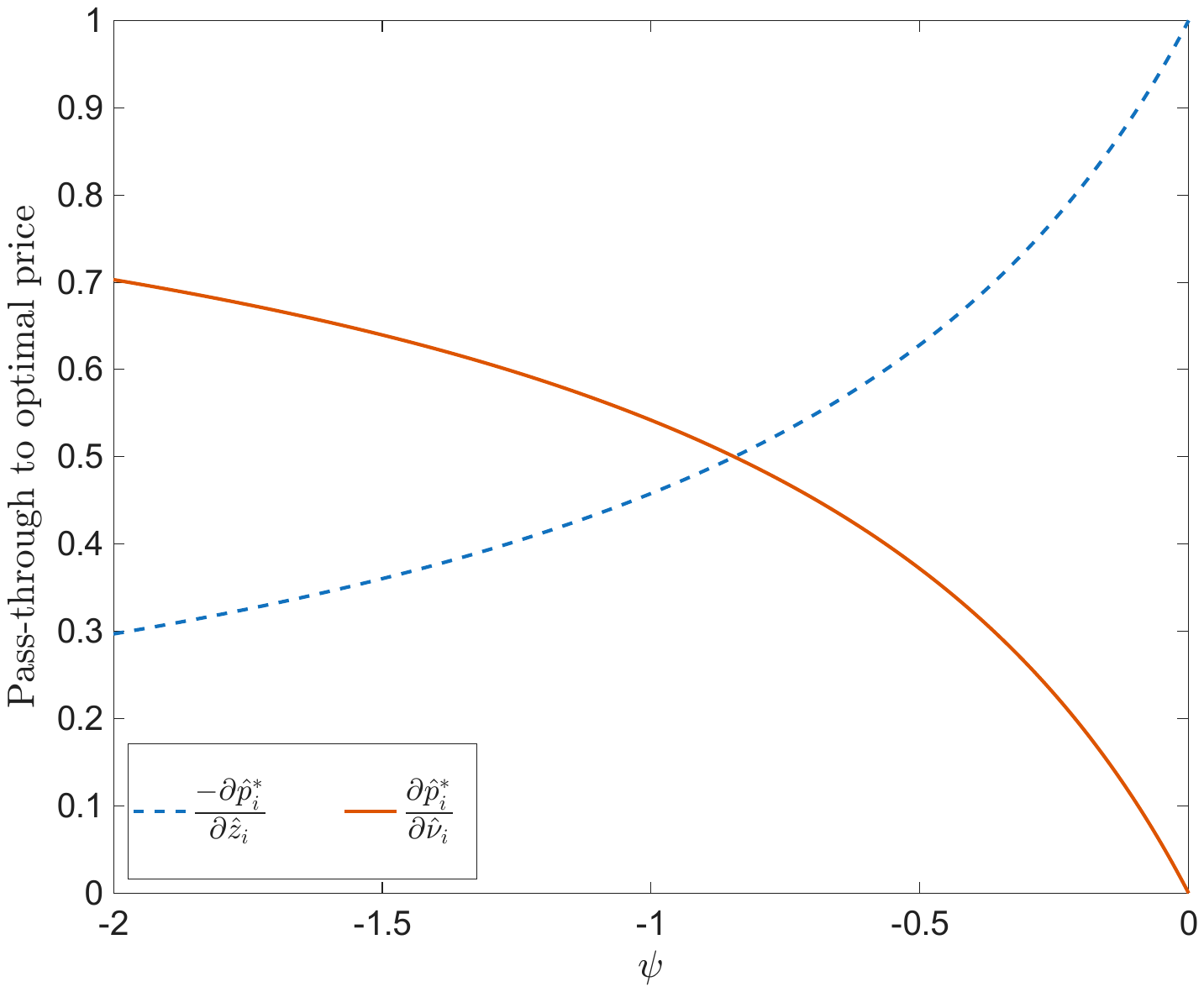}
\end{center}
\label{fig:DS_pass-through}
\footnotesize Notes. The curves plot \(-\partial \widehat{p}_i^{\,*}/\partial \widehat{z}_i=-1/(\omega\psi-1)\) and \(\partial \widehat{p}_i^{\,*}/\partial \widehat{\nu}_i=\omega\psi/(\omega\psi-1)\) as functions of \(\psi\) for a fixed \(\omega\) evaluated around a symmetric equilibrium with \(\nu_i=1\). At \(\psi=0\) (CES), productivity pass-through equals \(-1\) (in terms of \(\widehat{z}_i\)) and demand pass-through is \(0\). For \(\psi<0\), curvature increases, productivity pass-through moves toward \(0\) in absolute value, and demand pass-through becomes positive and larger in magnitude.
\end{figure}

Figure \ref{fig:DS_pass-through} plots $-\partial \widehat{p}_i^{\,*}/\partial \widehat{z}_i$ and $\partial \widehat{p}_i^{\,*}/\partial \widehat{\nu}_i$ as a function of $\psi$. When \(\psi=0\) (CES), \(\partial \widehat{p}_i^{\,*}/\partial \widehat{z}_i=-1\) and \(\partial \widehat{p}_i^{\,*}/\partial \widehat{\nu}_i=0\). For \(\psi<0\), curvature reduces the sensitivity of prices to productivity in absolute value and raises their sensitivity to idiosyncratic demand.

\section{Quantitative Menu-Cost Model}  \label{sec:quantitative}

The model consists of a measure 1 of identical households, perfectly competitive final good producers and monopolistically competitive intermediate good producers. Throughout this section, we use $x$, $x'$ and $x_{-1}$ to denote the value of a generic variable in the current, next and previous periods, respectively. Much of the detail is relegated to Appendix~\ref{app:model}, including the household's problem, which is entirely standard.\footnote{A representative household supplies labor to intermediate-variety firms in exchange for wage payments, trades a complete set of Arrow-Debreu securities, and consumes a final good. It also owns all firms in the economy and receives all accrued profits.}

\subsection{Final Good Producers} \label{sec:final_producer_model}
A representative firm combines intermediate varieties, $y^i$, to produce the final good, $Y$, using a \cite{kimball1995quantitative} aggregator $G(\cdot)$ as defined in \eqref{eq:kimball_spec}. The representative final-good producer chooses $y^i$ to maximize profits normalized by nominal expenditure $P Y$, taking as given variety prices $p^i$, the aggregate price level $P$, and idiosyncratic demand shifters $\nu^i$, solving
\begin{equation} \label{eq:final_producer_problem}
    \underset{y^i \geq 0}{\text{max}} \quad 1 - \int \limits_{0}^{1}\frac{p^i y^i}{P Y} \, di \quad \text{subject to } \int \limits_{0}^{1} G\left(\frac{\nu^i y^i}{Y}\right) di = 1.
\end{equation}
Under Kimball demand there is a relative price above which residual demand is zero, often labeled the ``choke price''. The problem above makes it clear that the intermediate variety producer's output remains non-negative, and we assume that in the case where the non-negativity constraint binds, the firm remains dormant for the period, producing no output.  Appendices \ref{app:app_final_producer} and \ref{app:app_int_producer} provide details about how the presence of the choke price and dormancy of the firms affect the analysis.



\subsection{Intermediate Variety Producers}

A continuum of intermediate-good producers, indexed by $i$, each produce a differentiated variety $y^i$ using a linear production technology with labor as the only input $y^i = z^i l^i$.\footnote{In the model, we consider only a single input in labor, but it should be interpreted as a composite input that combines labor, capital, and materials when connected to the data.} Producers are heterogeneous in their productivity, $z^i$, and face idiosyncratic demand shocks for their variety, $\nu^i$. In particular, idiosyncratic productivity, $z^i$, and idiosyncratic demand, $\nu^i$, evolve according to a VAR(1) process
\begin{equation}
\begin{pmatrix}
\log(z^i) \\ \log(\nu'^i)
\end{pmatrix}
=
\begin{bmatrix}
\rho_z & 0 \\ 0 & \rho_\nu
\end{bmatrix}
 \begin{pmatrix}
\log(z^i_{-1}) \\ \log(\nu^i_{-1})
\end{pmatrix} + u^{i} \text{ where } u^{i}  \sim  N\left( 0 , \begin{bmatrix}
    \sigma^2_z & \rho_{z\nu}\sigma_z \sigma_{\nu} \\ \rho_{z\nu}\sigma_z \sigma_{\nu} & \sigma^2_\nu
\end{bmatrix} \right). \label{eq:lom_zn}
\end{equation}
In our baseline calibration we set $\rho_{z\nu}=0$ and consider correlated innovations as a robustness exercise.

At the beginning of each period, intermediate-good producers inherit their prices from the previous period $p_{-1}^i$ and observe the realizations of $z^i$ and $\nu^i$. They then decide whether or not to adjust their nominal prices and if so, by how much. Nominal price adjustments are subject to a fixed cost, $f$, paid in labor at the contemporaneous wage.\footnote{When the non-negativity constraint in production stops binding, the firm needs to pay a fixed cost to reset its price to a level consistent with positive production.} Firms therefore face a standard discrete choice between keeping the inherited price and paying the menu cost to reset; Appendix \ref{app:app_int_producer} states the firm’s dynamic program and the implied profit function under Kimball demand.

\subsection{Aggregate Nominal Expenditure and Equilibrium}

Nominal expenditure $S = P Y$ grows deterministically at a constant rate, $\mu$. Because there is no aggregate uncertainty, we focus on a stationary equilibrium consisting of a time-invariant distribution over firms' idiosyncratic states $\left(p,\nu,z\right)$, two aggregate price indices $(P,\Lambda)$, and decision rules that jointly satisfy agents' optimality and aggregate market-clearing conditions. When we study monetary non-neutrality in Section~\ref{sec:nonneutrality}, we consider an unexpected one-time increase in nominal expenditure, $S$, after which $S$ returns to its deterministic growth path. The formal definition of the stationary equilibrium is laid out in Appendix \ref{app:eqm_definition}, and Appendix \ref{app:computational_strategy} describes the solution method.

\section{Calibration, Identification, and Micro Implications} \label{sec:calibration}

Conventional quantitative menu-cost models typically prioritize matching a small set of pricing moments, such as the frequency and size of price changes, and treat the underlying idiosyncratic driving processes as residual objects chosen to fit these moments. Our approach instead disciplines the primitives that govern firms' desired prices using direct evidence on firm dynamics. We calibrate the idiosyncratic productivity and demand processes and the curvature of demand using firm-level moments, and we target only the frequency of price changes. The remaining pricing moments are left untargeted and are reported as micro implications of the calibrated model.

There are 11 parameters to discipline. Four $(\beta,\chi,\mu,\rho_{z\nu})$ are externally calibrated. The remaining seven $(\rho_z,\sigma_z,\rho_\nu,\sigma_\nu,\omega,\psi,f)$ are determined internally. For the internal discipline we rely on \cite{foster2008reallocation}, who provide separate, empirically grounded estimates of firm-level productivity and demand shocks for U.S. manufacturing firms. Their decomposition identifies the variances and persistence of these processes and delivers additional firm-level moments that help pin down the elasticity level and curvature of the Kimball demand system.

\subsection{Externally Calibrated Parameters}

\begin{table}[t!]
    \centering
    \caption{Externally Calibrated Parameters}
    \scalebox{0.9}{
    \begin{tabular}{clcl} \toprule
        Parameter & Description & Value & Source \\ \midrule
        $\beta$ & Discount factor & 0.9966 & Annual discount rate of $4\%$ \\
        $\chi$ & Labor disutility & 1 & Normalization \\
        $\mu$  & Growth rate of $S$ & 0.002 & Annual inflation rate of 2.4\% \\ 
        $\rho_{z\nu}$ & Corr. between TFPQ and demand shocks & 0 & \cite{foster2008reallocation} \\ \bottomrule
    \end{tabular}  
    }
    \justify
    \footnotesize{Note: This table displays the externally calibrated parameters in the model.}
    \label{tab:ext_calib}
\end{table}

We begin by detailing the externally calibrated parameters. The model is calibrated to U.S. data, with each model period representing one month. The monthly discount factor $\beta$ is set to $0.9966$, corresponding to an annual discount rate of 4\%. Following standard practice in the literature \citep{midrigan2011menu}, the disutility of labor $\chi$ is normalized to 1, which pins down the units of labor and wages. The monthly growth rate of nominal expenditure $S$, $\mu$, is set at $0.2\%$, implying an annual inflation rate of approximately 2.4\%. For firm-level processes, we assume that idiosyncratic demand and productivity innovations are uncorrelated, $\rho_{z\nu}=0$.\footnote{This assumption aligns with \cite{foster2008reallocation}, where $\rho_{z\nu}=0$ is necessary for their estimation strategy, which assumes orthogonality between demand and productivity. Because TFPQ in their framework is a physical quantity-based measure---constructed from output quantities and input usage rather than revenues---it is less susceptible to contamination by demand-side variation than revenue-based alternatives. Using Colombian data and an alternative strategy that relaxes this assumption, \cite{eslava2024size} report a correlation of $-0.07$ between demand and productivity. We return to this issue in Section \ref{sec:correlated_shocks}, where we show our results are robust to allowing for correlated shocks.} Table \ref{tab:ext_calib} summarizes the externally calibrated parameters.

\subsection{Internally Calibrated Parameters} \label{sec:int_calib}

\cite{foster2008reallocation} provide direct estimates of firm-level idiosyncratic productivity and demand processes using data on U.S. manufacturing firms from 1977 to 1997.
We rely on these estimates to discipline the parameters governing the AR(1) processes for idiosyncratic firm productivity $(\rho_z,\sigma_z)$ and demand shocks $(\rho_\nu,\sigma_\nu)$ in the model. In addition, we use two moments from the same study to calibrate the parameters governing the Kimball demand system $(\omega,\psi)$. In contrast to the conventional calibration strategy in the menu cost literature, we rely on only a single pricing moment, the frequency of price changes.\footnote{Appendix \ref{app:calib_size} considers an alternative strategy where we target the average size of price changes instead of the frequency and show that the two strategies lead to very similar outcomes.}

We jointly calibrate these seven parameters to match the frequency of price changes and six firm-dynamics moments from \cite{foster2008reallocation}: the five-year autocorrelation and the cross-sectional dispersion of TFPQ and of idiosyncratic demand (measured as $\log \nu_t^i$), the correlation between price and TFPQ, and the estimated demand elasticity from an instrumental variable regression.\footnote{This refers to Equation (10) in \cite{foster2008reallocation}, which we restate in \eqref{eq:IV_reg}.} We construct these firm-dynamics moments from model-simulated data in the exact same way as \cite{foster2008reallocation} do in their empirical analysis.\footnote{Two maintained assumptions in \cite{foster2008reallocation} are flexible prices and CES demand. We compute the moments on model-generated data using the same procedure as in \cite{foster2008reallocation}, so any effects from these maintained assumptions affect the empirical and model moments through the same measurement procedure.} Details on the empirical targets and the calibration algorithm are delegated to Appendix \ref{app:calibration}.

Before presenting the results, it is useful to summarize how key target moments discipline specific parameters. The fixed cost, $f$, governs the model-implied frequency of price changes. The shock-process parameters $(\rho_z,\sigma_z,\rho_\nu,\sigma_\nu)$ discipline the persistence and dispersion of firm-level productivity and demand. The Kimball parameters $(\omega,\psi)$ govern the level and curvature of demand. Section~\ref{sec:simple_demand} shows that under CES demand, or under a level-only placement of idiosyncratic demand, desired prices are a constant markup over marginal cost and $\mathrm{Corr}(p,\mathrm{TFPQ})=-1$. With non-CES demand ($\psi<0$) and an effective-share placement of idiosyncratic demand shocks, $\mathrm{Corr}(p,\mathrm{TFPQ})>-1$. In the data, \cite{foster2008reallocation} report $\mathrm{Corr}(p,\mathrm{TFPQ})=-0.54$, favoring non-CES demand and an effective-share placement.

The parameter $\omega$ is disciplined by the IV demand elasticity estimated using productivity as an instrument for prices. Under symmetry and CES, the log-demand slope is $\omega/(1-\omega)<0$, so the absolute demand elasticity is $|\omega/(1-\omega)|$. When $\psi<0$, the log-linear specification is misspecified, but we estimate the same regression on model-simulated data, so the same mapping applies to the model and the data. Appendix~\ref{app:identification} discusses the identification of the model parameters in more detail.

\begin{table}[t!]
\centering
\caption{Internal Calibration}
\scalebox{0.9}{
\begin{tabular}{ccc} \toprule
\textbf{Moment} 				                                   & \textbf{Data}  & \textbf{Baseline Model} \\ \midrule
\multicolumn{1}{l}{Frequency of price changes}                     & 0.11 		    & {0.11}   \\ \midrule
\multicolumn{1}{l}{5-year autocorrelation of ${z_t^i}$}            & 0.31           & {0.31} \\
\multicolumn{1}{l}{Cross-sectional standard deviation of $z_t^i$}  & 0.26           & {0.25} \\
\multicolumn{1}{l}{5-year autocorrelation of $\log \nu_t^i$}       & 0.62           & {0.58} \\ 
\multicolumn{1}{l}{Cross-sectional standard deviation of $\log \nu_t^i$} & 1.16     & {1.14} \\ 
\multicolumn{1}{l}{IV coefficient}                                  & --2.40         & {--2.31} \\
\multicolumn{1}{l}{Corr. between price and TFPQ}                    & --0.54         & {--0.54}  \\ \midrule
    \textbf{Parameter} & \textbf{Description} 			                    & \textbf{Value}  \\ \midrule
    $\psi$  & \multicolumn{1}{l}{Super-elasticity} 		     	            & {--1.10}  \\
    $\omega$ & \multicolumn{1}{l}{Elasticity}	                            & {1.18}	\\
    $\rho_z$  & \multicolumn{1}{l}{Persistence of $z_t^i$}    	            & {0.98} \\
    $\sigma_z$  & \multicolumn{1}{l}{Standard deviation of $z_t^i$}     	& {0.06} \\
    $\rho_\nu$  & \multicolumn{1}{l}{Persistence of $\log \nu_t^i$}  	    & {0.998} \\
    $\sigma_\nu$  & \multicolumn{1}{l}{Standard deviation of $\log \nu_t^i$} & {0.03} \\
    $f$       & \multicolumn{1}{l}{Menu cost}                               & {0.02} \\ \bottomrule
\end{tabular}
}
\justify
\footnotesize{Note: The top panel compares the targeted moments and model-implied moments. The bottom panel reports the internally calibrated parameter values.}
\label{tab:int_calib}
\end{table}

The results of the internal calibration are reported in Table \ref{tab:int_calib}. The top panel shows the targeted moments, whose calculations are described in more detail in Appendices \ref{app:FHS_moments} and \ref{app:pricing_moments}. The bottom panel reports the seven parameters calibrated jointly. The results indicate that all seven moments are matched very closely. Under this calibration, price adjustment costs paid represent 0.18\% of total firm revenue in a given period. This is comparable to, though somewhat smaller than, the numbers reported in other quantitative menu-cost models such as \cite{golosov2007menu} (0.25\%), \cite{midrigan2011menu} (0.34\%), and \cite{vavra2014inflation} (0.5\%), as well as direct empirical estimates in \cite{levy1997magnitude} (0.7\%) and \cite{zbaracki2004managerial} (0.3\%, when considering only physical costs and managerial costs).

\begin{table}[t!]
\centering
\caption{Elasticity and Super-elasticity Specifications from the Literature}
\label{tab:compare_calib_elas}
\begin{tabular}{ccc} \toprule
     Source & Elasticity & Super-elasticity\\ \midrule
    \cite{kimball1995quantitative}    & 11 & 31.8  \\
    \cite{smets2007shocks} & 2.7 & 10  \\
    \cite{bergin2000staggered}   & 3 & 1.3 \\
    \cite{woodford2003} &    7.8 & 6.7  \\
    \cite{chari2000sticky}    & 10 & 35.7  \\
    \cite{fisher2005evaluating}    & 11 & 10   \\
    \cite{gopinath2010frequency}    & 5 & 4  \\
    \cite{klenow2016real} &    5 & 10  \\
    \cite{beck2020price} &   3.2 & 1.93  \\
    \cite{harding2022resolving}    & $\{11, 3.9\}$ & $\{134.2, 64.5\}$   \\ 
    \cite{harding2023understanding}    & 3.9 & 64.5   \\
    \cite{aruoba2025pricing}    & 4.1 & 6.9   \\  \midrule
    This paper   & 6.5 & 7.1   \\ \bottomrule
\end{tabular}
\justify
\vspace*{-0.1in}
\footnotesize{Note: This table summarizes the demand elasticity and super-elasticity used or estimated in the literature. We report two parameterizations from \cite{harding2022resolving}.
Appendix~\ref{app:compare_calib_elas_source} describes in more detail the source of elasticities and super-elasticities in the studies listed. 
For a broader survey of demand‑side real rigidity estimates, see \cite{shirota2025demand}, who compile a more extensive set of studies.} \\

\end{table}

The calibrated processes for idiosyncratic demand and productivity are both highly persistent at the monthly frequency, with autocorrelations of 0.998 and 0.98, respectively. To aid interpretation, it is useful to translate these to annual frequency. The monthly estimates imply annual AR(1) coefficients of approximately 0.78 for productivity and 0.98 for demand, and annual innovation standard deviations of approximately 0.19 for productivity and 0.10 for demand. These magnitudes are consistent with external estimates from the firm dynamics literature. For productivity, \cite{decker2020changing} use the Longitudinal Business Database and report annual persistence of revenue-based TFP in the range of 0.65 to 0.80, bracketing our implied value of 0.78. Their estimates of cross-sectional TFP dispersion (0.35 to 0.45) are somewhat larger than the 0.26 we target from \cite{foster2008reallocation}, likely because revenue-based measures confound productivity with demand variation. For demand, the implied annual persistence of 0.98 is high but consistent with evidence from marketing and applied microeconomics documenting highly durable consumer preferences. \cite{dube2010state} show that structural state dependence generates long-lasting brand loyalty, and \cite{bronnenberg2012evolution} find that geographic variation in market shares is explained by extremely persistent brand preferences that erode only slowly over time. While innovations to idiosyncratic productivity are larger at the monthly frequency, the stationary distribution of idiosyncratic demand exhibits greater cross-sectional dispersion due to its near-unit-root persistence.

Regarding the parameters governing the shape of the demand function, the calibrated values $\psi=-1.10$ and $\omega=1.18$ imply a price elasticity of demand of 6.5, a super-elasticity of 7.1 under a symmetric equilibrium. Table \ref{tab:compare_calib_elas} places these implied elasticities in the context of existing calibrations and estimates. The implied elasticity and super-elasticity in our baseline calibration lie within the range commonly used in the literature. It is also worth noting that micro-based estimates of demand elasticities and super-elasticities are rare, and \cite{beck2020price} is a notable exception.

\subsection{Micro Implications Beyond the Calibration Targets}  \label{sec:untargeted_moments}

This subsection reports moments that are not targeted in the internal calibration and are therefore implications of the calibrated model. We begin with standard micro pricing moments from CPI data, then examine the implied hazard of price adjustment. We then turn to implications for the cross-sectional distribution of markups and for cost pass-through.

\subsubsection{Pricing}

\begin{table}[t!]
    \centering
     \caption{Untargeted Pricing Moments}
    \begin{tabular}{ccc} \toprule
        Moments & Data &  Baseline \\ \midrule
        Average Size     & 0.08 &  0.06 \\
        Fraction Up      & 0.65 &  0.55 \\
        SD($\Delta p$)   & 0.08 & 0.07 \\  \bottomrule
    \end{tabular}
    \label{tab:untargeted_mom}
    \justify
\footnotesize{Note: This table shows three untargeted moments: average size of adjustment conditional on a price change, the fraction of adjustments that are positive, and the standard deviation of price changes excluding zeros, computed from the data and from model simulated data.}
\end{table}

Table \ref{tab:untargeted_mom} reports three important pricing moments that are not targeted in the baseline calibration, alongside their data counterparts computed from U.S. CPI microdata.\footnote{We borrow these estimates from \cite{vavra2014inflation}, which we review in more detail in Appendix \ref{app:pricing_moments}. One may be concerned about whether pricing facts from CPI data are the correct benchmark. However, \cite{nakamura2008five} show that the key pricing moments computed from Producer Pricing Index data do not differ significantly.} These moments include the average size of a price change conditional on a change, the fraction of adjustments that are positive, and the dispersion of non-zero price changes. The model captures the broad magnitudes of these moments despite not targeting them. This is noteworthy because the key demand parameters are disciplined using firm-level evidence, including the price--TFPQ correlation (Section \ref{sec:simple_demand}) and the \cite{foster2008reallocation} decomposition of idiosyncratic productivity and demand.

We also examine the price adjustment hazard function generated by the model, defined as the probability that a price changes  $\cal H$ periods after the last adjustment, conditional on the price spell lasting $\cal H$ periods. Empirically, the hazard is observed to be either mildly downward-sloping \citep{nakamura2008five,baley2019firm} or approximately flat \citep{klenow2008state}.\footnote{While the literature generally finds a mildly negative slope, \cite{alvarez2023consistent} show that controlling for product heterogeneity can reverse the sign of the slope. In a menu-cost model, the hazard shape is sensitive to the size and persistence of idiosyncratic driving forces. Persistent idiosyncratic shocks generate more reasons for adjustment shortly after a previous change, raising the early part of the hazard and tending to flatten or mildly tilt it downward. In our calibration, both productivity and demand shocks are disciplined by \cite{foster2008reallocation}, and the presence of two orthogonal shocks increases the scope for early adjustments relative to a one-shock environment.}

\begin{figure} [t!]
    \centering
    \caption{Hazard Function of Price Change}
    \includegraphics[scale=0.09]{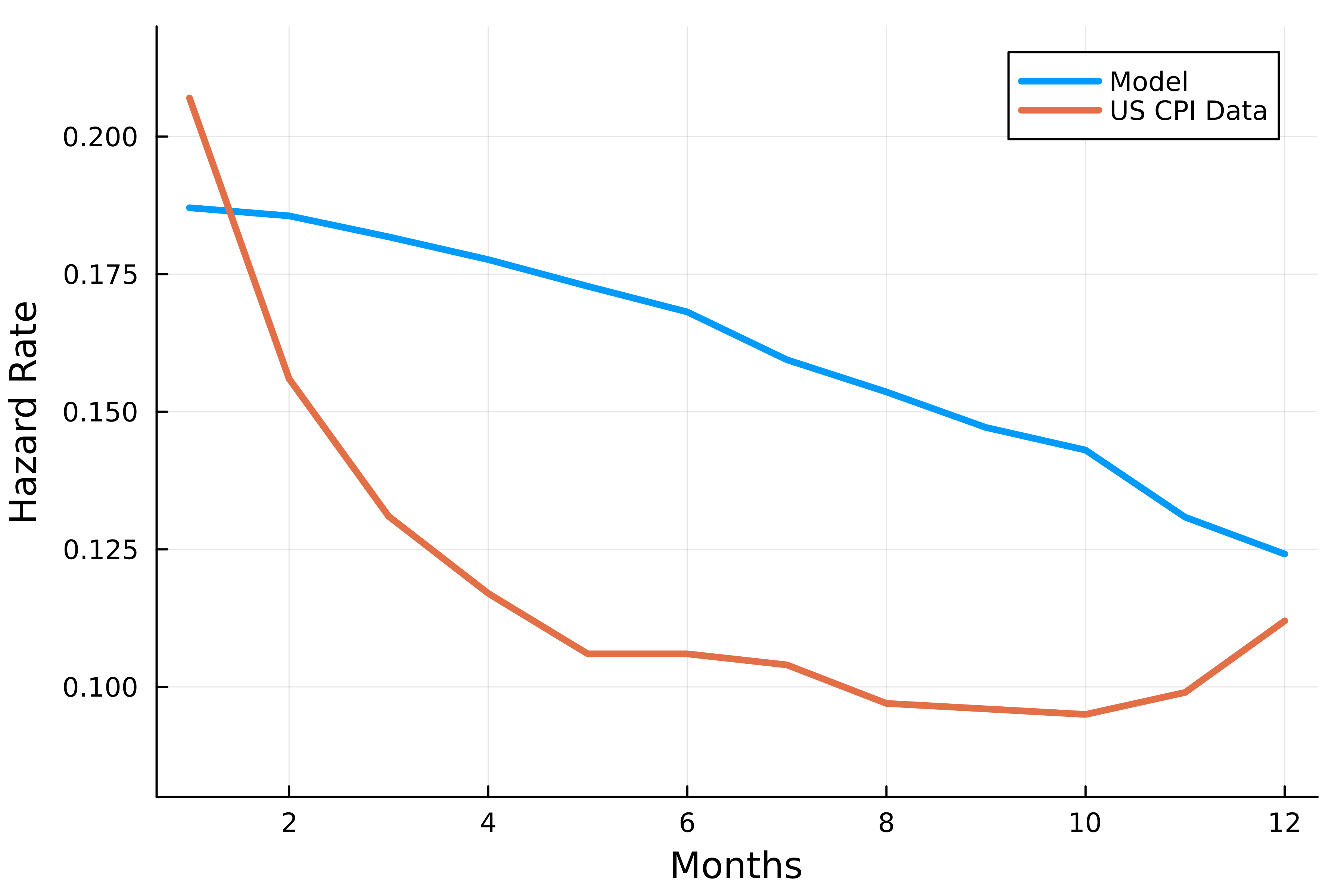}
    \label{fig:hazard}
    \justify
    \footnotesize{Note: This figure plots the pricing hazard from model generated data against the empirical hazard estimated from U.S. CPI in \cite{baley2019firm}.}
\end{figure}

Figure \ref{fig:hazard} compares the hazard implied by the model with the empirical hazard estimated by \cite{baley2019firm} using U.S. CPI data.\footnote{The shape of the hazard function reported by \cite{baley2019firm} aligns with that reported by \cite{nakamura2008five} for processed food, which they suggest is representative of the hazard function for many other product groups.} The model captures the overall mildly downward-sloping pattern, but it does not replicate the very steep initial decline in the empirical hazard. A natural extension is to add firm-level uncertainty and learning, as in \cite{baley2019firm}, which generates clustered adjustments and a high hazard immediately after a price change.

\subsubsection{Markup}

One prominent application of Kimball demand, among other non-CES systems, is in modeling variable markups. Because a firm's desired markup depends on both its idiosyncratic productivity and demand, the cross-sectional distribution of productivity and demand, combined with pricing frictions, results in a non-degenerate markup distribution in the model.

\begin{figure}[t!]
    \centering
    \caption{Cross-Sectional Distribution of Gross Markup: Model vs. Data}       
    \includegraphics[scale=0.4]{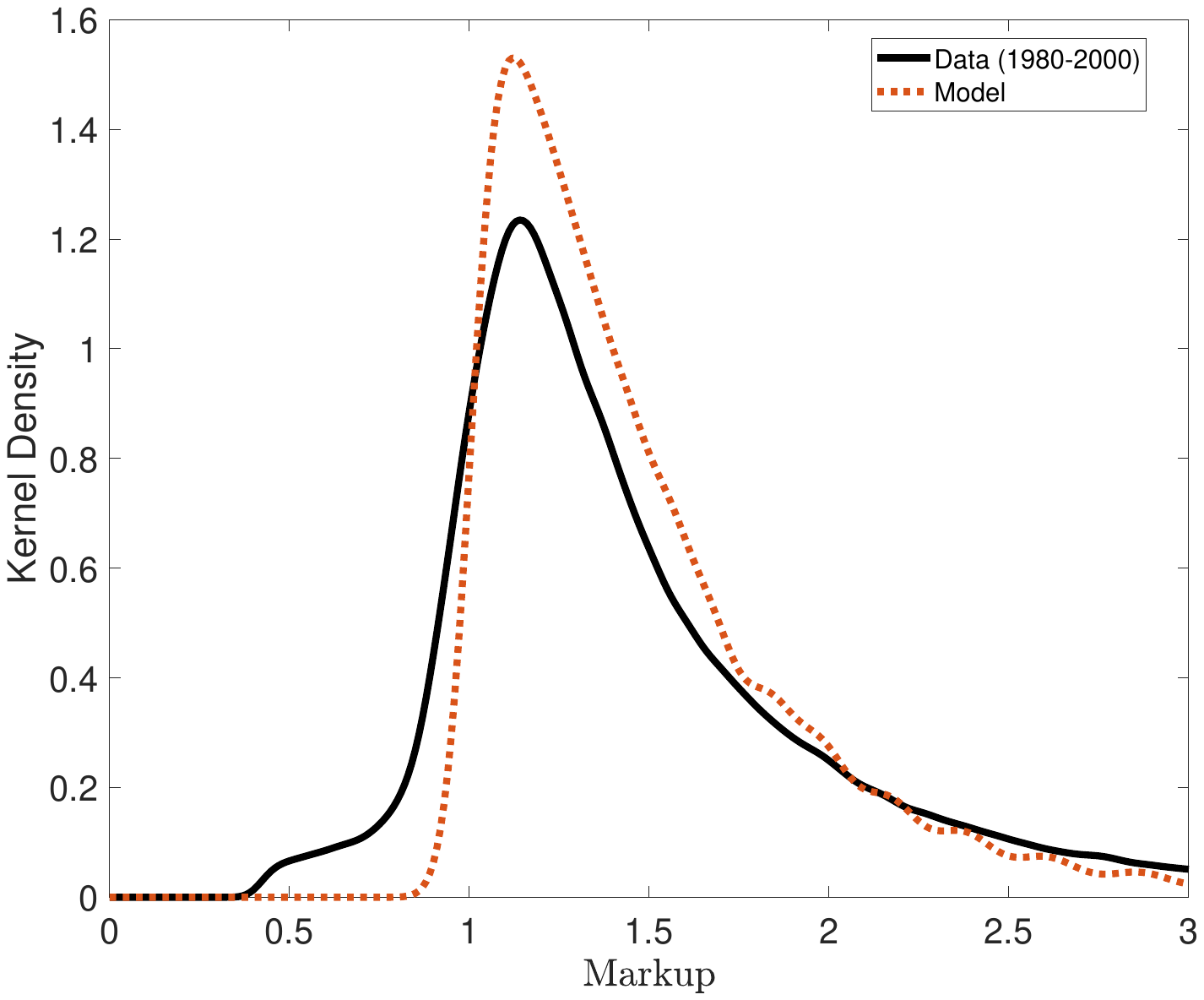}
    \label{fig:ksdensity_markup}
    \justify
    \footnotesize{The figure plots the kernel density of the empirical markup distribution from publicly traded firms in the U.S. as well as the kernel density of the markup distribution in the ergodic distribution of the model. Both kernel densities are computed using the optimal bandwidth for normal densities.}
\end{figure}

Figure \ref{fig:ksdensity_markup} plots the kernel density of the cross-sectional distribution of gross markups from both the model and the data. The empirical distribution of markups is computed using data on public firms between 1980 and 2000, following the method of \cite{de2020rise}.\footnote{Despite the potential limitations of using revenue data to estimate markups, \cite{de2024hitchhiker} show that while the estimated levels of markups can be noisy, this method provides reliable estimates of the dispersion and shape of the markup distribution. More details on the empirical markup estimation can be found in Appendix \ref{app:markup_passthrough}.} We find that the markup distribution generated by the model resembles the untargeted empirical distribution. The mean gross markup in the model is 1.49, compared to 1.56 in the data. The model-implied markup distribution exhibits lower variance than the empirical distribution, particularly in the tails. Specifically, our model generates fewer firms with gross markups exceeding 2 and does not produce markups significantly below 1. Matching these extreme tails would require additional features, such as non-Gaussian demand shocks, monopolistic firms, customer capital, or firm exit.

\subsubsection{Cost Pass-through}

We now compare the degree of cost pass-through implied by our calibrated demand system to empirical estimates from the literature. Given the calibrated values of $\omega$ and $\psi$, the model implies a desired-price cost pass-through of 43\%. The literature in both open-economy and closed-economy macroeconomics, reviewed in Appendix \ref{app:markup_passthrough}, has estimated cost pass-through using various datasets from different countries and consistently finds pass-through rates in the range of 20\% to 50\%.

\subsubsection{Firm Dynamics}

Because idiosyncratic demand and productivity processes determine the ergodic properties of firm growth, we can also compare the cross-sectional dispersion of output growth rates computed from model-simulated data with external evidence to gauge whether the estimates from \cite{foster2008reallocation} can be generalized beyond the eleven industries they study. We benchmark our estimates to \cite{davis2006volatility}, a study that utilizes the Longitudinal Business Database, providing a comprehensive measure of U.S. business dynamics. They estimate the cross-sectional standard deviation of firm revenue growth rates to be 0.39 over the period 1982--1997, while in our simulated data, this untargeted moment is 0.55.



\section{External Validity and Robustness} \label{sec:ext_valid_robust}

Section \ref{sec:calibration} disciplines demand curvature and firm-level shock processes using firm
dynamics evidence and targets only the frequency of price changes. This section evaluates how portable
these disciplined inputs are outside the narrow U.S. manufacturing industries studied by
\cite{foster2008reallocation} and how sensitive the quantitative implications are to alternative
calibration choices and modeling assumptions.

\subsection{External Validity: An Application to Colombian Manufacturing Data} \label{sec:external_val}

Drawing on empirical evidence from firm dynamics in narrowly defined U.S. manufacturing industries from
\cite{foster2008reallocation}, our baseline calibration demonstrates that a menu-cost model
incorporating micro real rigidities and disciplined supply and demand processes can align with
non-targeted firm-level pricing facts. A natural concern is whether the underlying firm-level moments
used for discipline generalize beyond the subset of U.S. industries studied by \cite{foster2008reallocation}.
Unfortunately, \cite{foster2008reallocation} remains the only U.S.-based study that systematically
estimates firm-level productivity and demand shocks, mostly because separate identification requires
firm-level price and quantity data.

To assess external validity, we turn to Colombia, where researchers have estimated closely related
moments for the universe of manufacturing firms using firm-level price and quantity data. We first
compare Colombian estimates of key moments that discipline demand curvature to the corresponding U.S.
moments. We then recalibrate the model using Colombian moments to test robustness outside the U.S.
context.

Employing largely the same methodology as \cite{foster2008reallocation}, \cite{eslava2013trade} use data
covering the universe of Colombian manufacturing firms to separately identify productivity and demand
processes at the firm level. They report similar values for the estimated IV coefficient and the
correlation between $TFPQ$ and price, which are crucial for disciplining demand curvature. Specifically,
they report Corr($TFPQ,P$) = $-0.65$ and an IV coefficient of $-2.05$, compared to the
\cite{foster2008reallocation} values of $-0.54$ and $-2.40$, respectively.\footnote{In
\cite{eslava2013trade}, the authors allow the elasticity of demand to vary by industry based on its
Herfindahl concentration index. The $-2.05$ is the two-stage least squares coefficient for the average industry, which has a Herfindahl index of 0.18.}

More recently, \cite{eslava2024size} apply an alternative approach to the same Colombian manufacturing
data, relaxing the orthogonality assumption between idiosyncratic supply and demand shocks. As we
explain in Appendix \ref{app:eslava_strategy}, rather than separately estimating production and demand
functions via an instrumental variables approach, they employ a Generalized Method of Moments (GMM)
framework to jointly estimate productivity and demand, which allows for a non-zero correlation between
productivity and demand innovations.

To test robustness outside the U.S. context, we recalibrate the baseline model using Colombian moments
and allowing productivity and demand innovations to be correlated. We follow a strategy similar to our
main analysis, selecting parameters $\{f,\rho_z,\sigma_z,\rho_\nu,\sigma_\nu,\omega,\psi\}$
to match six moments from \cite{eslava2024size} along with Colombia's price adjustment frequency. The
firm dynamics moments used to discipline the model are largely identical to those in the main analysis
with one exception. Because \cite{eslava2024size} use a GMM approach rather than the IV strategy in
\cite{foster2008reallocation}, we replace the IV estimate with the correlation between physical output
and markup. We also directly fix $\rho_{z\nu} = -0.07$, based on the estimate of \cite{eslava2024size}. For Colombia's price adjustment frequency, we draw on \cite{julio2011rigideces}, who compute micro pricing moments using Colombian CPI data. Between 1999 and 2008, the average monthly price
adjustment frequency for goods in the CPI basket ranged from approximately $10\%$ to $15\%$. In the
calibration, we target 12.5.

Table \ref{tab:colombia_calib} reports the calibration to Colombian data. Overall, the model fits
Colombian firm dynamics statistics well. Similar to what we observe in the U.S., the calibrated
idiosyncratic demand and productivity processes are both highly persistent, with a monthly autocorrelation
of 0.996 for demand and 0.988 for productivity, and idiosyncratic demand is more persistent than
idiosyncratic productivity. As in the U.S. calibration, innovations to productivity are more volatile
than innovations to demand. In terms of demand curvature, the calibrated values of $\psi$ and $\omega$
imply a cost pass-through of 51\%, slightly higher than in the U.S. calibration.

Interestingly, the model calibrated only to firm dynamics moments and price adjustment frequency also
delivers reasonable pricing implications in Colombia. In the calibrated model, the average size of
price adjustment is 20\%, consistent with Figure 8 of \cite{julio2011rigideces}, which shows median
price changes in Colombia fluctuating between 10\% and 20\% over 1999--2008.

\begin{table}[t!]
\centering
\caption{Calibration with Colombian Data}
\label{tab:colombia_calib}
\scalebox{0.9}{
\begin{tabular}{lcc}
\toprule
\textbf{Moment}                                              & \textbf{Data}  & \textbf{Model} \\ 
\midrule
Frequency of price changes                                   & 0.13     & 0.13 \\ 
\midrule
Yearly autocorrelation of $z_t^i$                            & 0.91           & 0.90 \\ 
Cross-sectional standard deviation of $z_t^i$                & 0.75           & 0.78 \\ 
Yearly autocorrelation of $\log \nu_t^i$                     & 0.98           & 0.97 \\ 
Cross-sectional standard deviation of $\log \nu_t^i$         & 0.89           & 1.03 \\ 
Corr. between price and TFPQ                           & --0.73         & --0.76 \\ 
Corr. between output and markup                        & 0.45           & 0.46 \\ 
\midrule
\textbf{Parameter (Internal)}                                & \textbf{Description}                            & \textbf{Value} \\ 
\midrule
$\psi$                                                       & \multicolumn{1}{l}{Super-elasticity}                                & --0.45 \\ 
$\omega$                                                     & \multicolumn{1}{l}{Elasticity}                                      & 2.1 \\ 
$\rho_z$                                                     & \multicolumn{1}{l}{Persistence of $z_t^i$}                          & 0.988 \\ 
$\sigma_z$                                                   & \multicolumn{1}{l}{Standard deviation of $z_t^i$}                   & 0.12 \\ 
$\rho_\nu$                                                   & \multicolumn{1}{l}{Persistence of $\log \nu_t^i$}                   & 0.996 \\ 
$\sigma_\nu$                                                 & \multicolumn{1}{l}{Standard deviation of $\log \nu_t^i$}            & 0.09 \\ 
$f$                                                          & \multicolumn{1}{l}{Menu cost}                                       & 0.04 \\ 
\midrule
\textbf{Parameter (External)}                                & \textbf{Description}                            & \textbf{Value} \\ 
\midrule
$\beta$                                                      & \multicolumn{1}{l}{Discount factor}                                 & 0.9966 \\ 
$\chi$                                                       & \multicolumn{1}{l}{Labor disutility}                                & 1 \\ 
$\mu$                                                        & \multicolumn{1}{l}{Growth rate of $S$}                              & 0.0056 \\ 
$\rho_{z \nu}$    & \multicolumn{1}{l}{Correlation of innovations} & -0.07 \\
\bottomrule
\end{tabular}
}
\justify
\footnotesize{
Note: The top panel compares the targeted moments and model-implied moments in the model
calibrated to Colombian data. The bottom panel reports the internally and externally calibrated 
parameter values. The discount factor $\beta$ implies an annual discount rate of 4\%. Labor 
disutility $\chi$ is normalized to 1. The growth rate $\mu$ is chosen to match an annual 
inflation rate of 7\%.}
\end{table}

\subsection{Robustness} \label{sec:robustness}

In this section, we explore a range of robustness exercises to our quantitative analysis.

\subsubsection{Targeting Broader Moments}
\label{sec:Model_1B}

Our baseline calibration follows the tradition in the menu-cost literature, where a parsimonious set of key moments is explicitly targeted using the same number of parameters, while the remaining moments are reported as untargeted checks. Here we consider an alternative strategy in which we choose the seven parameters $\{f,\rho_z,\sigma_z,\rho_\nu,\sigma_\nu,\psi,\omega\}$ to target a broader set of moments, including the pricing moments and markup moments reported above, in addition to the firm-dynamics moments used in the baseline.

\begin{table}[t!]
\centering
\caption{Internal Calibrations: Targeting Broader Moments} \label{tab:Model_1B}
\scalebox{0.9}{
\begin{tabular}{lcccc} \toprule
   {\textbf{Parameter}}& \textbf{Description} & \textbf{Baseline} &  \textbf{Broad Targets}\\ \midrule
    $f$       & \multicolumn{1}{l}{Menu cost}      & 0.016 & 0.019 \\
    $\rho_z$  & \multicolumn{1}{l}{Persistence of $z^i_t$}      & 0.977 & 0.979 \\
    $\sigma_z$& \multicolumn{1}{l}{Standard deviation of $z^i_t$}      & 0.06 & 0.06 \\
    $\rho_\nu$  & \multicolumn{1}{l}{Persistence of $\log \nu^i_t$}       & 0.998 & 0.998 \\
    $\sigma_\nu$& \multicolumn{1}{l}{Standard deviation of $\log \nu^i_t$}      & 0.030 & 0.032 \\
    $\psi$    & \multicolumn{1}{l}{Demand curvature}      & --1.10 & --0.93 \\
    $\omega$  & \multicolumn{1}{l}{Elasticity of substitution} & 1.18 & 1.42 \\ \midrule
    \multicolumn{1}{l}{\textbf{Narrow Moments}} & \textbf{Data} \\ \midrule
    SD($z_t^i$)                 & 0.26 & \textbf{0.25} & \textbf{0.26} \\
    5-year autocorr. of $z_t^i$ & 0.31 & \textbf{0.31} & \textbf{0.31} \\
    SD($\log \nu_t^i$)                 & 1.16 & \textbf{1.14} & \textbf{1.07} \\
    5-year autocorr. of $\log \nu_t^i$ & 0.62 & \textbf{0.58} & \textbf{0.69} \\
    IV coefficient              & --2.40 & \textbf{--2.45} & \textbf{--3.84} \\
    Corr(price, TFPQ)           & --0.54 & \textbf{--0.54} & \textbf{--0.47} \\
    Frequency of price changes  & 0.11 & \textbf{0.11} & \textbf{0.11} \\ \midrule
    \multicolumn{4}{l}{\textbf{Broader Moments}} \\ \midrule
    Growth dispersion           & 0.39 & 0.31 & \textbf{0.44} \\
    Fraction of price increases & 0.65 & 0.55 & \textbf{0.59} \\
    Average size of price change& 0.08 & 0.06 & \textbf{0.08} \\
    SD($\Delta p$)              & 0.08 & 0.07 & \textbf{0.09} \\
    Avg. markup                 & 1.56 & 1.49 & \textbf{1.71} \\
    SD(markup)                  & 0.72 & 0.47 & \textbf{0.63} \\ \midrule
    Pass-through of supply shocks & 20\%--50\% & 0.43 & 0.43 \\ \bottomrule
\end{tabular}
}
\justify
\footnotesize{This table reports the baseline calibration targeting only a narrow set of moments in column 3 and an alternative calibration in column 4 that targets both the narrow set of moments and the broader moments previously reported as untargeted validation checks. Bold-faced numbers denote targeted moments.}
\end{table}

Table~\ref{tab:Model_1B} summarizes the results. Targeting the broader set of moments improves the fit to the previously untargeted pricing and markup moments, though at the expense of a weaker fit to the IV coefficient and the correlation between price and TFPQ. The resulting parameter values remain close to the baseline calibration. In particular, both calibrations imply similar demand curvature and a similar degree of cost pass-through (0.43 in both cases).

\subsubsection{CES Demand with Decreasing Returns to Scale}

Our baseline focuses on Kimball demand as a source of micro real rigidity. An alternative source is
cost-side curvature stemming from decreasing returns to scale, which can generate strategic complementarity in
pricing even under CES demand. We explore this alternative in an environment similar to
\cite{burstein2007prices}, featuring CES demand and decreasing returns to scale at the firm level. By
applying the same calibration targets as in our baseline calibration, we show in Appendix
\ref{app:DRS} that this setup can also match several targeted and untargeted moments of the data.
However, to do so it requires a decreasing returns of 0.59, which is much lower than what is
typically documented in the empirical literature -- for example \cite{basu1997returns} reports results that are  close to constant returns to scale. This version of the model also does not replicate the untargeted markup distribution well. While we focus on demand-side curvature in the main analysis for simplicity, combining demand-
and cost-side sources of micro real rigidities is a natural direction for future work.

\subsubsection{Correlated Shocks} \label{sec:correlated_shocks}

The estimation strategy in \cite{foster2008reallocation} imposes orthogonality between idiosyncratic
productivity and demand shocks, which we adopt in our baseline calibration. Using Colombian
manufacturing data and an alternative identification strategy, \cite{eslava2024size} relax this
assumption and estimate a weak negative correlation between demand and productivity innovations of
$-0.07$.

Appendix \ref{app:correlated_shocks} reports two alternative calibrations in which we allow
$\rho_{z\nu}\neq 0$ while holding all other parameters fixed at their baseline values. The main
quantitative implications are similar even when we allow substantially larger correlations, such as
$0.4$ or $-0.4$. The correlation mainly affects the average size of price changes. When productivity
and demand shocks are positively correlated, the average absolute price change is smaller because
states with high productivity, which put downward pressure on desired prices, tend to coincide with
high demand, which pushes desired prices up, partially offsetting each other.

\subsubsection{Leptokurtic Shocks}
\label{sec:lepto}

Standard menu-cost models, including our baseline specification, typically generate a distribution of
non-zero price changes that is too thin-tailed relative to the data. Empirically, the distribution of
price changes displays substantial kurtosis (e.g., \cite{alvarez2016real}). In Appendix
\ref{app:lepto}, we follow \cite{midrigan2011menu} and introduce leptokurtic idiosyncratic demand
innovations to better match the tails of the price-change distribution. This modification improves the
fit of higher moments of $\Delta p$ while leaving the main quantitative conclusions of the paper
unchanged.

\subsubsection{Klenow and Willis (2016): Critique and the Role of Demand Shocks}
\label{sec:kw}

\cite{klenow2016real} take demand curvature as given, an elasticity level together with a super-elasticity parameter, and ask whether a menu-cost model can match the observed distribution of price changes. Their conclusion is that, under high-curvature parametrizations used in much of the macro literature, matching the price-change distribution requires a very volatile idiosyncratic residual process, with sharp implications for firm-level real outcomes. The purpose of this subsection is to restate this tension in terms of observable moments, then to clarify what is, and is not, identified in a productivity-only environment. Table~\ref{tab:kw_variants} reports three Klenow--Willis style variants featuring only idiosyncratic productivity shocks. KW1 is a broad productivity-only calibration. KW2 fixes Kimball curvature at values mapped from \cite{klenow2016real}. KW3 refits curvature to match the size of price changes and average markups. The baseline with idiosyncratic demand shocks is included for comparison.

\begin{table}[t!]
\centering
\caption{Alternative Model: Klenow--Willis Variants} \label{tab:kw_variants}
\scalebox{0.9}{
\begin{tabular}{lccccc} \toprule
    \multicolumn{1}{l}{\textbf{Parameter}} & \textbf{Description} & \textbf{Baseline} &
    \textbf{\shortstack{KW1\\broad}} &
    \textbf{\shortstack{KW2\\$(\omega,\psi)$ fixed}} &
    \textbf{\shortstack{KW3\\$(\omega,\psi)$ refit}} \\ \midrule
    $f$       & \multicolumn{1}{l}{Menu cost}                      & 0.016 & 0.027 & 0.015 & 0.023 \\
    $\psi$    & \multicolumn{1}{l}{Super-elasticity}               & --1.10 & --0.83 & --2 & -0.77 \\
    $\omega$  & \multicolumn{1}{l}{Elasticity}                     & 1.18 & 1.45 & 1.25 & 1.45 \\
    $\rho_z$  & \multicolumn{1}{l}{Persistence of $z^i_t$}         & 0.977 & 0.981 & 0.982 & 0.979 \\
    $\sigma_z$& \multicolumn{1}{l}{Std.\ dev.\ of $z^i_t$}        & 0.060 & 0.060 & 0.058 & 0.059 \\
    $\rho_\nu$  & \multicolumn{1}{l}{Persistence of $\log \nu^i_t$}         & 0.998 & -- & -- & -- \\
    $\sigma_\nu$& \multicolumn{1}{l}{Std.\ dev.\ of $\log \nu^i_t$}        & 0.030 & -- & -- & -- \\ \midrule
    \multicolumn{1}{l}{\textbf{Narrow Moments}} & \textbf{Data} &   &   &   &   \\ \midrule
    SD($z_t^i$)                 & 0.26 & \textbf{0.25} & \textbf{0.28} & \textbf{0.26} & \textbf{0.27} \\
    5-year autocorr. of $z_t^i$ & 0.31 & \textbf{0.31} & \textbf{0.32} & \textbf{0.33} & \textbf{0.30} \\
    IV demand elasticity        & --2.40 & \textbf{--2.31} & {--3.87} & {--7.11} & {--3.90} \\
    SD($\log \nu_t^i$)                 & 1.16 & \textbf{1.14} & {0.22} & {0.34} & {0.20} \\
    5-year autocorr. of $\log \nu_t^i$ & 0.62 & \textbf{0.58} & {0.08} & {0.05} & {0.05} \\
    Corr(price, TFPQ)           & --0.54 & \textbf{--0.54} & \textbf{--0.99} & {--0.97} & {--0.99} \\
    Frequency of price changes  & 0.11 & \textbf{0.11} & \textbf{0.10} & \textbf{0.11} & \textbf{0.11} \\ \midrule
    \multicolumn{6}{l}{\textbf{Broader Moments}} \\ \midrule
    Growth dispersion           & 0.39 & {0.55} & \textbf{0.33} & {0.40} & {0.33} \\
    Fraction of price increases & 0.65 & {0.55} & \textbf{0.63} & {0.64} & {0.55} \\
    Average size of price change& 0.08 & {0.06} & \textbf{0.08} & {0.04} & \textbf{0.08} \\
    SD($\Delta p$)              & 0.08 & {0.07} & \textbf{0.08} & {0.04} & {0.08} \\
    Avg. markup                 & 1.56 & {1.49} & \textbf{1.50} & {1.45} & \textbf{1.48} \\
    SD(markup)                  & 0.72 & {0.47} & \textbf{0.25} & {0.34} & {0.24} \\ \midrule
    Pass-through of supply shocks & 20\%--50\% & {0.43} & {0.45} & {0.29} & {0.47} \\ \bottomrule
\end{tabular}
}
\justify
\footnotesize{
This table reports alternative calibrations of the Kimball/menu-cost model.
Baseline is the benchmark with idiosyncratic demand shocks.
KW1 removes idiosyncratic demand shocks and recalibrates the remaining parameters to the broad set of moments used in our benchmark exercise.
KW2 fixes Kimball curvature at $(\omega,\psi)=(1.25,-2)$ and chooses remaining parameters to match the productivity-process moments and the frequency of price changes.
KW3 re-chooses $(\omega,\psi)$ to match the average size of price changes and average markups, while continuing to match frequency.
Targeted moments are boldfaced.
}
\end{table}

We start with KW1, which shuts down idiosyncratic demand shocks and recalibrates the remaining parameters to the same broad set of moments used in Section \ref{sec:Model_1B}.\footnote{Moments that require within-product demand variation, such as the persistence and dispersion of $\log \nu_t^i$ and the IV demand elasticity, are not treated as targets in KW1 because idiosyncratic demand shocks are shut down.} The goal is to assess how far a productivity-only Kimball/menu-cost model can go when confronted with the collection of pricing and firm-level moments used elsewhere in the paper. KW1 matches the idiosyncratic productivity process and several pricing moments, including the frequency of price changes, the average size of price changes, and the dispersion of $\Delta p$. At the same time, it implies Corr(price, TFPQ) essentially equal to $-1$. This conforms with Corollary~\ref{cor:corr_ps} that in a productivity-only environment, non-CES demand does not deliver a correlation between price and TFPQ meaningfully above negative unity.

KW2 restates the Klenow--Willis tension directly in observable moments. In KW2 we fix curvature at $(\omega,\psi)=(1.25,-2)$ and choose the remaining parameters to match the productivity-process moments and the frequency of price changes. This calibration matches the frequency (0.11) and the productivity-process moments, but it implies price changes that are too small. The average absolute price change is 0.04 versus 0.08 in the data, and SD($\Delta p$) is 0.04 versus 0.08. At the same time, KW2 delivers growth dispersion close to the data, 0.40 versus 0.39, even though growth dispersion is not targeted in this exercise. This indicates that the supply-side volatility embedded in the calibrated productivity process is already in the right range for an observable real-side dispersion statistic. Closing the gap in the size of price changes by increasing the volatility of idiosyncratic driving forces would therefore require departing from this disciplined productivity process, with immediate implications for observable real-side dispersion, including growth dispersion, and for the estimated productivity moments.

KW3 shows that the tension in KW2 is conditional on fixing curvature. In KW3 we allow $(\omega,\psi)$ to adjust so the productivity-only model matches the average size of price changes and the average markup, while continuing to match the frequency of price changes. Under this refit, the model reproduces the key pricing targets, but it does so with materially different curvature than in KW2; $(\omega,\psi)$ shifts from $(1.25,-2)$ to approximately $(1.45,-0.77)$. This illustrates that pricing moments alone provide weak discipline for Kimball curvature, and different combinations of curvature and other parameters can rationalize the same pricing statistics. The refit also has real-side implications, growth dispersion falls to 0.33 relative to 0.39 in the data. Finally, even when pricing moments are matched by construction, a productivity-only environment remains subject to the restriction that pushes Corr($p$, TFPQ) close to $-1$, while the empirical value is $-0.54$.

The benchmark model resolves these identification and fit issues by introducing idiosyncratic demand shocks and using micro restrictions to discipline curvature directly. With demand shocks, the IV demand elasticity provides mapped discipline for $\omega$. The correlation between price and TFPQ, together with the demand-process moments, disciplines $\psi$. Once curvature is disciplined in this way, the benchmark model matches the pricing moments in Table~\ref{tab:kw_variants} without choosing $(\omega,\psi)$ to fit them mechanically.

\section{Monetary Non-Neutrality} \label{sec:nonneutrality}

Having validated the model's ability to reconcile firm-level shocks with pricing dynamics, we now assess the degree of monetary non-neutrality generated by the calibrated model. We examine how an unanticipated nominal expenditure innovation, interpreted as a monetary policy shock, affects real output. Monetary policy is neutral if the nominal shock is fully transmitted to prices, leaving real aggregate output and consumption unchanged. If prices respond only partially, the shock has real effects.

We consider the response of real output to a one-time, unanticipated positive level shock to nominal expenditure $S$ of 0.2\%. Given that $S$ follows a trend growth rate of $\mu=0.2\%$ per month, this shock raises nominal expenditure growth to 0.4\% in the impact period.\footnote{As the model is solved non-linearly, the response of output may display non-linearity with respect to the size of the shock. We leave a systematic analysis of this nonlinearity to future work.} 

From this impulse response, we report three summary measures of non-neutrality. The first is the impact response of output as a fraction of the shock size. The second is the half-life of the impulse response, measuring the persistence of the shock's effect. Lastly, we report the cumulative impulse response (CIR) over 35 months, defined as the sum of the output response as a fraction of the shock over the horizon considered.

\begin{figure}[t!]
\centering
    \caption{Impulse Response of Real Output to a Nominal Expenditure Shock}
    \includegraphics[scale=0.5]{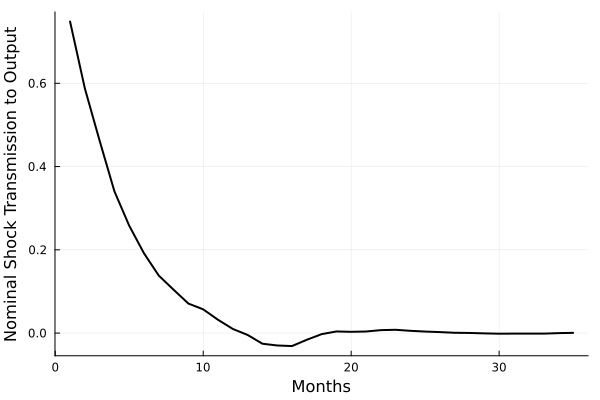}
    \label{fig:IRF_baseline}
     \justify
    \footnotesize{Note: This figure plots the impulse response of real output expressed as a fraction of the nominal expenditure shock on the vertical axis and periods elapsed since the shock on the horizontal axis.}
\end{figure}

Figure \ref{fig:IRF_baseline} displays the impulse response of real output as a fraction of the shock size for the baseline calibration reported in Table~\ref{tab:int_calib}. On impact, 75\% of the nominal expenditure increase translates into higher real output. The real effects of the shock diminish over time, with a half-life of 2.4 months, eventually dissipating after 12 months. The cumulative impulse response over the horizon considered is 2.92. The second column of Table~\ref{tab:nonneutrality_models} shows that the calibration targeting a broader set of moments (Table~\ref{tab:Model_1B}) generates a similar degree of non-neutrality, which is unsurprising given that it does not differ much from the baseline calibration.

The degree of monetary non-neutrality generated by the calibrated model lies in the upper range of values reported in the literature, as summarized by \cite{mongey2021market}. Moreover, the impact response to nominal shocks in our model is comparable to results from studies that incorporate macro real rigidities as an alternative source of real transmission. For example, \cite{nakamura2010monetary} develop a multi-sector model with production networks, introducing macro real rigidity via sticky marginal costs, and obtain an impact response of 0.80. This comparison highlights that micro real rigidity, as represented in our model, can generate a degree of non-neutrality comparable to models with macro real rigidities, while simultaneously aligning with micro-data on firm dynamics and pricing.\footnote{Table A.1 in \cite{mongey2021market} shows that menu-cost models without real rigidities that are calibrated to the U.S. economy typically generate an impact response in the range of 0.35 to 0.50. Richer models that include alternative sources of real rigidities find higher values in the 0.7--0.8 range.}

It is useful to highlight two sources of amplification in our setup relative to a simple menu-cost model \citep{golosov2007menu}. The first source is micro real rigidities, which introduce strategic complementarities in pricing via the Kimball demand system. With micro real rigidities, firms adjusting their prices tend to make smaller adjustments than they would under a CES demand system, reflecting their incentive to remain closer to the prices of competitors who are not adjusting, and this amplifies monetary non-neutrality.\footnote{This intuition is formalized in \cite{alvarez2022price}, who derive analytic results in a menu-cost model with strategic complementarities framed as a Mean Field Game. Their findings show that complementarity increases the impulse response of output to a nominal shock at every horizon. They also show that their theoretical result holds under a Calvo model of pricing.}

\begin{table}[t!]
    \centering
    \caption{Monetary Non-neutrality across Models}     \label{tab:nonneutrality_models}
    \begin{tabular}{lcccc} \toprule
                    & Baseline & Broad Moments & KW1 broad & CES  \\ \midrule
    Impact response & 0.75  & 0.69  & 0.64 & 0.65  \\
    Half-life (months) & 2.4 & 3.3 & 2.4 & 2.3 \\
    CIR             & 2.92  & 3.16  & 2.04 & 2.18  \\ \bottomrule
    \end{tabular}
    \justify
    \footnotesize{The first column refers to the baseline calibrated model as summarized in Table~\ref{tab:int_calib}. The second column refers to the calibration targeting a broader set of moments as presented in Table~\ref{tab:Model_1B}. The third column refers to the model with Kimball demand but no idiosyncratic demand shocks, described in Table~\ref{tab:kw_variants}. The last column refers to a model with CES demand and both idiosyncratic productivity and demand shocks, calibrated to the persistence and dispersion of idiosyncratic productivity and demand processes as documented in \cite{foster2008reallocation} as well as the frequency of price changes.}
\end{table}

We quantify the role of micro real rigidities in monetary non-neutrality by comparing our baseline model to a CES menu-cost model following \cite{golosov2007menu} that incorporates both idiosyncratic productivity and demand shocks. We calibrate this counterfactual model to the same firm-dynamics properties from \cite{foster2008reallocation} used in the baseline, alongside the observed frequency of price changes.\footnote{Calibration details are provided in Appendix \ref{app:ces_models}.} The last column of Table~\ref{tab:nonneutrality_models} summarizes the resulting non-neutrality, which has an impact response of 0.65, half-life of 2.3 months, and cumulative response of 2.18. Since the two models differ solely in their demand systems, this divergence reflects amplification driven by micro real rigidities via strategic pricing complementarity.

The second factor that affects the degree of monetary non-neutrality is the strength of selection in price adjustment. The real response to nominal shocks depends not only on how many prices adjust but also on which prices adjust. In menu-cost models, selection is strong, the firms that incur the fixed costs of price adjustment are those with the largest desired price changes. The introduction of idiosyncratic shocks weakens this selection effect, as firms respond to both aggregate shocks and disturbances in their individual states. \cite{nakamura2010monetary} show that, holding the frequency of price changes constant, more volatile idiosyncratic shocks lead to weaker selection and stronger non-neutrality.

In our baseline model, firms are subject to two orthogonal idiosyncratic shocks, which weakens selection. To quantify the impact of this force, we compare the baseline model with one without idiosyncratic demand shocks. Specifically, we compute the impulse response for the model with Kimball demand but no idiosyncratic demand shocks, calibrated to the same broad sets of moments used throughout the analysis.\footnote{This calibration is discussed in Section~\ref{sec:kw} and is labelled ``KW1 broad'' in Table~\ref{tab:kw_variants}.} The third column in Table \ref{tab:nonneutrality_models} reports non-neutrality statistics for this model. Moving from this model to one that includes idiosyncratic demand, both calibrated to the same broad set of moments, increases the impact response by roughly 8\%, from 0.64 to 0.69, and the cumulative response by 55\%, from 2.04 to 3.16.

Given the centrality of micro real rigidities in amplifying monetary non‑neutrality, it is natural to ask whether this channel could be equally important under other forms of nominal frictions. To this end, we consider the framework of \cite{calvo1983staggered}, due to its prevalence as a modelling tool of nominal rigidities in the macroeconomics literature. In the Calvo setup, firms can adjust prices only with an exogenous probability in each period. Because this probability is independent of the size and sign of the desired price change, price adjustment is purely time dependent rather than state dependent. As a result, there is no selection -- firms with highly misaligned prices are no more likely to adjust than firms whose prices are close to their desired levels. The lack of selection implies substantially stronger real effects of monetary shocks. Indeed, \cite{alvarez2016real} show that, when both models match the same  micro evidence on price changes, monetary non‑neutrality in a Calvo model is six times larger than in the Golosov–Lucas menu cost benchmark.

\begin{table}[t!]
    \centering
    \caption{Monetary Non-neutrality: Calvo}     \label{tab:calvo_nm}
    \begin{tabular}{ccccc} \toprule
                    & Baseline & Calvo 1 & Calvo 2 & Calvo 3  \\ \midrule
    Impact Response & 0.75  & 0.89  & 0.94 & 0.97  \\
    CIR             & 2.92  & 8.09  & 8.36 & 20.03  \\ \bottomrule
    \end{tabular}
    \justify
    \footnotesize{Calvo 1 refers to the theoretical impulse response from a Calvo model with CES demand and no idiosyncratic shocks. Calvo 2 refers to a Calvo model with CES demand $\psi=0$ and idiosyncratic demand and productivity shocks calibrated to match the five-yearly autocorrelation and dispersion of TFPQ and demand reported in \cite{foster2008reallocation}. Calvo 3 refers to a Calvo model with the same calibration as the baseline model which includes Kimball demand. All three models set the probability of price adjustment to be 0.11.}
\end{table}

We compare the degree of non-neutrality under three Calvo pricing models, which we present in Table~\ref{tab:calvo_nm}. In all three models, we set the probability of adjustment to 0.11, matching the empirical price adjustment frequency. We first study a CES model with Calvo pricing and no idiosyncratic shocks to firms. As we show analytically in Appendix~\ref{app:calvo}, this economy -- which we label Calvo 1 -- has an impact response of 0.89 and a CIR of 8.09. Calvo 2 then adds idiosyncratic TFPQ and demand shocks, with the parameters $(\rho_z,\sigma_z,\rho_{\nu},\sigma_{\nu}$) calibrated to match the four empirical moments pertaining to firm productivity and demand in Table~\ref{tab:int_calib}. The addition of idiosyncratic shocks slightly raises the degree of non-neutrality, with the impact response increasing to 0.94 and the CIR to 8.36. While firms that adjust are still chosen at random, the size and direction of their adjustments are influenced by both the aggregate nominal expenditure shock and idiosyncratic disturbances. In comparison, price adjustments are all uniform -- responding to the same aggregate shock -- in the model without idiosyncratic shocks. Lastly, we introduce micro real rigidities through Kimball demand in Calvo 3, which uses the same calibration as the baseline model in Section~\ref{sec:calibration}, replacing menu cost with the Calvo adjustment probability. In the presence of micro real rigidities, the degree of non-neutrality is substantially higher, with an impact response of 0.97 and CIR of around 20 -- more than double that of CES models. This shows that the amplification role of micro real rigidities is important regardless of the way nominal rigidities are modelled.

\section{Conclusion}  \label{sec:conclusion}

This paper revisits micro real rigidities as a source of monetary non-neutrality. The central insight in \citet{ball1990real} is not tied to any particular nominal friction. Whether prices are sticky because of menu costs, Calvo opportunities, or informational frictions, large real effects of nominal disturbances require forces that make firms want to adjust relative prices only modestly in response to aggregate conditions. In much of quantitative monetary economics, those forces have been sought in macro-level mechanisms such as sticky marginal costs, networks, or wage rigidity. By contrast, demand-based micro real rigidities, through variable markups, incomplete pass-through, and strategic complementarities in pricing, are workhorse ingredients in trade and IO, but have been used more sparingly in quantitative monetary transmission, in part because of concerns about empirical discipline and quantitative fit.

Our contribution is to show that a disciplined, demand-based micro real rigidities  framework can be quantitatively viable and portable across nominal-friction environments. We do so by combining three elements. First, we develop a simple price-setting model  that isolates how deviations from CES generate variable markups and incomplete pass-through, and why idiosyncratic demand shocks matter for desired prices only when they shift the residual-demand conditions at which local elasticities are evaluated. This yields clear, data-relevant restrictions linking the joint behavior of prices, quantities, and productivity to the key demand primitives. Second, we bring these restrictions to the data by disciplining the demand system and the firm-level demand and productivity processes using firm-level evidence that separates demand from technology. This step turns objects that are often treated as free parameters, the elasticity level, demand curvature, and the magnitude and persistence of idiosyncratic driving forces, into empirically grounded inputs. Third, we embed the disciplined demand system in a quantitative model with nominal price stickiness and show that it matches core micro pricing moments and firm-dynamics patterns largely out of sample, while generating substantial monetary non-neutrality.

The broader implication is that an empirically disciplined demand and markup framework  can serve as a bridge between literatures that are often studied separately. In monetary economics, it provides a tractable way to incorporate variable markups and incomplete pass-through into quantitative models of inflation dynamics and policy, and it delivers substantial monetary non-neutrality in a simple sticky-price environment once disciplined by firm-level evidence. In trade and IO, it offers a mapping from firm-level evidence on demand, markups, and pass-through to a demand system that can be used in macro applications. In firm dynamics, it provides a disciplined decomposition of demand and productivity that can be taken to models with richer real margins.

Several avenues for future research follow naturally from this bridge. One is to integrate the disciplined pricing and demand framework into life-cycle and firm-dynamics environments with entry and exit, investment, and workforce adjustment, to study how monetary shocks and pricing frictions interact with firm growth, selection, and reallocation. A second is to combine it with financial frictions, for example collateral constraints or working-capital needs, to understand how financing conditions affect pricing, pass-through, and the distribution of markups, and how those interactions shape monetary transmission. A third is to use the framework in open-economy settings where exchange-rate movements, imported-input costs, and global demand interact with variable markups, allowing a unified quantitative treatment of pass-through and monetary policy across domestic and international margins. More broadly, the empirical discipline developed here can be applied in any setting where prices and quantities can be used to separate demand from technology, opening the door to bringing micro evidence on demand, markups, and firm shocks into a wider class of macroeconomic models.


\section*{References}
\bibliographystyle{aernobold}
\vspace{-10mm}
\renewcommand\refname{}
\bibliography{kimball_reference}{}


\appendix

\clearpage

\renewcommand*\theequation{A-\arabic{equation}}
\setcounter{equation}{0}
\renewcommand*\thepage{A-\arabic{page}}
\setcounter{page}{1}
\renewcommand*\thetable{A-\arabic{table}}
\setcounter{table}{0}
\renewcommand*\thefigure{A-\arabic{figure}}
\setcounter{figure}{0}
\renewcommand*\thefootnote{A-\arabic{footnote}}
\setcounter{footnote}{0}

{\Large \bf Internet Appendix (For Online Publication) }

\section{Details for Section~\ref{sec:simple}} \label{app:sec2_details}

This appendix contains the derivations and proofs referenced in Section~\ref{sec:simple}.

\subsection{Real Rigidity and Strategic Complementarity} \label{app:rigidity_sc}

\paragraph{Real rigidity measure.}
Totally differentiating the firm’s first-order condition in the frictionless problem yields the local responsiveness of the desired relative price to aggregate demand,
\begin{equation}
\phi \;\equiv\; \dfrac{\partial \left (\dfrac{p_i^*}{P}\right )}{\partial \left ( \dfrac{S}{P} \right )}
\;=\; -\,\frac{\Pi_{12}\!\left(\frac{p_i^*}{P},\frac{S}{P},A_i\right)}{\Pi_{11}\!\left(\frac{p_i^*}{P},\frac{S}{P},A_i\right)}. \label{eq:app_SC1}
\end{equation}
With \(\Pi_{11}<0\) (concavity) and \(\Pi_{12}>0\) (local stability), we have \(\phi>0\).

\paragraph{Link to strategic complementarity.}
Define 
\[\zeta \equiv \dfrac{\partial \ln p_i^*}{\partial \ln P}.\] 
In a symmetric equilibrium with \(p_i^*=P\) and \(Y\equiv S/P\),
\begin{equation}
\frac{\partial \ln \left(\dfrac{p_i^*}{P} \right )}{\partial \ln \left ( \dfrac{S}{P} \right)} \;=\; 1-\zeta.  \label{eq:app_SC2}
\end{equation}
Hence, stronger real rigidity (smaller responsiveness of \(p_i^*/P\) to \(S/P\)) corresponds to stronger strategic complementarity in pricing (larger \(\zeta\)).

\subsection{Firm’s Problem and Lerner Pricing (Proof of Proposition~\ref{prop:markup})} \label{app:lerner_proof}

The firm chooses \(p_i\) to maximize \(\Pi(p_i/P,\;S/P,\;A_i)\). Let \(D(p_i/P;\,\cdot)\) denote residual demand for the firm (with aggregate objects taken as given by an atomistic firm) and $C\left(D\right)$ denote the cost of producing $D$ units. Profit is written as \(\pi_i=p_iD\left(p_i/P;\,\cdot\right) - C\left(D\left(p_i/P;\,\cdot\right)\right)\). The first-order condition is
\[
0 \;=\; D\left(\dfrac{p_i}{P}; \cdot\right) + p_i \cdot \frac{\partial D}{\partial p_i} - C'\left(D\left(\dfrac{p_i}{P}; \cdot\right)\right) \cdot \frac{\partial D}{\partial p_i} .
\]
Define the (absolute) residual-demand elasticity \(\sigma(\varphi_i)>1\) at the operating point \(\varphi_i\) by
\[
\sigma(\varphi_i) \;\equiv\; -\,\frac{\partial \ln (D)}{\partial \ln \left (\dfrac{p_i}{P} \right )}.
\]
Plugging 
\[\dfrac{\partial D}{\partial p_i} = -\dfrac{\sigma(\varphi_i)}{p_i}\,D\] into the FOC gives
\[
0 \;=\; D\left(\dfrac{p_i}{P}; \cdot\right) + \left(p_i- C'(D) \right)  \left(-\dfrac{\sigma(\varphi_i)}{p_i}\,D\right)
\quad\Rightarrow\quad
p_i^* \;=\; \frac{\sigma(\varphi_i)}{\sigma(\varphi_i)-1}\,C'(D) \;\equiv\; \mu(\varphi_i)\,C'(D),
\]
which proves Proposition~\ref{prop:markup}.



\subsection{Markup–Curvature Mapping and Generic Pass-through} \label{app:generic_pt}

Let 
\[\mu(\varphi)=\dfrac{\sigma(\varphi)}{\sigma(\varphi)-1}.\]
Differentiating,
\[
\frac{\partial \ln \mu(\varphi)}{\partial \ln \varphi}
\;=\; -\,\frac{1}{\sigma(\varphi)-1}\,\frac{d\ln \sigma(\varphi)}{d\ln \varphi}.
\]

\paragraph{Cost pass-through.}
From \(\ln p_i^*=\ln \mu(\varphi_i)+\ln mc_i\),

\begin{eqnarray*}
    \frac{\partial\ln p_{i}^{*}}{\partial\ln mc_{i}}	&=&\frac{\partial\ln\mu\left(\varphi_{i}\right)}{\partial\ln\varphi_{i}}\frac{\partial\ln\varphi_{i}}{\partial\ln mc_{i}^{*}}+1 \\
\frac{\partial\ln p_{i}^{*}}{\partial\ln mc_{i}}	&=&\left[\frac{-1}{\sigma\left(\varphi_{i}\right)-1}\frac{d\ln\sigma\left(\varphi_{i}\right)}{d\ln\varphi_{i}}\right]\frac{\partial\ln\varphi_{i}}{\partial\ln p_{i}^{*}}\frac{\partial\ln p_{i}^{*}}{\partial\ln mc_{i}}+1 \\
\frac{\partial\ln p_{i}^{*}}{\partial\ln mc_{i}}	&=&\dfrac{1}{1+\dfrac{1}{\sigma\left(\varphi_{i}\right)-1}\dfrac{d\ln\sigma\left(\varphi_{i}\right)}{d\ln\varphi_{i}}\dfrac{\partial\ln\varphi_{i}}{\partial\ln p_{i}^{*}}}
\end{eqnarray*}

Under CES, \(d\ln\sigma/d\ln\varphi=0\), so \(\partial \ln p_i^*/\partial \ln mc_i=1\).

\paragraph{Demand pass-through.}
For an idiosyncratic shifter \(\nu_i\) that can move \(\varphi_i\),
\begin{eqnarray*}
    \frac{\partial\ln p_{i}^{*}}{\partial\ln\nu_{i}}	&=&\frac{\partial\ln\mu\left(\varphi_{i}\right)}{\partial\ln\varphi_{i}}\frac{\partial\ln\varphi_{i}}{\partial\ln\nu_{i}} \\
\frac{\partial\ln p_{i}^{*}}{\partial\ln\nu_{i}}	&=&\frac{\partial\ln\mu\left(\varphi_{i}\right)}{\partial\ln\varphi_{i}}\left[\left.\frac{\partial\ln\varphi_{i}}{\partial\ln\nu_{i}}\right|_{p_{i}^{*}}+\frac{\partial\ln\varphi_{i}}{\partial\ln p_{i}^{*}}\frac{\partial\ln p_{i}^{*}}{\partial\ln\nu_{i}}\right] \\
\frac{\partial\ln p_{i}^{*}}{\partial\ln\nu_{i}}	&=&\frac{\partial\ln\mu\left(\varphi_{i}\right)}{\partial\ln\varphi_{i}}\left.\frac{\partial\ln\varphi_{i}}{\partial\ln\nu_{i}}\right|_{p_{i}^{*}}+\frac{\partial\ln\mu\left(\varphi_{i}\right)}{\partial\ln\varphi_{i}}\frac{\partial\ln\varphi_{i}}{\partial\ln p_{i}^{*}}\frac{\partial\ln p_{i}^{*}}{\partial\ln\nu_{i}} \\
\frac{\partial\ln p_{i}^{*}}{\partial\ln\nu_{i}}	&=&\frac{\dfrac{\partial\ln\mu\left(\varphi_{i}\right)}{\partial\ln\varphi_{i}}\left.\dfrac{\partial\ln\varphi_{i}}{\partial\ln\nu_{i}}\right|_{p_{i}^{*}}}{1-\dfrac{\partial\ln\mu\left(\varphi_{i}\right)}{\partial\ln\varphi_{i}}\dfrac{\partial\ln\varphi_{i}}{\partial\ln p_{i}^{*}}} \\
\frac{\partial\ln p_{i}^{*}}{\partial\ln\nu_{i}}&=&\frac{\dfrac{-1}{\sigma\left(\varphi_{i}\right)-1}\dfrac{d\ln\sigma\left(\varphi_{i}\right)}{d\ln\varphi_{i}}\left.\dfrac{\partial\ln\varphi_{i}}{\partial\ln\nu_{i}}\right|_{p_{i}^{*}}}{1+\dfrac{1}{\sigma\left(\varphi_{i}\right)-1}\dfrac{d\ln\sigma\left(\varphi_{i}\right)}{d\ln\varphi_{i}}\dfrac{\partial\ln\varphi_{i}}{\partial\ln p_{i}^{*}}}
\end{eqnarray*}


If \(\partial \ln \varphi_i/\partial \ln \nu_i=0\) (level-only placement), demand pass-through is zero. 
CES demand (\(d\ln\sigma/d\ln\varphi=0\)) also implies zero demand pass-through.

\subsection{Price–Productivity Correlation (Proof of Corollary~\ref{cor:corr_ps})} \label{app:corr_ps_appendix}

Write \(\widehat{p}_i^{\,*}=\alpha\,\widehat{z}_i+\beta\,\widehat{\nu}_i+\text{const}\) and \(\widehat{\mathrm{TFP}}_i=-\widehat{z}_i\). Let \(\sigma_z^2=\mathrm{Var}(\widehat{z}_i)\), \(\sigma_\nu^2=\mathrm{Var}(\widehat{\nu}_i)\), and \(\rho=\mathrm{Corr}(\widehat{z}_i,\widehat{\nu}_i)\). 
Then
\[
\mathrm{Cov}\big(\widehat{p}_i^{\,*},\widehat{\mathrm{TFP}}_i\big)
=-\alpha\,\sigma_z^2-\beta\,\rho\,\sigma_z\sigma_\nu,
\quad
\mathrm{Var}\big(\widehat{p}_i^{\,*}\big)
=\alpha^2\sigma_z^2+\beta^2\sigma_\nu^2+2\alpha\beta\rho\,\sigma_z\sigma_\nu,
\quad
\mathrm{Var}\big(\widehat{\mathrm{TFP}}_i\big)=\sigma_z^2.
\]
Hence
\begin{equation} \label{eq:corr_p_TFP_expression}
\mathrm{Corr}\!\left(\widehat{p}_i^{\,*},\,\widehat{\mathrm{TFP}}_i\right)=-\,\frac{\alpha\,\sigma_z+\beta\,\rho\,\sigma_\nu}
{\sqrt{\alpha^2\sigma_z^2+\beta^2\sigma_\nu^2+2\alpha\beta\rho\,\sigma_z\sigma_\nu}}.
\end{equation}

When $\beta=0$ -- that is if desired price does not respond to changes in idiosyncratic demand, $\mathrm{Corr}\!\left(\widehat{p}_i^{\,*},\,\widehat{\mathrm{TFP}}_i\right)= -1$.

When $\beta>0$ and $\vert\rho\vert<1$, we can show that $\mathrm{Corr}\!\left(\widehat{p}_i^{\,*},\,\widehat{\mathrm{TFP}}_i\right) \in (-1,1)$.
To see this, observe that the square of the numerator in Equation~\ref{eq:corr_p_TFP_expression} is strictly smaller than the square of the denominator
\[
\left(\alpha^2\sigma_z^2+\beta^2\sigma_\nu^2+2\alpha\beta\rho\,\sigma_z\sigma_\nu\right)-\left(\alpha \sigma_z + \beta \rho \sigma_{\nu}\right)^2 =\beta^2 \sigma_\nu^2\left(1-\rho^2\right) > 0.
\]
This also shows that under $\beta>0$, the correlation equals $-1$ exactly only when $\vert \rho\vert=1$, i.e. idiosyncratic demand and productivity are perfectly correlated.




\subsection{Kimball Demand: Aggregator, Elasticities, and Closed Forms} \label{app:kimball}

\subsubsection{Aggregator and inverse residual demand} \label{app:kimball_agg}
Consider the Kimball aggregator
\[
1=\int_0^1 G(x_i)\,di,\qquad x_i \equiv \tilde{\varphi}_i \equiv \nu_i\,\frac{y_i}{Y},
\]
with \(G'(x)>0\) and \(G''(x)<0\). Expenditure minimization implies inverse residual demand
\[
\frac{p_i}{P} \;=\; \Lambda\,G'(x_i),\qquad \Lambda>0.
\]
We adopt the parametric form \eqref{eq:kimball_spec}
\[
G(x)=\frac{\omega}{1+\omega\psi}\Big[(1+\psi)x-\psi\Big]^{\frac{1+\omega\psi}{\omega(1+\psi)}}+1-\frac{\omega}{1+\omega\psi},
\qquad \omega>1,\quad \psi\in\mathbb{R},
\]
which has the residual demand
\[
\frac{\nu_{i}y_{i}}{Y}=\frac{1}{1+\psi}\left[\left(\frac{p_{i}}{\Lambda \nu_{i}P}\right)^{\frac{\omega \left(1+\psi\right)}{1-\omega}}+\psi\right].
\]

\subsubsection{Elasticity and super-elasticity} \label{app:kimball_elast}
At the effective share \(x_i=\tilde{\varphi}_i\), the price elasticity $\epsilon_i$ and super-elasticity $\gamma_i$ are
\[
\epsilon_i \;\equiv\; \frac{\partial y_i}{\partial p_i}\frac{p_i}{y_i}  \;=\; \frac{\omega(1+\psi)}{1-\omega}\,
\frac{\left(\dfrac{p_i}{\Lambda\,\nu_i\,P}\right)^{\frac{\omega(1+\psi)}{1-\omega}}}{\left(\dfrac{p_i}{\Lambda\,\nu_i\,P}\right)^{\frac{\omega(1+\psi)}{1-\omega}}+\psi} \;=\; \frac{\omega}{1-\omega}\frac{(1+\psi)\tilde{\varphi_i}-\psi}{\tilde{\varphi_i}},
\]
\[
\gamma_i \;\equiv\; \frac{\partial \epsilon_i}{\partial p_i}\frac{p_i}{\epsilon_i}
 \;=\; \frac{\omega}{1-\omega}\,\frac{\psi}{\tilde{\varphi}_i} \;=\; \frac{\omega}{1-\omega}\frac{\psi(1+\psi)}{\left(\dfrac{p_i}{\Lambda \nu_i P}\right)^{\frac{\omega}{1-\omega}}+\psi}.
\]
When \(\psi=0\), the model collapses to CES with constant elasticity \(\omega/(\omega-1)\).

\subsubsection{Log-linear pricing and closed-form pass-through} \label{app:kimball_pt}
Let \(\widehat{z}_i\) and \(\widehat{\nu}_i\) denote log deviations of productivity and the idiosyncratic demand shifter, respectively. Linearizing the optimal pricing condition around a symmetric steady state with \(\nu_i=1\) yields
\[
\widehat{p}_i^{\,*}
\;=\;
\frac{\omega\psi}{\omega\psi-1}\big(\widehat{\Lambda}+\widehat{P}+\widehat{\nu}_i\big)
\;+\;
\frac{1}{\omega\psi-1}\,\widehat{z}_i.
\]
Hence the partial pass-through elasticities are
\[
\frac{\partial \widehat{p}_i^{\,*}}{\partial \widehat{z}_i}
=\frac{1}{\omega\psi-1},
\qquad
\frac{\partial \widehat{p}_i^{\,*}}{\partial \widehat{\nu}_i}
=\frac{\omega\psi}{\omega\psi-1}.
\]
Equivalently, with \(\widehat{mc}_i\equiv -\,\widehat{z}_i\),
\[
\frac{\partial \widehat{p}_i^{\,*}}{\partial \widehat{mc}_i}
= -\,\frac{1}{\omega\psi-1}.
\]
In the CES limit \(\psi=0\), these specialize to \(\partial \widehat{p}_i^{\,*}/\partial \widehat{z}_i=-1\), \(\partial \widehat{p}_i^{\,*}/\partial \widehat{\nu}_i=0\), and \(\partial \widehat{p}_i^{\,*}/\partial \widehat{mc}_i=1\).






\subsection{Parallels Outside HDIA: HIIA } \label{app:hiia}

In the HIIA class (homothetic indirect implicit additivity), direct residual demand takes the form \(\varphi=g(x)\) with \(x\equiv p_i/P\) and the price index \(P\) fixed by duality. The residual own-price elasticity is
\[
\sigma(x) \;=\; -\,\frac{d\ln g(x)}{d\ln x} \;=\; -\,\frac{x\,g'(x)}{g(z)}.
\]
The pass-through logic is parallel: under CES (\(\sigma\) constant) cost pass-through is one and demand pass-through is zero; when \(\sigma\) varies with \(z\), cost pass-through falls below one (with downward-sloping demand) and demand pass-through is nonzero if the shifter moves the argument of \(\sigma\).

\subsection{Cost-side Micro Real Rigidities} \label{app:cost_side_theory}

Let marginal cost depend on output, \(mc_i = \tilde{z}_i\,m(y_i)\), where \(\tilde{z}_i\) is an efficiency index with \(m'(\cdot)>0\). Define cost-side curvature as the elasticity of marginal cost to quantity \(\kappa(y_i)\equiv \partial \ln mc_i/\partial \ln y_i = \partial \ln m(y_i)/\partial \ln y_i\). 

From Proposition~\ref{prop:markup}, the firm’s optimal pricing rule is
\[
p_i^*=\mu(\varphi_i)\,mc_i=\mu(\varphi_i)\,\tilde{z}_i\,m(y_i),
\]
where 
\[\mu(\varphi_i)=\dfrac{\sigma(\varphi_i)}{\sigma(\varphi_i)-1}.\]
The elasticity of optimal markup to the operating point is then,
\[
\frac{\partial \ln \mu(\varphi)}{\partial \ln \varphi}
\;=\; -\,\frac{1}{\sigma(\varphi)-1}\,\frac{d\ln \sigma(\varphi)}{d\ln \varphi}.
\]

\paragraph{Cost pass-through under increasing marginal cost.}

From \(\ln p_i^*=\ln \mu(\varphi_i)+\ln mc_i\), we get the analog of Equation~\eqref{eq:section2_cost_pt}

\begin{eqnarray*}
    \frac{\partial\ln p_{i}^{*}}{\partial\ln\tilde{z}_{i}}&=&1+\frac{\partial\ln\mu\left(\varphi_{i}\right)}{\partial\ln\varphi_{i}}\frac{\partial\ln\varphi_{i}}{\partial\ln p_{i}^{*}}\frac{\partial\ln p_{i}^{*}}{\partial\ln\tilde{z}_{i}}+\frac{\partial\ln m\left(y_{i}\right)}{\partial\ln y_{i}}\frac{\partial\ln y_{i}}{\partial\ln p_{i}^{*}}\frac{\partial\ln p_{i}^{*}}{\partial\ln\tilde{z}_{i}}\\\frac{\partial\ln p_{i}^{*}}{\partial\ln\tilde{z}_{i}}&=&1+\left[\frac{-1}{\sigma\left(\varphi_{i}\right)-1}\frac{d\ln\sigma\left(\varphi_{i}\right)}{d\ln\varphi_{i}}\right]\frac{\partial\ln\varphi_{i}}{\partial\ln p_{i}^{*}}\frac{\partial\ln p_{i}^{*}}{\partial\ln\tilde{z}_{i}}+\frac{\partial\ln m\left(y_{i}\right)}{\partial\ln y_{i}}\frac{\partial\ln y_{i}}{\partial\ln p_{i}^{*}}\frac{\partial\ln p_{i}^{*}}{\partial\ln\tilde{z}_{i}}\\\dfrac{\partial\ln p_{i}^{*}}{\partial\ln\tilde{z}_{i}}&=&\dfrac{1}{1+\dfrac{1}{\sigma\left(\varphi_{i}\right)-1}\dfrac{d\ln\sigma\left(\varphi_{i}\right)}{d\ln\varphi_{i}}\dfrac{\partial\ln\varphi_{i}}{\partial\ln p^*_{i}}-\kappa\left(y_{i}\right)\dfrac{\partial\ln y_{i}}{\partial\ln p_{i}^{*}}}
\end{eqnarray*}

Even under CES demand (\(d\ln\sigma(\varphi_i)/d\ln\varphi_i=0\)), convex costs (\(\kappa\left(y_i\right)\)) generate incomplete pass through of productivity to desired price given downward-sloping demand (\(\partial \ln y_i/\partial \ln p^*_i<0\)).

\paragraph{Demand pass-through under increasing marginal cost.}

\begin{eqnarray*}
\dfrac{\partial\ln p_{i}^{*}}{\partial\ln\nu_{i}}	&=&	\dfrac{\partial\ln\mu\left(\varphi_{i}\right)}{\partial\ln\varphi_{i}}\dfrac{\partial\ln\varphi_{i}}{\partial\ln\nu_{i}}+\dfrac{\partial\ln m\left(y_{i}\right)}{\partial\ln y_{i}}\dfrac{\partial\ln y_{i}}{\partial\ln\nu_{i}} \\
\dfrac{\partial\ln p_{i}^{*}}{\partial\ln\nu_{i}}	&=&	\dfrac{\partial\ln\mu\left(\varphi_{i}\right)}{\partial\ln\varphi_{i}}\dfrac{\partial\ln\varphi_{i}} {\partial\ln\nu_{i}}+\kappa\left(y_{i}\right)\dfrac{\partial\ln\varphi_{i}}{\partial\ln\nu_{i}} \\
\dfrac{\partial\ln p_{i}^{*}}{\partial\ln\nu_{i}}	&=&	\left[\dfrac{\partial\ln\mu\left(\varphi_{i}\right)}{\partial\ln\varphi_{i}}+\kappa\left(y_{i}\right)\right]\left[\left.\dfrac{\partial\ln\varphi_{i}}{\partial\ln\nu_{i}}\right|_{p_{i}^{*}}+\dfrac{\partial\ln\varphi_{i}}{\partial\ln p_{i}^{*}}\dfrac{\partial\ln p_{i}^{*}}{\partial\ln\nu_{i}}\right] \\
\dfrac{\partial\ln p_{i}^{*}}{\partial\ln\nu_{i}}	&=&	\left[\dfrac{\partial\ln\mu\left(\varphi_{i}\right)}{\partial\ln\varphi_{i}}+\kappa\left(y_{i}\right)\right]\left.\dfrac{\partial\ln\varphi_{i}}{\partial\ln\nu_{i}}\right|_{p_{i}^{*}}+\left[\dfrac{\partial\ln\mu\left(\varphi_{i}\right)}{\partial\ln\varphi_{i}}+\kappa\left(y_{i}\right)\right]\dfrac{\partial\ln\varphi_{i}}{\partial\ln p_{i}^{*}}\dfrac{\partial\ln p_{i}^{*}}{\partial\ln\nu_{i}} \\
\dfrac{\partial\ln p_{i}^{*}}{\partial\ln\nu_{i}}	&=&	\dfrac{\left[\dfrac{\partial\ln\mu\left(\varphi_{i}\right)}{\partial\ln\varphi_{i}}+\kappa\left(y_{i}\right)\right]\left.\dfrac{\partial\ln\varphi_{i}}{\partial\ln\nu_{i}}\right|_{p_{i}^{*}}}{1+\left[\dfrac{\partial\ln\mu\left(\varphi_{i}\right)}{\partial\ln\varphi_{i}}+\kappa\left(y_{i}\right)\right]\dfrac{\partial\ln\varphi_{i}}{\partial\ln p_{i}^{*}}} \\
\dfrac{\partial\ln p_{i}^{*}}{\partial\ln\nu_{i}}&=&\dfrac{\left[\kappa\left(y_{i}\right)-\dfrac{1}{\sigma\left(\varphi_{i}\right)-1}\dfrac{d\ln\sigma\left(\varphi_{i}\right)}{d\ln\varphi_{i}}\right]\left.\dfrac{\partial\ln\varphi_{i}}{\partial\ln\nu_{i}}\right|_{p_{i}^{*}}}{1-\left[\kappa\left(y_{i}\right)-\dfrac{1}{\sigma\left(\varphi_{i}\right)-1}\dfrac{d\ln\sigma\left(\varphi_{i}\right)}{d\ln\varphi_{i}}\right]\dfrac{\partial\ln\varphi_{i}}{\partial\ln p_{i}^{*}}}
\end{eqnarray*}

Given that idiosyncratic demand moves the operating point directly (\(\partial \ln\varphi_i/\partial \nu_i\)), demand pass-through to prices is positive if marginal cost is increasing (\(\kappa(y_i)>0\)) even if demand is CES (\(d\ln \sigma(\varphi_i) / d\ln \varphi_i=0\)).

\section{Quantitative Model} \label{app:model}

\subsection{Model Details}

\subsubsection{Household} \label{app:app_household}

A representative household supplies labor,  $h_t$, to firms in exchange for wage payments, purchases a complete set of Arrow-Debreu securities, $\mathbf{B_{t+1}}$, and consumes a final good, $C_t$. It also owns all firms in the economy and receives all accrued profits. The representative household solves the following problem
\begin{equation}
    \max \limits_{C_t,h_t,\mathbf{B_{t+1}}} \quad \mathbb{E}_{0} \sum_{t=0}^{\infty} \beta^t \left[ \log\left(C_t\right) - \chi h_t \right]
\end{equation}
subject to the budget constraint
\begin{equation} \label{eq:HH_budget}
    P_t C_t + \mathbf{Q_{t}} \cdot \mathbf{B_{t+1}} \leq B_t + W_t h_t + D_t,  
\end{equation}
where  $\mathbf{Q_{t}}$ is a vector that contains the prices of the state-contingent securities, $\mathbf{B_{t+1}}$. $B_t$ represents the payoff of the state-contingent security purchased in period $t-1$ that had a non-zero payoff in period $t$. $P_t$ and $W_t$ are the price of the final good and nominal wage, respectively, both of which are taken as given by the households. $D_t$ denotes the net dividends the household receives from the producers. 

Using recusrive notation, household optimality requires
\begin{equation}
\dfrac{W}{P} = \chi C, \label{eq:HH_labor} \\    
\end{equation}
and we can also define the household's stochastic discount factor as  
\begin{equation} \label{eq:HH_SDC}
    \Xi \equiv \beta \mathbb{E}\left (\dfrac{C}{C'} \right ).
\end{equation}

\subsubsection{Final-Good Producer} \label{app:app_final_producer}

Given the Kimball aggregator in \eqref{eq:kimball_spec}, solving the final producer's cost-minimization problem yields the following demand function for each variety $i$
\begin{equation} \label{eq:kimball_demand}
\frac{\nu_{i}y_{i}}{Y}=\frac{1}{1+\psi}\left[\left(\frac{p_{i}}{\Lambda \nu_{i}P}\right)^{\frac{\omega \left(1+\psi\right)}{1-\omega}}+\psi\right],
\end{equation}
provided that $y_i > 0 $, where $\Lambda$ is the Lagrangian multiplier on the aggregator $G$ in the cost-minimization problem and $P$ is the aggregate price index. Throughout we use $\Sigma$ to  denote the mass of non-producing varieties ($y_i = 0$) and define
\[
B=\big[1-\Sigma\,(-\psi)^{\frac{1+\omega\psi}{\omega(1+\psi)}}\big]^{-1}
\]

The Lagrange multiplier $\Lambda$ is obtained by substituting the final-good producer's first-order condition into the Kimball aggregator
\begin{equation} \label{eq:lambda}
\Lambda=\left[B\int_{y_i>0}\!\left(\frac{p_i}{\nu_i P}\right)^{\frac{1+\omega\psi}{1-\omega}}di\right]^{\frac{1-\omega}{1+\omega\psi}}
\end{equation}

The ideal price index for the final-good is derived from the zero-profit condition
\begin{equation} \label{eq:price_index}
P=\frac{B^{-\frac{\omega(1+\psi)}{1+\omega\psi}}}{1+\psi}
 \left[\int_{y_i>0}\!\left(\frac{p_i}{\nu_i}\right)^{\frac{1+\omega\psi}{1-\omega}}di\right]^{\frac{1-\omega}{1+\omega\psi}}
 +\frac{\psi}{1+\psi}\int_{y_i>0}\!\frac{p_i}{\nu_i}\,di,
\end{equation}

\subsubsection{Intermediate-Good Producers} \label{app:app_int_producer}

Given the demand schedule for individual varieties, the intermediate producers' gross profit when they charge price $p^i$ is 
\begin{equation}
    \pi(p^i,z^i,\nu^i,\tilde{\mathcal{S}}) =\left(\frac{p^i}{P}-\frac{W}{z^{i} P}\right) \frac{Y}{\nu^i}\frac{1}{1+\psi}\left[\left(\frac{p^i}{\Lambda \nu^{i}P}\right)^{\frac{\omega \left(1+\psi\right)}{1-\omega}}+\psi\right],
\end{equation}
where $\tilde{\mathcal{S}} \equiv (P,W,Y,\Lambda)$ collects all the relevant aggregate state variables.

Because money supply $S=P Y$ grows deterministically, nominal prices will increase over time. To ensure that the state variables remain stationary, we normalize all nominal variables by $S$. As such, we can rewrite the firm's profit function
\begin{equation} \label{eq:firm_profit_stationary}
    \pi\left(\dfrac{p^i}{S},z^i,\nu^i,\mathcal{S}\right) = \left(\dfrac{\dfrac{p^i}{S}}{\dfrac{P}{S}}-\dfrac{\dfrac{W}{S}}{z^{i}\left(\dfrac{P}{S}\right)}\right)\dfrac{Y}{\nu^i}\dfrac{1}{1+\psi}\left[\left(\dfrac{\dfrac{p^{i}}{S}}{\Lambda\nu^{i}\left(\dfrac{P}{S}\right)}\right)^{\frac{\omega\left(1+\psi\right)}{1-\omega}}+\psi\right] 
\end{equation}
where the set of aggregate state variables is now 
\[\mathcal{S} \equiv \left(\frac{P}{S},\frac{W}{S},Y,\Lambda \right ).\]

Unlike the case of CES demand, it is possible for firms to find it unprofitable to supply positive quantities. This is attributed to the fact that quantity demanded under Kimball reaches zero at a finite price, typically called the ``choke price''. To handle this, we introduce an additional state variable $D$ that takes on value of 1 when the firm is not producing (i.e. it is dormant) and 0 when the firm is active. 

Firms choose whether or not to change their prices by solving the problem
\begin{equation}\label{eq:V}
\begin{split}
V\Bigl(
  D_i,\,
  \tfrac{p_{-1}^{i}}{S},\,
  z^i,\,
  \nu^i,\,
  \mathcal{S}
\Bigr)
&=
\begin{cases}
\max\Bigl\{
  V_{N}\bigl(\tfrac{p_{-1}^{i}}{S}, z^i, \nu^i, \mathcal{S}\bigr),\\
\qquad\qquad
  V_{A}\bigl(z^i, \nu^i, \mathcal{S}\bigr),
  V_{D}\bigl(z^i, \nu^i, \mathcal{S}\bigr)
\Bigr\}
  & \text{if } D = 0,\\[0.4em]
\max\Bigl\{
  V_{A}\bigl(z^i, \nu^i, \mathcal{S}\bigr),
  V_{D}\bigl(z^i, \nu^i, \mathcal{S}\bigr)
\Bigr\}
  & \text{if } D = 1.
\end{cases}
\end{split}
\end{equation}
where $V_N(.)$ and $V_A(.)$ are the values for the firm not adjusting and adjusting their prices, respectively; and $V_D(.)$ is the value of dormancy. 

The value of not adjusting prices is
\begin{equation} \label{eq:VN}
    V_{N}\left(\frac{p_{-1}^{i}}{S},z^i,\nu^i,\mathcal{S}\right)=\pi\left(\frac{p_{-1}^{i}}{S},z^i,\nu^i,\mathcal{S}\right)+\mathbb{E}\left[\Xi V\left(0,\frac{p^{i}}{S'},\nu^{i'},z^{i'},\mathcal{S}'\right)\right],
\end{equation}
which is equal to the flow profit evaluated at last period's price plus a continuation value. However, if the quantity the firm needs to produce at price $p_{-1}$ is negative, we set $V_N(.) = -\infty$. 

If the firm chooses to adjust its price, it pays the fixed price adjustment cost and chooses $p_t^i$ to maximize the sum of current flow profit and the present discounted value of future profits given by 
\begin{equation}\label{eq:VA}
\begin{split}
V_{A}\bigl(z^i,\nu^i,\mathcal{S}\bigr)
=
-f\,\frac{W/S}{P/S} 
\quad
+ \underset{p^{i}/S}{\max}\Biggl\{
  \pi\left(\frac{p^{i}}{S}, z^i, \nu^i, \mathcal{S}\right) 
  + \mathbb{E}\left[
      \Xi\,
      V\left(0,
        \frac{p^{i}}{S'},
        z^{i'},
        \nu^{i'},
        \mathcal{S}'
      \right)
    \right]
\Biggr\}.
\end{split}
\end{equation}
where maximization problem respects the constraint $y^i \geq 0$ and prices that violate this constraint are not considered as options. If $y^i < 0$ for any price the firm considers for a particular combination of states, then we set $V_A(.) = -\infty$. 

If the firm does not produce, it earns zero profit in the current period and the value is simply the future continuation value
\begin{equation} \label{eq:VD}
    V_D\left(z^i,\nu^i,\mathcal{S}\right) =   \mathbb{E}\left[\Xi V\left(1,\frac{p^{i}}{S'},\nu^{i'},z^{i'},\mathcal{S}'\right)\right].   
\end{equation}
Given the way we set the problem up, a firm may choose to go dormant even if its other options do not involve negative quantities but they yield sufficiently negative current profits. Note that the value functions \eqref{eq:V}-\eqref{eq:VD} imply that firms do not need to pay the adjustment cost to stop producing, but it is required for the firm to become active again.

\subsection{Stationary Equilibrium Definition} \label{app:eqm_definition}

\begin{definition}
The stationary equilibrium consists of 
\begin{enumerate}
    \item intermediate-good producer value functions $V(\cdot)$, $V_A(\cdot)$, $V_N(\cdot)$, and $V_D(\cdot)$, pricing rule $p_{i}^*/S(\cdot)$, and labor demand $l_i^{*}(\cdot)$
    \item demand function for each variety $y^i(\cdot)$
    \item household decisions $C$, $\Xi$, $h$
    \item aggregate prices $W/P$, $P/S$, $\lambda$, aggregate output $Y$, and aggregate profits $\Pi$
    \item distribution of intermediate-good producers $\Phi(z^i,\nu^i,\frac{p^i}{S})$
\end{enumerate}
such that
\begin{enumerate}
    \item Given $\Pi$ and $W/P$, households optimally choose $C$, $\Xi$ and $h$ satisfying 
    \begin{eqnarray}
        PC &=& Wh + \Pi \\
        \frac{W}{P} &=& \chi C \\
        \Xi &=& \beta
    \end{eqnarray}

    \item Given $P/S,W/P,Y,\Lambda$, intermediate-good producers' solves \eqref{eq:V}, \eqref{eq:VA}, \eqref{eq:VN}, \eqref{eq:VD}, yielding the pricing decision rule $p^{*}_i/S(\cdot)$, supply function $y^*_i(\cdot)$, and labor demand function $l^{*}_i(\cdot)$. The sum of profits net of adjustment cost yields aggregate dividend $\Pi$.

    \item The final-good producer solves \eqref{eq:final_producer_problem}, yielding intermediate variety demands $y^{i*}(\cdot)$ as well as the Lagrangian multiplier $\Lambda$ from \eqref{eq:lambda}. The zero-profit condition holds which yields the ideal price index $P$ for the final good in \eqref{eq:price_index}.

    \item Market clearing
        \begin{eqnarray}
            h & = & \int_{y_i>0}  l_i^{*}(\cdot) \Phi(\cdot) di \\
            1&=&\int_0^1 \frac{\omega}{1+\omega\psi}\Big[(1+\psi)\frac{\nu_iy_i}{Y}-\psi\Big]^{\frac{1+\omega\psi}{\omega(1+\psi)}}+1-\frac{\omega}{1+\omega\psi}di \\
            C & = & Y  \\   
            \frac{S}{P} & = & Y
        \end{eqnarray}

    \item The distribution $\Phi(\frac{p_{i,-1}}{S},D_i,z_i,\nu_i)$ is time-invariant and consistent with the optimizing decisions of the household and firms.

\end{enumerate}
\end{definition}

\subsection{Computational Strategy} \label{app:computational_strategy}

\subsubsection{Steady State} \label{app:algo_SS}

The following describes the algorithm for solving for the stationary equilibrium of the model.

\begin{enumerate}
    \item Construct discretized grids for idiosyncratic productivity $z_i$, demand $\nu_i$, and price $p_i/S$. For $z_i$ and $\nu_i$, use the Rouwenhorst method which approximates highly persistent AR(1) processes well. For $p_i/S$, use an equi-spaced grid. To accommodate trend inflation, the grid spacing is chosen such that movements caused by trend inflation exactly coincide with points on the grid, thereby avoiding interpolation across periods. Specifically, for a trend inflation rate of $\mu$, the spacing between consecutive grid points is set to $\dfrac{1}{s}\times \mu$ with $s = \{1,2,3,...\}$ being the step factor.
    \item Initialize guesses for equilibrium aggregate prices $\left(\hat{\frac{P}{S}},\hat{\Lambda}\right)$.
    \item Given $(\hat{\frac{P}{S}},\hat{\Lambda})$, solve the intermediate-good producer's optimization problem in \eqref{eq:V}-\eqref{eq:VD} using value function iteration, yielding value functions $\{V_A,V_N,V_D\}$, optimal pricing decision rule $p^*_i/S\left(p_{-1}/S,z,\nu;\hat{P},\hat{\Lambda}\right)$, firm supply $y^*_i\left(p_{-1}/S,z,\nu;\hat{P},\hat{\Lambda}\right)$, and labor demand $l^*_i\left(p_{-1}/S,z,\nu;\hat{P},\hat{\Lambda}\right)$.
    \item Initialize a distribution $\Phi_0\left(p_{-1}/S,D,z,\nu\right)$ over the idiosyncratic states. Given the optimal decision rules $\{p^*_i/S(\cdot),y_i(\cdot)\}$ and the laws of motion for $(S,z_i,\nu_i)$, simulate the distribution forward starting from $\Phi_0$ using the histogram method until the mass of firms at each state $\left(p_{-1}/S,D,z,\nu\right)$ is stationary, yielding the stationary distribution $\Phi^*\left(p_{-1}/S,D,z,\nu ; \hat{\frac{P}{S}},\hat{\Lambda}\right)$.
    \item Compute the aggregate price indices $P/S$ and $\Lambda$ implied by $\Phi^*\left(p_{-1}/S,D,z,\nu; \hat{\frac{P}{S}},\hat{\Lambda}\right)$ using \eqref{eq:price_index} and \eqref{eq:lambda}. Compute the absolute differences between the guesses $\left(\hat{\frac{P}{S}},\hat{\Lambda}\right)$ and the implied values $\left({\frac{P}{S}},{\Lambda}\right)$. If the differences are larger than a pre-determined tolerance level, update the guesses using a convex combination of the original guesses and the implied values and repeat Steps 3 to 5 until the differences are sufficiently small.
    \item \textbf{Model simulation for computing calibration statistics:} initialize a panel of $N$ firms to $T$ periods. Sample from the stationary distribution $\Phi^*\left(p_{-1}/S,D,z,\nu; \frac{P^*}{S},\Lambda^* \right)$ for the first period of the simulated panel. For period 2 onwards, simulate the idiosyncratic productivity and demand processes $\{z^i_t,\nu^i_t\}_{t=2}^T$ using the respective discretized Markov transition matrices and simulate firm prices and outputs using the optimal decision rules $p_i^*/S(\cdot)$ and $y_i^*(\cdot)$. When computing pricing moments, exclude firms that are transitioning to and from not producing ($y^i_t=0$).
\end{enumerate}

\subsubsection{Transition -- Monetary Policy Shock} \label{app:algo_tran}

The following describes the algorithm for solving for the transition following a one-time unanticipated shock to nominal expenditure $S$.

\begin{enumerate}
    \item Specify the size of the shock to $S$ and choose the length of the transition $T$.
    \item Construct discretized grids for idiosyncratic productivity $z_i$, demand $\nu_i$, and price $p_i/S$. For $z_i$ and $\nu_i$, use the Rouwenhorst method which approximates highly persistent AR(1) processes well. For $p_i/S$, use an equi-spaced grid. To accommodate trend inflation, the grid spacing is chosen such that movements caused by trend inflation exactly coincide with points on the grid, thereby avoiding interpolation across periods. Specifically, for a trend inflation rate of $\mu$, the spacing between consecutive grid points is set to $\mu/s$ with $s = \{1,2,3,...\}$ being the step factor.
    \item Solve for the stationary equilibrium using the algorithm in \ref{app:algo_SS}.
    \item Initialize the guess for the transition paths of the aggregate price indices $\left\{ \hat{\frac{P}{S},}\hat{\Lambda} \right\}_{t=1}^T$.
    \item Given $\left\{ \hat{\frac{P}{S},}\hat{\Lambda} \right\}_{t=1}^T$, solve the firm's optimization problem \eqref{eq:V}-\eqref{eq:VD} starting from period $T$. In period $T$, use the steady state value function $V$ as the continuation value to obtain value functions $\left(V_{T-1},V^A_{T-1},V^N_{T-1},V^D_{T-1}\right)$ as well as firm decision rules $\left(\frac{p^{i*}_{T-1}}{S},y^{i*}_{T-1}\right)$. For periods $t=T-1$ through $t=1$, use $V_{t+1}$ as the continuation value and obtain value functions $\left(V_{t},V^A_{t},V^N_{t},V^D_{t}\right)$ and firm decision rules $\left(\frac{p^{i*}_{t}}{S},y^{i*}_t\right)$.
    \item Starting from the stationary distribution $\Phi^*\left(p_{-1}/S,y,z,\nu\right)$, use the decision rules $\left\{\frac{p^{i*}_{t}}{S},y^{i*}_t\right\}_{t=1}^T$ and the laws of motion for $\left(S,z_i,\nu_i\right)$ to simulate forward using the histogram method from period $t=1$ onward through $t=T$ to obtain $\left\{\Phi\left(p_{-1}/S,y,z,\nu\right)\right\}_{t=1}^T$. Apply the shock to nominal expenditure shock $S$ in the first period.
    \item From $\left\{\Phi\left(p_{-1}/S,y,z,\nu\right)\right\}_{t=1}^T$, compute the implied path of aggregate price indices $\left\{ {\frac{P}{S},}{\Lambda} \right\}_{t=1}^T$. Compute the absolute difference between the guess $\left\{ \hat{\frac{P}{S},}\hat{\Lambda} \right\}_{t=1}^T$ and the implied path $\left\{ {\frac{P}{S},}{\Lambda} \right\}_{t=1}^T$. If the differences are larger than a pre-determined tolerance level, update the guesses using a convex combination of the original guesses and the implied values and repeat Steps 4 to 6 until the differences are sufficiently small.
\end{enumerate}


\section{Calibration}  \label{app:calibration}

\subsection{Empirical Moments}

\subsubsection{Firm-Level Productivity and Demand Processes}  \label{app:FHS_moments}

Using the quinquennial Census of Manufactures between 1977 to 1997, \cite{foster2008reallocation} estimate firm-level productivity and demand for eleven product markets with minimal vertical differentiation.\footnote{Examples include bread, block ice, and ready-mix concrete.} With data on sales, quantity sold, and input usage, they estimate the production function of firms assuming Cobb-Douglas technology and recover firm-level physical TFP (TFPQ) as the residual in the regression
\begin{equation}
    {{TFPQ}_{it}} = \ln q_{it} - \alpha_l \ln l_{it} - \alpha_k \ln k_{it} - \alpha_m \ln m_{it} - \alpha_e \ln e_{it}, \label{eq:tfpq}
\end{equation}
where ${TFPQ}_{it}$ is the firm-level physical TFP of firm $i$ at time $t$, $q_{it}$ is the quantity produced by the firm, $l_{it}$ is the labor input, $k_{it}$ is the capital input, $m_{it}$ represents intermediate inputs used in production, and $e_{it}$ is the energy used by the firm. \cite{foster2008reallocation} also estimate revenue-based TFP, which is derived similarly but replaces quantity produced with the firm's revenue, 
\begin{equation}
    {{TFPR}_{it}} = \ln p_{it}q_{it} - \alpha_l \ln l_{it} - \alpha_k \ln k_{it} - \alpha_m \ln m_{it} - \alpha_e \ln e_{it}.
\end{equation}
To recover firm-level idiosyncratic demand, \cite{foster2008reallocation} estimate the demand function 
\begin{equation}
    \ln q_{it} = \alpha_0 + \alpha_1 \widehat{\ln p_{it}} + \sum_t \alpha_t \text{YEAR}_t + \alpha_2 \ln(\text{INCOME})_{mt} + {\nu^{i}_{t}}
\end{equation}
using an instrumental variable regression. Here, the log of price, $\ln p_{it}$, is instrumented by the TFPQ estimate from \eqref{eq:tfpq}, which acts as a supply shifter. The regression includes time fixed effects and the average income in a plant's local market, $m$, defined using the Bureau of Economic Analysis' Economic Areas. The residual from this equation is then interpreted as a pure idiosyncratic demand shifter for the firm.

For the eleven products analyzed, \cite{foster2008reallocation} report average five-year autocorrelations of 0.31 for idiosyncratic TFPQ and 0.62 for demand. The cross-sectional dispersion of TFPQ and demand are 0.26 and 1.16, respectively. This indicates that demand shocks are more persistent and more dispersed across firms. Additionally, they report a correlation of --0.54 between firm-level prices and TFPQ.

\subsubsection{Pricing Moments} \label{app:pricing_moments}

For moments related to micro-level pricing behavior, we reference \cite{vavra2014inflation} who reports pricing moments using CPI micro-data from the Bureau of Labor Statistics spanning the period from 1988 through 2012.\footnote{The same dataset is widely used in the literature, see \cite{bils2004some} and \cite{nakamura2008five}.} Price data are at the product-outlet level and temporary sales are discarded from the analysis. In his sample, \cite{vavra2014inflation} reports a monthly frequency of a regular price change to be 11\%, of which 65\% are upward adjustments. The average size of a price change excluding non-adjustments is 7.7\%, and the standard deviation of price changes is 0.075.

\subsubsection{Markup and Pass-Through of Cost Shocks to Prices}  \label{app:markup_passthrough}

Following the methods of \cite{de2020rise}, we estimate the markup distribution of U.S. public firms using Standard and Poor's Compustat data. To be in line with the time period in \cite{foster2008reallocation}, we restrict the analysis to data between 1980 and 2000. We follow the production approach and compute firm-level markups as the ratio of sales to cost of goods sold, multiplied by the output elasticity of variable inputs estimated at the two-digit NAICS level.\footnote{Following the literature, we exclude the following two-digit industries: utilities, finance and insurance, real estate and rental and leasing, as well as public administration.} In our sample, the average markup is 56\% and the median markup is 33\%.

A major theoretical implication of a Kimball demand system is the incompleteness of cost pass-through to prices. One of the ways of capturing empirically the magnitude of cost pass-through can be found in the international finance literature. This literature looks at the pass-through of exchange rate shocks to importer prices, with the understanding that the exchange rate movements are exogenous from the viewpoint of importers. The empirical evidence is overwhelmingly in support of an incomplete pass-through of costs even in the medium and long-run: \cite{campa2005exchange} estimate the long-run pass-through in the US to be 42\% whereas \cite{bergin2009pass} report 24\%, \cite{gopinath2010frequency} find it to be between 20\% to 40\%, and \cite{gopinath2010currency} find an aggregate pass-through of 30\%. Estimation of cost pass-through is more challenging in a purely domestic setting, due to the scarcity of appropriate data and well-identified shocks. Using indirect estimates of marginal costs, \cite{de2016prices} report cost pass-through between 31\% to 41\% among manufacturing firms in India. \cite{carlsson2022dispersion} estimate that between 21\% to 33\% of innovations to firm productivity are passed through to prices using data on Swedish manufacturing firms. Using Belgian manufacturing data in structural model with Calvo pricing friction and strategic complementarity in pricing, \cite{gagliardone2025anatomy} estimate a 43\% cost pass-through conditional on adjustment and a 57\% pass-through of competitors' price changes. Using the same underlying data, \cite{amiti2019international} find a 60\% cost pass-through -- in the higher end among estimates in the literature -- and a 40\% pass-through of competitors' price changes in a static setup. Recent studies using merged data on both costs and prices recover cost pass-through estimates that are similar to the international macro evidence. Using Chilean supermarket-supplier merged data, \cite{aruoba2025pricing} find that 29\% of a supplier price change is passed onto the retail price conditional on a price change at the supermarket level. Overall, the evidence from both the open- and closed-economy literature points to incomplete cost pass-through to prices in the range of 20\% to 50\%.

\subsection{Calibration Details} \label{app:calibration_procedure}

The model-based moments we need for calibration are computed via simulation. In particular, we solve and simulate the model for a large number of Sobol quasi-random parameter draws. For each set of parameters, we simulate 20,000 firms for 700 periods and drop the first 300 periods before computing any statistics. We then find the parameters that deliver the closest fit, as measured by Euclidean distance, between the empirical and model-implied moments.

Computing moments that are monthly is straightforward. In order to compute moments that have their data counterpart in \cite{foster2008reallocation}, we aggregate the simulated data to the corresponding frequency and replicate their methodology. In particular, we aggregate the simulated monthly data into annual frequency by taking simple sums of revenue, sales, and employment. We then construct a panel dataset with the same time structure as \cite{foster2008reallocation}, namely five waves of annual observations that are five years apart. Because labor is the only input and production technology is constant returns to scale in the model, we recover firm-level TFPQ as,
\begin{eqnarray}
      {{TFPQ}_{it}} &=& \ln q_{it} - \ln l_{it}.  
\end{eqnarray}
This is equivalent to mapping our unique inputs to their basket of inputs. We estimate the demand function using the same IV specification as \cite{foster2008reallocation}, 
\begin{equation}
\label{eq:IV_reg}
    \ln q_{it} = \alpha_0 + \alpha_1 \widehat{\ln p_{it}}+ \text{Time FE} + \eta_{it},
\end{equation}
where $\ln p_{it}$ is instrumented by ${TFPQ}_{it}$, and recover firm-level demand shifters as the residuals, $\eta_{it}$. 

After recovering model-simulated firm productivity $TFPQ$ and demand $\eta$, we obtain five-yearly measures that are direct counterparts of those computed by \cite{foster2008reallocation}. In particular, the five-yearly autocorrelation is obtained by regressing firm TFP and demand on their lagged values. To be clear, in computing the model-implied moments we treat the model-generated data exactly the same way they treat actual data.

\subsection{Identification of Model Parameters}  \label{app:identification}

In order to demonstrate that the calibration targets are indeed  informative for the respective parameters we borrow an exercise from \cite{daruich2022macroeconomic}, which mimics the first stage of the multistart global optimization proposed by \cite{Guvenen_tiktak}.

The main idea is to generate variation in the parameter space and investigate how the implied calibration targets are impacted -- essentially taking a partial derivative. To do so, we first draw 2,000 parameter vectors from uniform Sobol points given a hypercube of the parameter space, which generates a quasi-random set of candidate parameter vectors.\footnote{A uniform Sobol sequence \citep{sobol1967distribution} is a sequence of points that spans the $n$-dimensional hypercube in an even and quasi-random manner. For the purpose of the exercise, using quasi-random Sobol numbers are more efficient than drawing random numbers because Sobol numbers are designed to sample the space of possibilities evenly given the total number of draws, whereas a truly random sample is subject to sampling noise.} Then, for each parameter vector, we solve and simulate the model to compute the relevant model-implied moments. This allows us to see how each of the seven parameters influences each of the seven calibration targets. 

Figure \ref{fig:identification} plots the values of three key model-implied target moments against the values of the parameter it is assigned to. In particular, we group the values of each parameter in deciles, which we plot on the horizontal axis. Then, for each decile, we show the median value of the associated moment in red circled dots and the 25$^{th}$ and 75$^{th}$ percentiles in blue down-pointing triangles and green up-pointing triangles, respectively. The slope of the scatter plot is informative about the importance of that parameter, whereas the vertical dispersion reveals the influences of all other parameters on a particular moment. The horizontal line shows the value for the data moment we use for calibration.

\begin{figure} 
\caption{Identification of Internally-Calibrated Parameters}

    \centering
    \subfloat[Price Adjustment Freq. vs. $f$]{\includegraphics[scale=0.26]{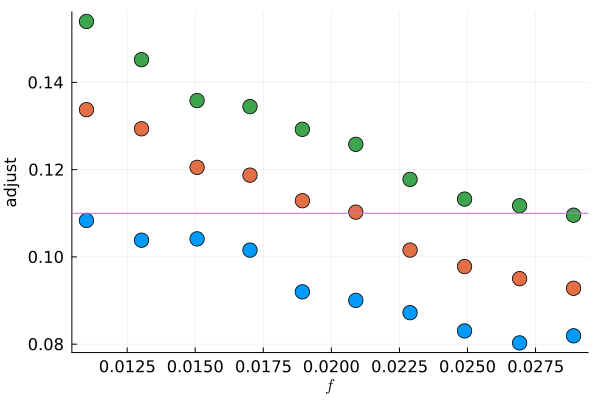}}    
    \subfloat[Corr$(TFPQ,P)$ vs. $\psi$]{\includegraphics[scale=0.26]{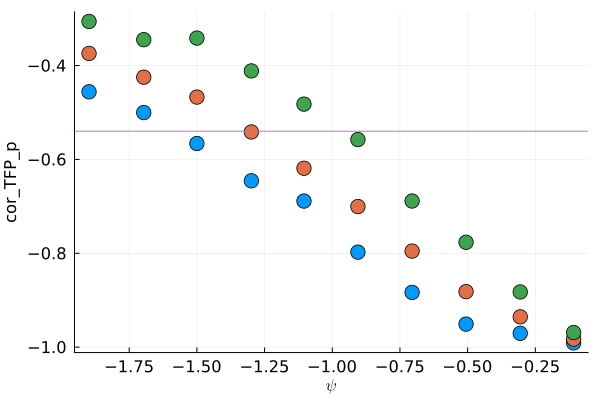}} 
    \subfloat[IV coefficient vs. $\omega$]{\includegraphics[scale=0.26]{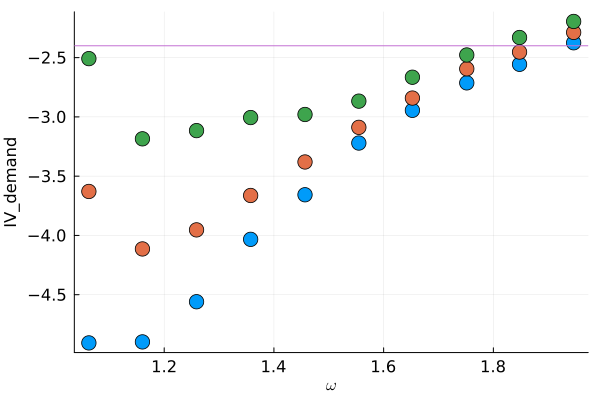}} 
\label{fig:identification}

\justify
\vspace*{-0.1in}
\footnotesize{Note: For each decile of a given parameter plotted on the horizontal axis, the red dot shows the median of the moment that is assigned to the parameter. The blue down-pointing triangles and green up-pointing triangles show the 25$^{th}$ and 75$^{th}$ percentiles respectively. The horizontal line shows the value for the data moment we use for calibration.}
\end{figure}

The frequency of price adjustment exhibits a strong negative correlation with the menu cost $f$. Meanwhile, other parameters also play a role as is evident in the vertical dispersion. For example, for a fixed value of $f$, larger idiosyncratic shocks generate more frequent price changes. Consistent with our reasoning, we recover a strong negative relationship between Corr(TFPQ,P) and $\psi$. When $\psi \approx 0$ (CES case), this correlation is $-1$ and as $\psi$ falls below 0 the correlation remains negative but weakens. 

To see the link between the IV coefficient and $\omega$, consider 
\begin{equation} \label{appeq:IV}
\log y_i
=
constant
+
\varpi \log \frac{p_i}{P}
-
(1+\varpi)\log \nu_i
+
\underbrace{\log\!\left[
1+\psi \Lambda^\varpi \nu_i^\varpi \left(\frac{P}{p_i}\right)^\varpi
\right]}_{\text{nonlinear correction}}
\end{equation}
which follows from the demand for variety $i$ given in \eqref{eq:kimball_demand} and 
\[
\varpi = \frac{\omega(1+\psi)}{1-\omega}.
\] 

For the sake of this demonstration, we ignore the nonlinear terms in \eqref{appeq:IV}. The problem with OLS estimation of an equation such as \eqref{appeq:IV} is well understood. Since $p_i/P$ and $\nu_i$ are correlated when $\psi < 0$, an OLS regression on $\log y_i$ on $\log p_i/P$ would yield a biased coefficient and will not recover the elasticity $\varpi$. The IV coefficient, however, will be unbiased and would equal to $\varpi$. This shows that given $\psi$ and ignoring nonlinearities, there is a clear link between $\omega$ and the IV coefficient.

\begin{figure}[t!]
\centering
\caption{Identification of Internally-Calibrated Parameters}
\subfloat[Five-yearly AR of TFPQ vs. $\rho_z$]{\includegraphics[scale=0.35]{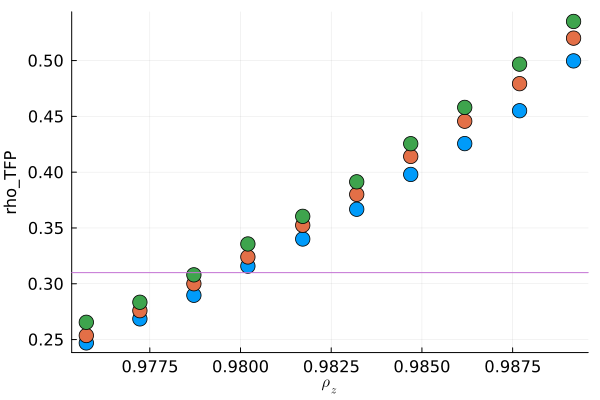}}
    \subfloat[Cross-sectional SD of TFPQ vs. $\sigma_z$]{\includegraphics[scale=0.35]{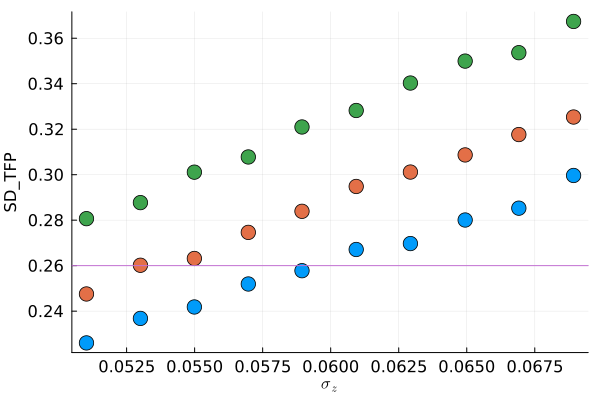}} \\
    \subfloat[Five-yearly AR of demand vs. $\rho_{\nu}$]{\includegraphics[scale=0.35]{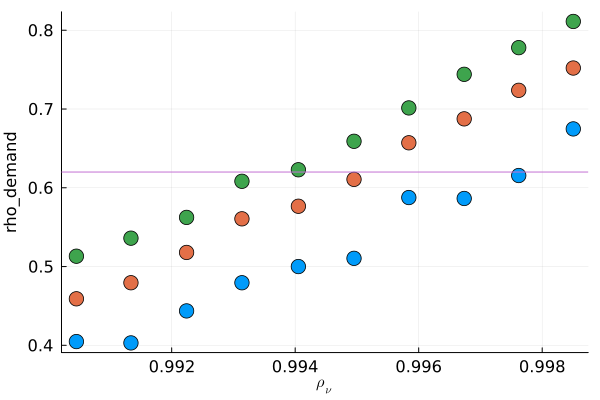}}
    \subfloat[Cross-sectional SD of demand vs. $\sigma_{\nu}$]{\includegraphics[scale=0.35]{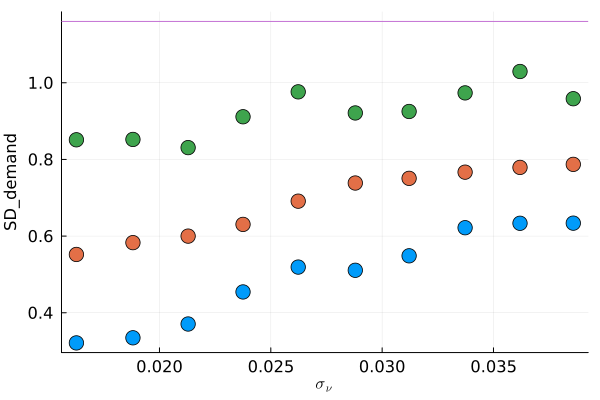}} \\
\label{fig:identification_app}
\justify
\vspace*{-0.1in}
\footnotesize{Note: For each decile of a given parameter plotted on the horizontal axis, the red dot shows the median of the moment that is assigned to the parameter. The blue down-pointing triangles and green up-pointing triangles show the 25$^{th}$ and 75$^{th}$ percentiles respectively. The horizontal line shows the value for the data moment we use for calibration.}
\end{figure}

Figure \ref{fig:identification_app} exhibits the link between the parameters governing the idiosyncratic productivity and demand processes and the corresponding empirical moments. All four pairs display a positive correlation, some stronger than others. $\rho_z$ is very strongly correlated with the five-year autocorrelation of firm productivity. $\sigma_z$, on the other hand, displays a strong correlation but other parameters, in particular $\rho_z$ has an influence as well, as can be seen in the vertical variation in the scatter plots.  Because the value of $\rho_{\nu}$ is generally very close to one, the resulting cross-sectional dispersion of demand is very sensitive to the value of $\rho_{\nu}$ in addition to $\sigma_{\nu}$. The strength of the link between $\rho_{\nu}$ and the  five-year autocorrelation of demand is weaker because demand is identified using an IV which is influenced by other parameters. Finally for the same reasons the link between $\sigma_{\nu}$ and the cross-sectional standard deviation of demand is noticably noisier, though still positive.  The key takeaway from this exercise is that the links between the parameters and moments are quite tight. 



\subsection{Details of Table~\ref{tab:compare_calib_elas}} \label{app:compare_calib_elas_source}

The section lists the source of demand elasticity and super-elasticity of the studies listed in Table~\ref{tab:compare_calib_elas}.

\begin{itemize}

    \item \cite{kimball1995quantitative} sets the demand elasticity to 11, corresponding to a markup of 1\%. The curvature of demand is chosen such that a one percent increase in market share cause a fall in the elasticity of demand to decrease from 11 to 8, which implies a super-elasticity of $31.8$.

    \item \cite{chari2000sticky} sets the demand elasticity to 10 and consider the same degree of curvature as \cite{kimball1995quantitative}. The implied super-elasticity is $35.7$.

    \item The estimated model of \cite{smets2007shocks} has a posterior mode of $1.61$ for the parameter governing the gross markup -- implying an elasticity of about $2.7$ -- and assumes a super-elasticity (Kimball curvature) of $10$.
    
    \item \cite{klenow2016real} uses a Kimball demand with the functional form
    \begin{equation}
    \Upsilon(x) = 1 + (\bar{\theta} - 1)
    \exp\!\left(\frac{1}{\varepsilon}\right)
    \varepsilon^{(\bar{\theta}/\varepsilon)-1}
    \left(
      \Gamma\!\left(\frac{\bar{\theta}}{\varepsilon}, \frac{1}{\varepsilon}\right)
      -
      \Gamma\!\left(\frac{\bar{\theta}}{\varepsilon}, \frac{x^{\varepsilon/\bar{\theta}}}{\varepsilon}\right)
    \right),
    \end{equation}
    where $\bar{\theta}$ and $\epsilon$ are the demand elasticity and super-elasticity around a symmetric steady state. They assume $\bar{\theta}=5$, implying a markup of 25\%, and $\epsilon=10$ following \cite{smets2007shocks}.

    \item \cite{bergin2000staggered} use a translog demand system characterized by the unit expenditure function 
    \[\ln P = \sum_{i=1}^N \alpha_i \ln P_i + \frac{1}{2} \sum_{i=1}^{j}\sum_{j=1}^N \gamma_{ij} \ln P_i \ln P_j.\]
    The elasticity of demand is given by $1-\frac{\gamma_{ii}}{s_{i}}$ where $s_i$ denotes the market share of good $i$. The super-elasticity is given by $\frac{\gamma_{ii}^2}{s_i^2}\frac{1}{1-\gamma_{ii}/s_i}$. Under a symmetric equilibrium where $\gamma_{ii}=-\gamma/N$ and $s_i=1/N$, demand elasticity $1+\gamma$ and super-elasticity is $\frac{\gamma^2}{1+\gamma}$. The authors chose $\gamma=2$ to match a steady markup of 1.5, which in turns imply a demand elasticity of 3 and super-elasticity of $1.33$ around a symmetric equilibrium.

    \item \cite{eichenbaum2007estimating} use $\omega=11$ to generate a markup of 10\% and assume a super-elasticity of 10.

    \item \cite{gopinath2010frequency} use the same Kimball functional form as \cite{klenow2008state}. The quantitative exercises use a demand elasticity of 5 from the trade literature and a super-elasticity of 4 to match the middle of their long-run exchange rate pass-through estimates.

    \item \cite{chari2000sticky} choose a demand elasticity of 10 to target a markup of 11\% and assume 

    \item \cite{beck2020price} use European homescan data to estimate a nested multi-nomial logit model of consumer choice. The estimated model can then be used to derive the implied demand elasticity and super-elasticity without imposing a specific functional form. Across product categories, they report a median price elasticity of demand of 3.2 and super-elasticity of 1.9.
    
    \item In the stylized model presented in Section 3 of \cite{harding2022resolving}, $\omega$ is calibrated to 1.1 to deliver a low markup of 10\% and $\psi=-12.2$ is chosen to match a Phillips curve slope of 0.0012. The full model in Section 5, adapt the \cite{smets2007shocks} model but re‑estimate the pricing block on 1965Q1–2007Q4 US data using Bayesian estimation, changing only the priors and treatment of price‑setting parameters. Specifically, they calibrate the Calvo parameter to $0.667 $, and estimate the markup parameter (corresponding to $\omega$ under our specification) and Kimball curvature parameter (which equals the super-elasticity around a symmetric equilibrium) under the respective priors of $ \mathcal{N}(1.2,0.05) $ and $ \mathcal{N}(75,25) $ respectively. The posterior mode of the markup parameter is 1.34 and that of the curvature parameter is 64.5. These imply an elasticity of 3.9 and super-elasticity of 64.5. The exact same estimation procedure is used in \cite{harding2023understanding}.
    
    \item \cite{aruoba2025pricing} calibrate a menu-cost model with Kimball demand to Chilean supermarket data. In their calibration, $\omega=1.32$ and $\psi=-1.68$ are calibrated to match the average markup and cost pass-through. These parameter values imply an elasticity of 4.1 and a super-elasticity of 6.9 around a symmetric steady state.

\end{itemize}

\subsection{Alternative Calibration Strategy} \label{app:calib_size}

We explore an alternative calibration where we target the size rather than the frequency of price changes as we did in the baseline calibration.

\begin{table}[t!]    
    \centering
    \caption{Baseline Calibration and Calibration Targeting Size} 
    \scalebox{0.85}{
    \begin{tabular}{ccccc} \toprule
    \textbf{Moment}                         & \textbf{Data}  & \textbf{Baseline} & \textbf{Target Size}  \\ \midrule
    Frequency of price changes              & 0.11           & \textbf{0.11}     & {0.12}                     \\
    Fraction of price increases             & 0.65           & 0.55              & 0.57                            \\
    Size of price changes                   & 0.08           & {0.06}            & \textbf{0.08}                      \\  
    5-year autocorr of $z_t^i$              & 0.31           & \textbf{0.31}     & \textbf{0.32}               \\
    Cross-sectional SD of $z_t^i$           & 0.26           & \textbf{0.25}     & \textbf{0.26}               \\
    5-year autocorr of $v_t^i$              & 0.62           & \textbf{0.58}     & \textbf{0.67}               \\
    Cross-sectional SD of $v_t^i$           & 1.16           & \textbf{1.04}     & \textbf{1.05}               \\
    IV Coefficient                  & --2.40           & \textbf{--2.31}     & \textbf{--2.31}                       \\
    Corr b/w price and TFPQ                 & --0.54          & \textbf{--0.54}    & \textbf{--0.54}             \\ \midrule
    \textbf{Parameter} & \textbf{Description}    &  &    \\ \midrule
    $\psi$   & Super-elasticity                      & \textbf{--1.10}   & \textbf{--0.87}       \\
    $\omega$ & Elasticity of Substitution            & \textbf{1.18}    & \textbf{1.34}           \\
    $\rho_z$ & Persistence of $z_t^i$                & \textbf{0.98}    & \textbf{0.98}           \\
    $\sigma_z$ & Standard deviation of $z_t^i$       & \textbf{0.06}    & \textbf{0.06}           \\
    $\rho_{\nu}$ & Persistence of $\nu_t^i$                & \textbf{0.998}   & \textbf{0.998}          \\
    $\sigma_{\nu}$ & Standard deviation of $\nu_t^i$       & \textbf{0.03}    & \textbf{0.03}           \\
    $f$      & Menu cost                             & \textbf{0.016}    & \textbf{0.021}           \\ \midrule
    \end{tabular}
    }
    \justify
    \footnotesize{Note: The top panel of this table compares the targeted moments and model-implied moments for the two model specifications, where the bolded numbers highlight moments that are targeted in the calibration. The bottom panel shows the parameter values for each calibration.}
     \label{tab:calib_size}
\end{table}

The results are presented in Table \ref{tab:calib_size}. The first two columns replicate the results in Table \ref{tab:int_calib}, where we continue to use boldface to emphasize the moments being targeted and parameters used to do so. The third column reports the results from the calibration where we target the size of non-zero price changes. This alternative calibration delivers a similar fit to the data and the calibrated parameters do not differ much from the baseline. At the same time, it is able to match both the frequency of price changes and fraction of price changes that are price increases -- both untargeted -- well.


\section{Robustness}  \label{app:robustness}

\subsection{Colombian Data}
\label{app:eslava_strategy}

This section summarizes the empirical strategy in \cite{eslava2024size}. 

The paper combines plant-level Colombian data with a structural model that jointly disciplines demand and production. The data are from the Annual Manufacturing Survey, a long panel census of manufacturing plants covering the time period 1982 to 2012 that reports, at a detailed product level, both quantities and values for outputs and material inputs, along with plant‑level information on employment, wages, capital, and other inputs.

The empirical framework nests a production function and a demand system that are estimated jointly at the sector level. On the production side, plants use capital, labor, and materials in a Cobb–Douglas technology with possibly non‑constant returns to scale, and the residual from this real‑quantity production function is interpreted as TFPQ. On the demand side, consumers have nested CES preferences over establishments and products, which implies a plant‑level inverse demand curve linking the plant’s output price and quantity to a common demand elasticity and to a plant‑specific “appeal” or quality shifter, which we interpret as idiosyncratic demand shifter in our context. The two sides are estimated together using a proxy‑based GMM procedure that exploits the panel structure and timing assumptions to identify production elasticities, returns to scale, and the demand elasticities. Unlike the methodology in \cite{foster2008reallocation}, the identification restrictions on productivity and demand shocks are that innovations to productivity be orthogonal to lagged demand in levels in the previous, and that innovations to demand be orthogonal to productivity in the previous period. As such, productivity and demand shocks are allowed to be correlated contemporaneously.

When calibrating the model to Colombian data, we do not implement the GMM approach in \cite{eslava2024size} due to computational constraints. Instead, we build model-implied moments using the true simulated productivity and demand processes.

\subsection{Model with CES Demand and Decreasing Returns to Scale} \label{app:DRS}

In this section, we consider an alternative source of micro real rigidity arising from decreasing returns to scale technology. To do this, we make two changes to our baseline model. Instead of Kimball demand, intermediate goods producers face CES demand with a demand shifter
\begin{equation}
    y_{it} = \nu_{it} Y_{it} \left(\frac{p_i}{P}\right)^{-\theta}
\end{equation}
In addition, the production technology exhibits decreasing returns to scale ($\alpha<1$).\footnote{Note that a setup with CES demand and increasing returns to scale leads to strategic substitution in pricing. Even though we do not impose $\alpha<1$, it is needed to match the model to data.}
\begin{equation}
    y = z_{it} l_{it}^\alpha
\end{equation}

Besides these two modifications, all other aspects of the model remain unchanged and the resulting structure is identical to the baseline model in \cite{burstein2007prices}.

This version of the model share many properties as our baseline model, which is unsurprising given the discussion in Section \ref{sec:simple} that both curvature in the demand function and marginal cost function can raise the concavity of the profit function, which is central to micro real rigidity. In particular, the two model specifications deliver similar implications for the pass-through of productivity and demand shocks to prices.

To see this, consider again the static price-setting problem of a firm under flexible prices. The first-order condition to the static profit-maximization problem is
\begin{equation} \label{eq:DRS_foc}
    \left(1-\theta\right)p_{it}^{-\theta}\nu_{it}Y_{t}P_{t}^{\theta-1}+\frac{\theta}{\alpha}\frac{W_{t}}{P_{t}}\left(\frac{\nu_{it}}{z_{it}}Y_{t}P_{t}^{\theta}\right)^{\frac{1}{\alpha}}p_{it}^{\frac{-\theta}{\alpha}-1}	=0
\end{equation}

Log-linearizing (\ref{eq:DRS_foc}) around a symmetric steady state yields the following expression for the optimal price:
\begin{equation}
\hat{p}^{*} = \frac{1-\alpha}{\alpha\theta-\theta-\alpha}\left(\hat{\nu}+\hat{Y}\right)+\frac{\alpha\theta-\theta}{\alpha\theta-\theta-\alpha}\hat{P}-\frac{\alpha}{\alpha\theta-\theta-\alpha}\hat{W}+\frac{1}{\alpha\theta-\theta-\alpha}\hat{z}
\end{equation}
where hatted variables denote log-deviations from the steady state.

The cost and demand pass-throughs under this specification are then given by
 \begin{eqnarray}   
    \dfrac{\partial \hat{p}^{*}_{i}}{\partial \widehat{mc}} &=& \frac{-1}{\alpha\theta-\theta-\alpha}   \\
    \dfrac{\partial \hat{p}^{*}_{i}}{\partial \hat{\nu}_{i}} &=& \frac{1-\alpha}{\alpha\theta-\theta-\alpha},
\end{eqnarray}

Notice that with constant returns to scale ($\alpha=1$), the cost pass-through to price is complete and the demand pass-through to price is zero. When the returns to scale parameter $\alpha$ falls below one so that technology exhibits decreasing returns to scale, the cost pass-through becomes smaller while demand begins to matter for optimal pricing as shown in Figure (\ref{fig:BH_passthrough}). 

\begin{figure}[H]
    \centering 
    \caption{Pass-through of Demand and Cost Shocks to Price}  \label{fig:BH_passthrough}
    \includegraphics[scale=0.4]{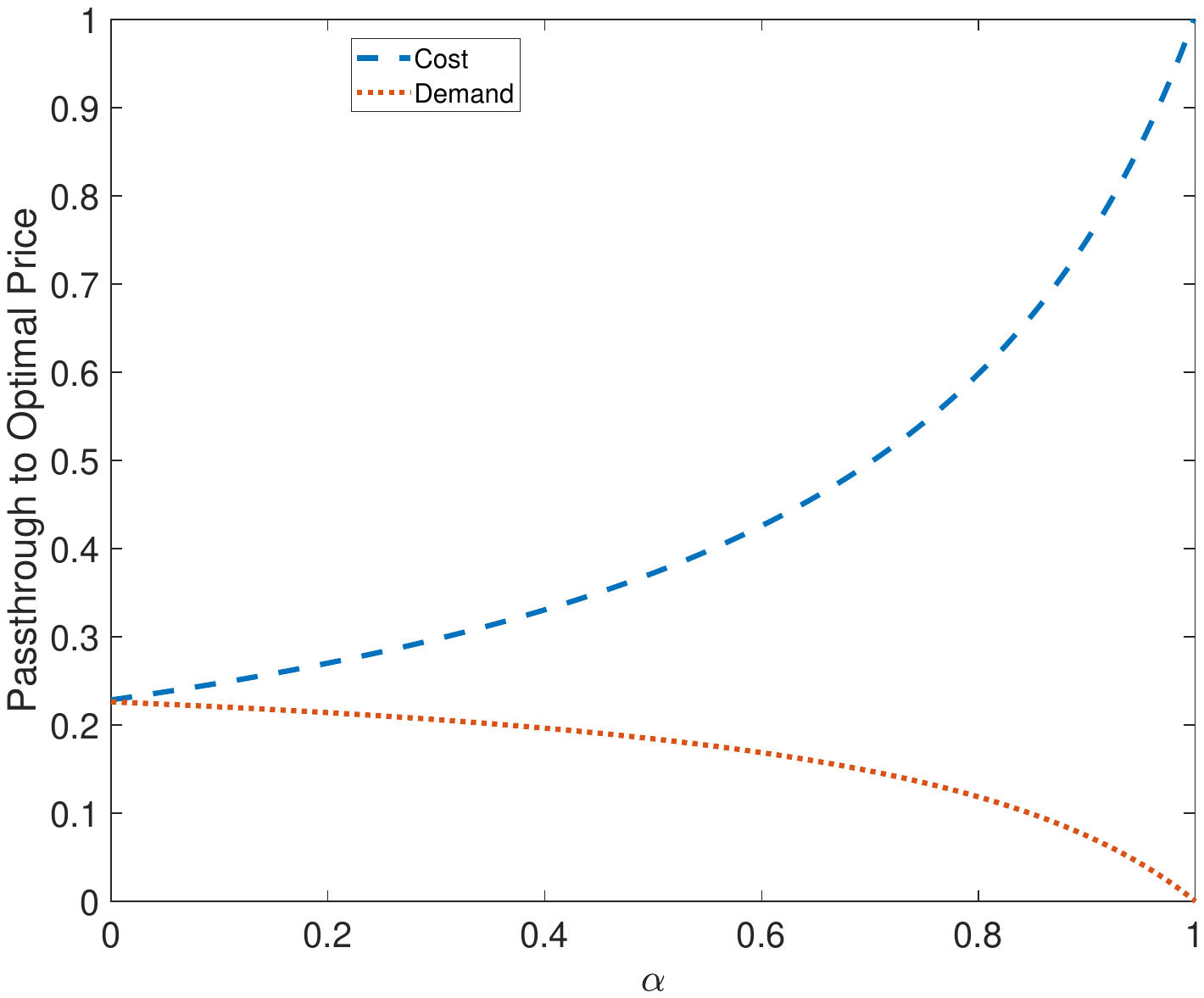}
    \justify
\footnotesize{This figure plots the cost and demand pass-through in a model with CES demand with varying levels of returns to scale $\alpha$, holding $\theta$ fixed.}
\end{figure}

Next, we explore whether cost-side micro real rigidities can also reconcile firm dynamics and pricing facts when calibrated using our strategy to the same empirical data. In doing so, we replace the two parameters pertaining to the Kimball demand system ($\omega,\psi$) with the elasticity of demand $\theta$ and return-to-scale parameter $\alpha$, while targeting the same empirical moments as in the baseline. Externally calibrated parameters are held fixed at their baseline values.

As shown in Table~\ref{tab:DRS_summary}, two variants are considered. Model A targets only the narrow firm-dynamics moments, while Model B additionally targets broader moments such as growth dispersion and the markup distribution. We first note that a model with CES demand, idiosyncratic demand shifters, and decreasing returns to scale technology can match the respective empirical moments as well as our baseline model under both calibration strategies. However, focusing on Model A -- which follows the same calibration strategy as our baseline calibration in the main text, its performance on untargeted pricing moments is worse than the baseline model using Kimball demand. In particular, it over-predicts the average size of price changes as well as the dispersion of price adjustments. The implied pass-through of supply shocks of 63\% is also higher than the baseline model calibration. In contrast, Model B -- by targeting additional moments -- delivers a much better fit. It matches the average size of price changes, fraction of upward adjustments, markup distribution closely, and implies a more plausible cost pass-through of 31\% which is within the range documented in the empirical literature.

However, the cost-side micro real rigidities framework has two shortcomings compared to our baseline model with Kimball demand and constant returns to scale. Firstly, although the average and cross-sectional dispersion of the model-implied markup distribution compare well with the data, the shape of the model-implied distribution exhibits a large hump near the frictionless desired markup which is counterfactual to the data, as shown in Figure \ref{fig:ksdensity_markup_BH}. Secondly, the calibrated returns to scale parameter of 0.59 in Model A and 0.22 in Model B are both implausibly low given that the literature typically finds constant returns to scale among manufacturing firms \citep{foster2008reallocation,eslava2024size}.

\begin{table}[t!]
\centering
\caption{Internal Calibrations: CES Demand with DRS} \label{tab:DRS_summary}
\scalebox{0.9}{
\begin{tabular}{lccc} \toprule
    {\textbf{Parameter}}& \textbf{Description} & \textbf{Model A} & \textbf{Model B} \\ \midrule
    $f$       & \multicolumn{1}{l}{Menu cost}      & 0.05 & 0.20 \\
    $\rho_z$  & \multicolumn{1}{l}{Persistence of $z^i_t$}      & 0.98 & 0.98 \\
    $\sigma_z$& \multicolumn{1}{l}{Standard deviation of $z^i_t$}      & 0.06 & 0.06 \\
    $\rho_{\nu}$  & \multicolumn{1}{l}{Persistence of $\nu^i_t$}       & 0.991 & 0.990 \\
    $\sigma_{\nu}$& \multicolumn{1}{l}{Standard deviation of $\nu^i_t$}      & 0.15 & 0.14 \\
    $\alpha$  & \multicolumn{1}{l}{Returns to scale} & 0.59 & 0.22 \\
    $\theta$  & \multicolumn{1}{l}{Elasticity of substitution}      & 2.43 & 3.84 \\ \midrule
    \multicolumn{1}{l}{\textbf{Narrow Moments}} & \textbf{Data} \\ \midrule
    SD($z_t^i$)                 & 0.26 & \textbf{0.27} & \textbf{0.31} \\
    5-year autocorr. of $z_t^i$ & 0.31 & \textbf{0.31} & \textbf{0.33} \\
    SD($\nu_t^i$)                 & 1.16 & \textbf{1.14} & \textbf{0.96} \\
    5-year autocorr. of $\nu_t^i$ & 0.62 & \textbf{0.60} & \textbf{0.58} \\
    IV coefficient              & --2.40 & \textbf{--2.45} & \textbf{--3.84} \\
    Corr(price, TFPQ)           & --0.54 & \textbf{--0.49} & \textbf{--0.37} \\
    Frequency of price changes  & 0.11 & \textbf{0.12} & \textbf{0.14} \\ \midrule
    \multicolumn{4}{l}{\textbf{Broader Moments}} \\ \midrule
    Growth dispersion           & 0.39 & 0.31 & \textbf{0.24} \\
    Fraction of price increases & 0.65 & 0.58 & \textbf{0.60} \\
    Average size of price change& 0.08 & 0.14 & \textbf{0.10} \\
    SD($\Delta p$)              & 0.08 & 0.15 & \textbf{0.10} \\
    Avg. markup                 & 1.56 & 1.75 & \textbf{1.56} \\
    SD(markup)                  & 0.72 & 0.27 & \textbf{0.68} \\ \midrule
    Pass-through of supply shocks & 20\%-50\% & 0.63 & 0.31 \\ \bottomrule
\end{tabular}
}
\justify
\footnotesize{This table reports the calibration of the model with CES demand and decreasing returns to scale technology. The top panel reports the values of the calibrated parameters. The middle panel reports the targeted moments in the data and the model. The bottom panel reports various untargeted moments in the data and their model counterparts.}
\end{table}

\begin{figure}[t!]
    \centering
    \caption{Cross-Sectional Distribution of Gross Markup: Model vs. Data}       
    \includegraphics[scale=0.4]{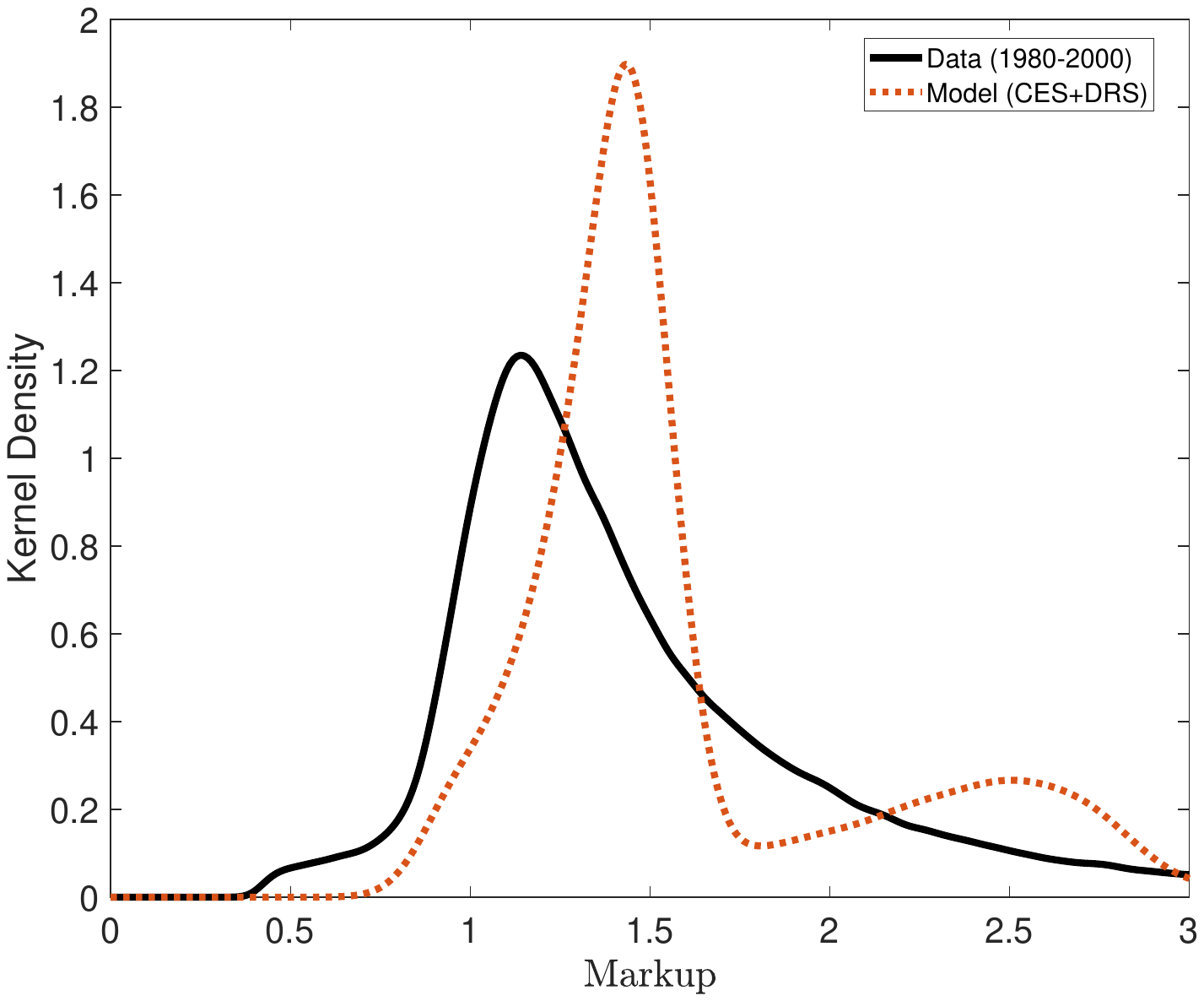}
    \label{fig:ksdensity_markup_BH}
    \justify
    \footnotesize{The figure plots the kernel density of the empirical markup distribution from publicly traded firms in the U.S. as well as the kernel density of the markup distribution in the ergodic distribution of the model with CES demand and decreasing returns to scale technology. Both kernel densities are computed using the optimal bandwidth for normal densities.}
\end{figure}

\subsection{Correlated Shocks} \label{app:correlated_shocks}

In this section, we relax the assumption of $\rho_{z\nu}=\frac{\sigma_{\nu z}}{\sigma_{z} \sigma_{\nu}}=0$ and allow stochastic innovations to firm demand and productivity to be correlated. In particular, we solve the model with two levels of correlation between demand and productivity shocks $\rho_{z\nu}=\{-0.4,0.4\}$. In doing so, we keep all other parameters fixed at their baseline values but vary the menu cost $f$ to keep the frequency of price changes identical across all specifications. The results are summarized in Table (\ref{tab:correlated}). 

\begin{table}[t!] 
\centering
\caption{Model Summary with Correlated Supply and Demand Shocks} \label{tab:correlated}
\medskip
\scalebox{0.8}{
\begin{tabular}{ccccc} \toprule
 \textbf{Moment} 	                              & \textbf{Data}  & \textbf{$\rho_{z\nu}=0$}  &  \textbf{$\rho_{z\nu}=-0.40$}  &  \textbf{$\rho_{z\nu}=0.40$} \\ \midrule
\multicolumn{1}{l}{5-year autocorr of ${z_t^i}$} &   0.31         &  \textbf{0.32}                    &  0.37                            &  0.38  \\
\multicolumn{1}{l}{Cross-sectional SD of $z_t^i$}  &   0.26         &  \textbf{0.25}                      &  0.25                           &  0.27  \\ 
\multicolumn{1}{l}{5-year autocorr of ${\nu_t^i}$} &   0.62         &  \textbf{0.62}                    &  0.62                              &  0.60  \\
\multicolumn{1}{l}{Cross-sectional SD of $\nu_t^i$}  &   1.16         &  \textbf{1.05}                    &  1.58                            &  0.92 \\
\multicolumn{1}{l}{IV Coefficient}         & --2.4           &  \textbf{--2.3}                      & --2.2                            & --1.8 \\ 
\multicolumn{1}{l}{Corr b/w price and TFPQ}        & --0.54         &  \textbf{--0.57}             & --0.69                        & --0.49  \\   \midrule
\multicolumn{1}{l}{Frequency of Price Changes}     &  0.11          &  {\bf  0.11}                   & {\bf  0.11}                  & {\bf  0.11}    \\
\multicolumn{1}{l}{Average Size of Price Changes}  &  0.08          &  0.07                          & 0.08                          & 0.04  \\   \midrule
\multicolumn{1}{l}{Average Markup}                 &  1.56          &  1.42                         &  1.44                             &  1.64 \\
\multicolumn{1}{l}{Cross-sectional SD of Markup}   &   0.72         &  0.39                          &  0.37                             &  0.64 \\ \bottomrule
\end{tabular}
}
\end{table}

\subsection{Leptokurtic Demand Shocks} \label{app:lepto}

It is well known that the distribution of price changes implied by a standard menu-cost model with Gaussian shocks exhibits negative excess kurtosis, in contrast to the positive excess kurtosis observed in the data. We therefore consider an alternative calibration with leptokurtic demand shocks, which better matches the empirical kurtosis.\footnote{Our baseline model delivers a raw kurtosis of 2.91, falling short of the U.S. estimate of 4.5, which lies in the middle of the range reported in \cite{alvarez2016real}.}

\begin{table}[t!]
    \centering
    \caption{Baseline Calibration and Calibration with Leptokurtic Demand Shocks}
    \label{tab:lepto}
    \scalebox{0.85}{
    \begin{tabular}{cccc}
    \toprule
    \textbf{Moment}                         & \textbf{Data}  & \textbf{Baseline}  & \textbf{Leptokurtic} \\ \midrule
    Frequency of price changes              & 0.11           & \textbf{0.11}      & \textbf{0.12}        \\
    Fraction of price increases             & 0.65           & 0.55               & 0.60                 \\
    Size of price changes                   & 0.08           & 0.06               & 0.08                 \\
    Raw kurtosis of price changes           & 4.50           & 2.91               & \textbf{4.48}        \\ \midrule
    5-year autocorr. of $z_t^i$             & 0.31           & \textbf{0.31}      & \textbf{0.32}        \\
    Cross-sectional SD of $z_t^i$           & 0.26           & \textbf{0.25}      & \textbf{0.24}        \\
    5-year autocorr. of $\nu_t^i$           & 0.62           & \textbf{0.58}      & \textbf{0.65}        \\
    Cross-sectional SD of $\nu_t^i$         & 1.16           & \textbf{1.04}      & \textbf{1.03}        \\
    IV coefficient                          & $-2.40$        & \textbf{$-2.31$}   & $-2.39$              \\
    Corr. between price and TFPQ            & $-0.54$        & \textbf{$-0.54$}   & \textbf{$-0.52$}     \\ \midrule
    \textbf{Parameter} & \textbf{Description} &  &  \\ \midrule
    $\psi$         & Super-elasticity                  & \textbf{$-1.10$}   & \textbf{$-0.73$} \\
    $\omega$       & Elasticity of substitution        & \textbf{1.18}      & \textbf{1.49}    \\
    $\rho_z$       & Persistence of $z_t^i$            & \textbf{0.98}      & \textbf{0.98}    \\
    $\sigma_z$     & Standard deviation of $z_t^i$     & \textbf{0.06}      & \textbf{0.06}    \\
    $\rho_{\nu}$   & Persistence of $\nu_t^i$          & \textbf{0.998}     & \textbf{0.886}   \\
    $\sigma_{\nu}$ & Standard deviation of $\nu_t^i$   & \textbf{0.03}      & \textbf{0.24}    \\
    $p_{\nu}$      & Poisson prob. for $\nu_t^i$ shock & --                 & \textbf{0.031}   \\
    $f$            & Menu cost                         & \textbf{0.016}     & \textbf{0.017}   \\ \midrule
    \end{tabular}
    }
    \justify
    \footnotesize{Note: The top panel compares the targeted moments and model-implied moments for the two specifications; boldface highlights the moments targeted in the calibration. The bottom panel reports the corresponding parameter values.}
\end{table}

The results are presented in Table \ref{tab:lepto}. The first two columns reproduce the results in Table \ref{tab:int_calib}, where we continue to use boldface to emphasize the moments targeted in the calibration and the corresponding parameter values. The last column introduces a leptokurtic shock, following \cite{midrigan2011menu}, applied to the idiosyncratic demand shock in order to match the kurtosis of non-zero price changes. Specifically, we assume that the demand shock $\nu_t^i$ follows the persistent AR(1) process in (\ref{eq:lom_zn}) with probability $p_\nu$ and remains unchanged with probability $(1-p_\nu)$ from one month to the next. We therefore add $p_\nu$ to the set of calibrated parameters and the kurtosis of non-zero price changes, equal to 4.5, to the set of targeted moments.

The calibration delivers a kurtosis of 4.48 while continuing to match the rest of the moments well. The estimate $p_\nu = 0.031$ implies that, in 96.9\% of months, a firm inherits its previous demand shock.\footnote{For comparison, \cite{vavra2014inflation} obtains $p_z = 0.13$ when applying the leptokurtic process to firm-level TFP. We instead apply it to demand rather than TFP because demand is much more persistent than TFP in \cite{foster2008reallocation}, and the infrequent-jump specification provides a convenient way to capture that persistence.} This calibration reproduces the high persistence documented in \cite{foster2008reallocation} using a low $p_\nu$, while $\rho_{\nu}$ is substantially lower and $\sigma_{\nu}$ is substantially higher. The resulting path of firm-level demand shocks is therefore flat for long periods and punctuated by occasional large jumps, which in turn generates a distribution of price changes with high kurtosis.

\subsection{Calvo Model} \label{app:calvo}

We first study a CES model with Calvo pricing and no idiosyncratic shocks. We consider the same monetary experiment as in Section~\ref{sec:nonneutrality}: a one-time, unanticipated positive level shock to nominal expenditure. This shock raises nominal expenditure growth only in the impact period; thereafter, nominal expenditure returns to its deterministic growth path. In the Calvo environment, this implies that nominal expenditure remains permanently higher by $\mu$, so that $\hat S_t=\mu$ for all $t\geq 1$.

Each period, a random fraction $\alpha$ of firms can adjust their prices freely, while the remaining fraction $1-\alpha$ keeps its previous price unchanged. For simplicity, assume that there are no aggregate risks apart from this shock and that the economy is initially in a symmetric equilibrium in which all firms set the same nominal price $\bar p$ and nominal expenditure $S=PC$ is equal to $\bar S$.

Under Calvo pricing, the aggregate price index can be written as
\begin{equation}
    P_t = \left[ \alpha X_t^{1-\theta} + \left(1-\alpha\right) P_{t-1}^{1-\theta} \right]^{\frac{1}{1-\theta}},
\end{equation}
where $X_t$ is the reset price chosen by adjusting firms.

Log-linearizing around the initial steady state, where $P_t=X_t=\bar p$, and using hatted variables to denote log deviations from steady state yields
\begin{equation}
    \hat P_t = \alpha \hat X_t + \left(1-\alpha\right)\hat P_{t-1}.
\end{equation}

Because the nominal wage is proportional to nominal expenditure and the optimal markup is constant at $\frac{\theta}{\theta-1}$ under CES demand, firms that are able to adjust choose a reset price equal to the level shift in nominal expenditure:
\begin{equation}
    \hat X_t = \mu \qquad \text{for all } t\geq 1.
\end{equation}

It follows that the aggregate price level in period 1 is
\begin{equation}
    \hat P_1 = \alpha \mu.
\end{equation}
Iterating forward, the aggregate price level in period $h$ is
\begin{equation}
    \hat P_h = \alpha \mu \sum_{i=1}^{h} \left(1-\alpha\right)^{i-1}
             = \left[1-\left(1-\alpha\right)^h\right]\mu.
\end{equation}

The response of real output is therefore
\begin{eqnarray}
    \hat C_h &=& \hat S_h - \hat P_h \\
             &=& \mu - \left[1-\left(1-\alpha\right)^h\right]\mu \\
             &=& \left(1-\alpha\right)^h \mu,
\end{eqnarray}
so the output response as a fraction of the shock is $(1-\alpha)^h$ at horizon $h$.

Given the empirical frequency of price adjustment, we set $\alpha=0.11$. This implies an initial, and peak, output response equal to 89\% of the shock and a cumulative impulse response of 8.09.\footnote{The cumulative impulse response in the Calvo model is $\sum_{t=1}^{\infty} 0.89^t = 8.09$. This also coincides with the result in \cite{alvarez2016real}, who show that the cumulative response can be expressed in terms of the kurtosis and frequency of price changes. According to their formula, the Calvo model has a cumulative impulse response of 8.09, whereas a menu-cost model \`a la \cite{golosov2007menu} has a cumulative impulse response of 1.35.}

\subsection{Model with CES Demand} \label{app:ces_models}

In this section, we present the calibration of the CES model used in Section~\ref{sec:nonneutrality}. In this calibration, we fix $\psi=0$ and $\omega=1.33$, and calibrate $(\rho_z,\sigma_z,\rho_{\nu},\sigma_{\nu})$ to match the five-year autocorrelations and cross-sectional dispersions of idiosyncratic TFPQ and demand.

\begin{table}[t!]
\centering
\caption{Internal Calibration}
\scalebox{0.9}{
\begin{tabular}{cccc}
\toprule
\textbf{Moment} & \textbf{Data} & \textbf{CES} & \textbf{Baseline} \\
\midrule
\multicolumn{1}{l}{Frequency of price changes} & 0.11 & \textbf{0.11} & \textbf{0.12} \\
\multicolumn{1}{l}{Fraction of price increases} & 0.65 & 0.58 & 0.58 \\
\multicolumn{1}{l}{Size of price changes} & 0.08 & 0.14 & 0.07 \\
\midrule
\multicolumn{1}{l}{5-year autocorr of ${z_t^i}$} & 0.31 & \textbf{0.32} & \textbf{0.32} \\
\multicolumn{1}{l}{Cross-sectional SD of $z_t^i$} & 0.26 & \textbf{0.26} & \textbf{0.25} \\
\multicolumn{1}{l}{5-year autocorr of ${\nu_t^i}$} & 0.62 & \textbf{0.62} & \textbf{0.62} \\
\multicolumn{1}{l}{Cross-sectional SD of $\nu_t^i$} & 1.16 & \textbf{1.18} & \textbf{1.05} \\
\multicolumn{1}{l}{Corr b/w TFPR and TFPQ} & 0.75 & 0.00 & \textbf{0.74} \\
\multicolumn{1}{l}{Corr b/w price and TFPQ} & --0.54 & --1.00 & \textbf{--0.57} \\
\midrule
\textbf{Parameter} & \textbf{Description} &  &  \\
\midrule
$\psi$ & \multicolumn{1}{l}{Super-elasticity} & 0 & \textbf{--1.27} \\
$\omega$ & \multicolumn{1}{l}{Elasticity} & 1.33 & \textbf{1.29} \\
$\rho_z$ & \multicolumn{1}{l}{Persistence of $z_t^i$} & \textbf{0.98} & \textbf{0.98} \\
$\sigma_z$ & \multicolumn{1}{l}{Standard deviation of $z_t^i$} & \textbf{0.05} & \textbf{0.06} \\
$\rho_{\nu}$ & \multicolumn{1}{l}{Persistence of $\nu_t^i$} & \textbf{0.992} & \textbf{0.997} \\
$\sigma_{\nu}$ & \multicolumn{1}{l}{Standard deviation of $\nu_t^i$} & \textbf{0.05} & \textbf{0.02} \\
$f$ & \multicolumn{1}{l}{Menu cost} & \textbf{0.03} & \textbf{0.03} \\
\bottomrule
\end{tabular}
}
\justify
\footnotesize{Note: The top panel of this table compares the targeted moments and model-implied moments for the two remaining model specifications (CES and Baseline). Bolded numbers highlight targeted moments. The bottom panel shows the calibrated parameter values.}
\label{tab:calib_CES_app}
\end{table}

\end{document}